%% file: paper.tex
\documentclass[11pt]{article}

\input{tex/preamble-arxiv}
\usepackage{csquotes}
\usepackage[
sorting=none, sortcites=true,]{biblatex}
\usepackage{pdfpages}
\input{tex/draft-macros}
\input{tex/math-macros}

\input{tex/problem-macros}
\input{tex/theorem-macros-arxiv}
\input{tex/appendix-toc}

\title{When Representative Samples Produce Worse Outcomes:\\Scale-up Decisions and Testing in Small-Budget RCTs}

\author{
  Hannah Li \quad Hongseok Namkoong \quad Isaac Scheinfeld \thanks{Authors listed in alphabetical order.} \\[0.3em]
  Columbia Business School
}
\date{}

\begin{document}

\maketitle
\vspace{-2em}

\begin{abstract}
  Small randomized controlled trials are often used to screen interventions before running larger follow-up studies.
  This is a critical phase of experimentation, as missing effective interventions or scaling up harmful ones can be very costly.
  A common proposal to mitigate these errors is to recruit samples that are representative of the target population, but this is often challenging in resource-constrained pilots.
  We challenge the narrative that representative samples are always superior by showing that when statistical significance testing determines whether interventions receive further study, the pilot trial composition that maximizes the downstream expected improvement in outcomes depends critically on its budget size.
  In the large-budget limit, the optimal pilot design converges to a sample that is representative of the target population.
  However, in the small-budget regime, the pilot designer maximizes expected impact by sampling only from a single homogeneous sub-population, chosen in a manner that depends on sampling costs and the designer's prior beliefs about heterogeneous treatment effects.
  Our proof of the small-budget result applies more generally when an RCT and significance test are used to decide whether to receive any non-adaptive downstream payoff, a result that may be applicable to other settings with constrained experimentation budgets.
\end{abstract}

\input{paper/01_introduction}
\input{paper/02_model}
\input{paper/03_large_budget}
\input{paper/04_small_budget}
\input{paper/05_simulation_studies}
\input{paper/06_related_work}
\input{paper/07_conclusion}
\input{paper/08_acknowledgements}

\newpage
\appendix
\printappendixcontents
\input{paper/A1_alternate_objectives}
\input{paper/A2_notation_for_proofs}
\input{paper/A3_small_budget_first_stage_objectives}
\input{paper/A4_elliptical_priors_and_small_budget_index_decompositions}
\input{paper/A5_large_budget_two_stage_impact_general_prior}
\stopappendixcontents

\printbibliography

\end{document}

%% file: tex/preamble-arxiv.tex
\usepackage[T1]{fontenc}
\usepackage[utf8]{inputenc}
\usepackage{lmodern}
\usepackage{microtype}
\usepackage[english]{babel}

\usepackage{tex/packages-arxiv}
\usepackage{tex/formatting-arxiv}

\setlist{
  listparindent=\parindent,
  parsep=0pt,
}

\makeatletter
\@ifundefined{Bbbk}{}{}
\makeatother

%% file: tex/math-macros.tex
\newcommand{\mbf}[1]{\mathbf{#1}}

\newcommand{\norm}[1]{\left\|{#1}\right\|}

\newcommand{\defeq}{:=}
\newcommand{\eqdef}{=:}

\newcommand{\indic}[1]{\mbf{1}\left\{#1\right\}}

\newcommand{\R}{\mathbb{R}}

\newcommand{\approptoinn}[2]{\mathrel{\vcenter{
      \offinterlineskip\halign{\hfil$##$\cr
#1\propto\cr\noalign{\kern0.5pt}#1\sim\cr\noalign{\kern-0.5pt}}}}}

\newcommand{\nproptoslashshift}{3.0mu}
\let\npropto\relax
\DeclareRobustCommand{\npropto}{%
  \mathrel{%
    \ooalign{%
      $\propto$\cr
      \hidewidth$\mkern\nproptoslashshift\not\phantom{\propto}$\hidewidth\cr
    }%
  }%
}

\newcommand{\E}{\mathbb{E}}

\renewcommand{\P}{\mathbb{P}}
\newcommand{\var}{{\rm Var}}
\newcommand{\cov}{{\rm Cov}}

\newcommand{\ind}{\perp\!\!\!\!\perp}

\providecommand{\sgn}{\mathop{\rm sgn}}

\providecommand{\maximize}{\mathop{\rm maximize}}
\providecommand{\subjectto}{\mathop{\rm subject\;to}}

%% file: tex/problem-macros.tex
\newcommand{\invst}{I}
\newcommand{\idx}{\pi}
\newcommand{\followup}{2}
\newcommand{\singlesuccess}{\mathrm{1S}}
\newcommand{\twosuccess}{\mathrm{2S}}
\newcommand{\singleimpact}{\mathrm{1I}}
\newcommand{\twoimpact}{\mathrm{2I}}
\newcommand{\oracle}{\mathrm{oracle}}
\newcommand{\noisefree}{\infty}
\newcommand{\ATE}{\mathrm{ATE}}

\DeclareMathOperator{\Test}{Test}

\makeatletter
\newcommand{\smallsym}[2]{#1{\mathpalette\make@small@sym{#2}}}
\newcommand{\make@small@sym}[2]{%
  \vcenter{\hbox{$\m@th\downgrade@style#1#2$}}%
}
\newcommand{\downgrade@style}[1]{%
  \ifx#1\displaystyle\scriptstyle\else
  \ifx#1\textstyle\scriptstyle\else
  \scriptscriptstyle
  \fi\fi
}
\makeatother

%% file: tex/theorem-macros-arxiv.tex
\theoremstyle{plain}
\newtheorem{theorem}{Theorem}
\newtheorem{lemma}[theorem]{Lemma}
\newtheorem{proposition}[theorem]{Proposition}
\newtheorem{corollary}[theorem]{Corollary}

\theoremstyle{definition}

\makeatletter
\providecommand{\LinkedResultRegisterAux}[3]{\linkedresult@register{#1}{#2}{#3}}
\providecommand{\LinkedResultSectionRegisterAux}[2]{}
\providecommand{\SectionLabelOfRegisterAux}[2]{\sectionlabelof@registeralias{#1}{#2}}

\newcommand{\linkedresult@register}[3]{%
  \expandafter\gdef\csname linkedresult@mode@#1\endcsname{#2}%
  \expandafter\gdef\csname linkedresult@body@#1\endcsname{#3}%
}

\newcommand{\sectionlabelof@register}[2]{%
  \expandafter\gdef\csname sectionlabelof@#1\endcsname{#2}%
}

\newcommand{\sectionlabelof@current}{}

\newcommand{\sectionlabelof@registeralias}[2]{%
  \sectionlabelof@register{#1}{#2}%
  \@ifundefined{r@#2}{}{%
    \expandafter\global\expandafter\let\csname r@sectionof:#1\expandafter\endcsname\csname r@#2\endcsname
    \@ifundefined{r@#2@cref}{}{%
      \expandafter\global\expandafter\let\csname r@sectionof:#1@cref\expandafter\endcsname\csname r@#2@cref\endcsname
    }%
  }%
}

\newcommand{\sectionlabelof@ifcurrentsectionTF}[2]{%
  \expandafter\sectionlabelof@ifcurrentsectionTF@i\@currentHref.\@nil{#1}{#2}%
}

\def\sectionlabelof@ifcurrentsectionTF@i#1.#2\@nil#3#4{%
  \def\sectionlabelof@prefix{#1}%
  \def\sectionlabelof@sectionprefix{section}%
  \def\sectionlabelof@appendixprefix{appendix}%
  \ifx\sectionlabelof@prefix\sectionlabelof@sectionprefix
  #3%
  \else
  \ifx\sectionlabelof@prefix\sectionlabelof@appendixprefix
  #3%
  \else
  #4%
  \fi
  \fi
}

\newcommand{\sectionlabelof@maybeupdatecurrent}[1]{%
  \sectionlabelof@ifcurrentsectionTF{\gdef\sectionlabelof@current{#1}}{}%
}

\newcommand{\sectionlabelof@registercurrent}[1]{%
  \ifx\sectionlabelof@current\@empty
  \else
  \sectionlabelof@registeralias{#1}{\sectionlabelof@current}%
  \protected@write\@auxout{}{%
    \string\SectionLabelOfRegisterAux{#1}{\sectionlabelof@current}%
  }%
  \fi
}

\newcommand{\linkedresult@bodytextC}[1]{%
  \NoHyper\nameCref{#1}~\labelcref{#1}\endNoHyper
}

\newcommand{\linkedresult@bodytextc}[1]{%
  \NoHyper\namecref{#1}~\labelcref{#1}\endNoHyper
}

\newcommand{\linkedresult@heading}[3]{%
  \ifstrequal{#1}{proof}
  {Proof of \hyperref[#2]{\linkedresult@bodytextC{#2}}%
  \IfNoValueF{#3}{\space(#3)}}
  {\hyperref[#2]{\linkedresult@bodytextC{#2}}, Extended}%
}

\newcommand{\linkedresult@refcap}[2]{%
  \ifstrequal{#1}{proof}
  {\linkedresult@bodytextC{#2}}
  {\linkedresult@bodytextC{#2}, Extended}%
}

\newcommand{\linkedresult@reflow}[2]{%
  \ifstrequal{#1}{proof}
  {\linkedresult@bodytextc{#2}}
  {\linkedresult@bodytextc{#2}, extended}%
}

\newcommand{\linkedresult@ref}[2]{%
  \ifcsname linkedresult@mode@#1\endcsname
  \ifstrequal{#2}{capital}
  {\expanded{%
      \noexpand\hyperref[#1]{%
        \noexpand\linkedresult@refcap
        {\csname linkedresult@mode@#1\endcsname}%
        {\csname linkedresult@body@#1\endcsname}%
      }%
  }}
  {\expanded{%
      \noexpand\hyperref[#1]{%
        \noexpand\linkedresult@reflow
        {\csname linkedresult@mode@#1\endcsname}%
        {\csname linkedresult@body@#1\endcsname}%
      }%
  }}%
  \else
  \ifstrequal{#2}{capital}{\Cref{#1}}{\cref{#1}}%
  \fi
}

\newcommand{\linkedCref}[1]{\linkedresult@ref{#1}{capital}}
\newcommand{\linkedcref}[1]{\linkedresult@ref{#1}{lower}}
\newcommand{\sectionof}[1]{sectionof:#1}

\let\sectionlabelof@origlabel\label
\renewcommand{\label}[1]{%
  \sectionlabelof@maybeupdatecurrent{#1}%
  \sectionlabelof@registercurrent{#1}%
  \sectionlabelof@origlabel{#1}%
}

\NewDocumentEnvironment{linkedresult}{m m o}
{%
  \par\addvspace{3pt}%
  \begingroup
  \edef\linkedresult@currentmode{#1}%
  \edef\linkedresult@currentbody{#2}%
  \def\@currentlabelname{}%
  \def\label##1{%
    \@bsphack
    \begingroup
    \edef\linkedresult@currentlabel{##1}%
    \expanded{%
      \noexpand\linkedresult@register
      {\linkedresult@currentlabel}%
      {\linkedresult@currentmode}%
      {\linkedresult@currentbody}%
    }%
    \expanded{\noexpand\sectionlabelof@registercurrent{\linkedresult@currentlabel}}%
    \protected@write\@auxout{}{%
      \string\LinkedResultRegisterAux{\linkedresult@currentlabel}{\linkedresult@currentmode}{\linkedresult@currentbody}%
    }%
    \protected@write\@auxout{}{%
      \string\newlabel{\linkedresult@currentlabel}{{\@currentlabel}{\thepage}{\@currentlabelname}{\@currentHref}{\@kernel@reserved@label@data}}%
    }%
    \endgroup
    \@esphack
  }%
  \noindent\textbf{\linkedresult@heading{#1}{#2}{#3}.}\ \itshape\ignorespaces
}
{%
  \par\endgroup\addvspace{3pt}%
}
\makeatother

\crefname{theorem}{Theorem}{Theorems}
\Crefname{theorem}{Theorem}{Theorems}
\crefname{lemma}{Lemma}{Lemmas}
\Crefname{lemma}{Lemma}{Lemmas}
\crefname{proposition}{Proposition}{Propositions}
\Crefname{proposition}{Proposition}{Propositions}
\crefname{corollary}{Corollary}{Corollaries}
\Crefname{corollary}{Corollary}{Corollaries}
\crefname{definition}{Definition}{Definitions}
\Crefname{definition}{Definition}{Definitions}
\crefname{problem}{Problem}{Problems}
\Crefname{problem}{Problem}{Problems}
\crefname{remark}{Remark}{Remarks}
\Crefname{remark}{Remark}{Remarks}
\crefname{ulem}{Unproven Lemma}{Unproven Lemmas}
\Crefname{ulem}{Unproven Lemma}{Unproven Lemmas}
\crefname{uprop}{Unproven Proposition}{Unproven Propositions}
\Crefname{uprop}{Unproven Proposition}{Unproven Propositions}
\crefname{uthm}{Unproven Theorem}{Unproven Theorems}
\Crefname{uthm}{Unproven Theorem}{Unproven Theorems}
\crefname{ucor}{Unproven Corollary}{Unproven Corollaries}
\Crefname{ucor}{Unproven Corollary}{Unproven Corollaries}

%% file: paper/01_introduction.tex
\section{Introduction}
\label{sec:introduction}

Randomized controlled trials (RCTs) are the cornerstone of scientific discovery across medicine, social science, and public policy~\cite{ImbensRu15}.
While the statistical foundations of experimentation are well-developed, the question of how experimental designs affect downstream decision-making and outcomes is not as well-understood.
This is especially true when treatments are evaluated in a series of experiments, as is common in both natural and social sciences~\cite{ChenWi13,KasySa21}.

While there is a rich literature on single-stage designs, experimentation is rarely a one-off process: across domains ranging from drug development (FDA trial stages) to educational interventions (IES goals or project types), a common structure is that a smaller RCT determines whether a treatment merits further investigation, followed by a larger-scale confirmatory trial that informs adoption decisions~\cite{FDACD22,DepartmentofEducationBo22}.
This pilot to scale-up to adoption paradigm is society's solution to a fundamental resource allocation problem---how to efficiently screen promising interventions while managing costs and risks.
The role of small RCTs in screening promising interventions is critical: both rejecting an effective intervention and scaling up a harmful one can be costly to society~\cite{RossellMu07}.
Correspondingly, funding and regulatory bodies have highlighted the need for methodological progress on the design of early-stage trials~\cite{Tipton21,FDAHH20,DepartmentofEducationBo22}.

In practice, due to the tight resource constraints of most pilot studies, researchers often rely on \emph{convenience samples}---participants recruited from subpopulations where studies are easiest to conduct~\cite{TiptonSp21}.
When the effectiveness of the treatment under study varies with subject characteristics, the treatment effect in such a non-representative sample may be significantly better or worse than in the target population.
A common proposal to mitigate these errors is to recruit \emph{representative samples} that better match the target population~\cite{FDAHH20,DepartmentofEducationBo22}.
This is in line with a broad push for generalizability or external validity in science~\cite{EgamiHa23, BryanTi21, Rothwell05, HenrichHe10, VoelklVo18}.

In this work, we study how the \emph{design of a pilot\footnote{We focus on the use of small RCTs as screening mechanisms for scale-up decisions, and refer to these as \emph{pilot studies}.
Pilots have many definitions across fields, and sometimes refer to even earlier stages of experimentation where the goal is ``not to test the effectiveness of the intervention but to assess the feasibility of the protocol'' \cite[Figure 2]{EldridgeLa16}.} RCT} affects the outcomes of a target population when the pilot determines whether an intervention merits further study.
We characterize how the impact of pilot design propagates through three stages: the pilot study, the follow-up study, and the adoption decision, as shown in \Cref{fig:exp_process_flow}.
By modeling the multi-stage nature of scientific discovery, we are able to analyze the role statistical significance tests play as primary continuation criteria for making decisions in the current scientific status quo~\cite{FDA19,Berry91}.
In doing so, we are neither arguing for significance tests as an optimal continuation criterion, nor that impact is a normatively correct goal of any given stakeholder.
Instead, we seek to understand how a researcher whose work is evaluated by a significance test might design studies if they care about impact, and how pushing for more representative samples without changing how studies are evaluated might affect research outcomes.

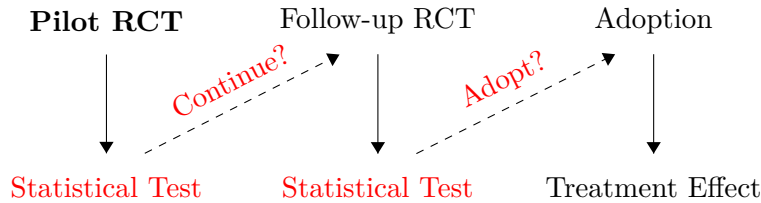
\begin{figure}[ht]
  \centering
  \vspace{0.8\baselineskip}
  \input{diagrams/experiment_process_flow_chart}
  \caption{
    \textbf{The three-stage experiment pipeline we model.}
    The intervention is first evaluated in a budget-constrained pilot RCT, where the result of a significance test determines whether it advances to a follow-up RCT.
    A second significance test determines adoption of the intervention in the target population, leading to a potential improvement in average outcomes if successful.
  }
  \label{fig:exp_process_flow}
\end{figure}

We show that when optimizing for expected downstream impact, neither \emph{representative sampling} nor \emph{convenience sampling} is optimal across all regimes.
In fact, \textbf{the optimal pilot design depends heavily on the budget of the pilot study}, as illustrated in \Cref{fig:small-budget-vs-proportional-vs-feasible}.
Under mild assumptions, we characterize the optimal pilot design for sufficiently small and asymptotically large pilot budgets under \emph{general prior distributions} of the conditional average treatment effects across subpopulations.
In the large-budget regime where pilot studies achieve high statistical precision, our model verifies the conventional wisdom that pilot designs that match the composition of the target population achieve optimal expected downstream outcomes.
However, when the budget constraint is severe, the optimal pilot design takes the opposite extreme to a representative sample: it concentrates sampling on a single subpopulation.
This provides surprising evidence that the convenience samples commonly used in early-stage experiments may have close to optimal impact, and that representative samples may be counterproductive for highly constrained pilot studies when statistical significance is important.
Because our objective yields dramatically different designs under small budgets than those found in work on external validity~\cite{Tipton13,EgamiLe23,Bouyamourn25} and minimax-regret objectives~\cite{HuZh24,OleaPr25}, we believe substantial work remains to understand resource-constrained experiment design for decision-making.

\begin{figure}[ht]
  \vspace{-0.2\baselineskip}
  \centering
  \includegraphics[width=0.73\textwidth]{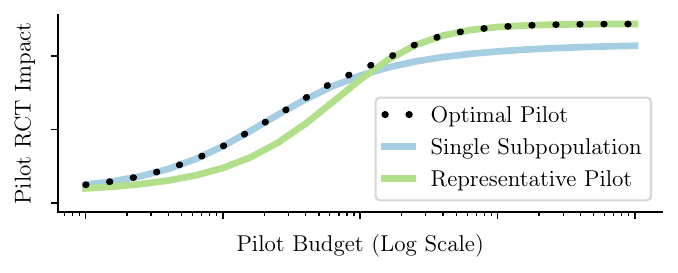}
  \vspace{-0.2\baselineskip}
  \caption{
    \textbf{Optimal design depends on pilot resources.}
    We plot the expected downstream impact achieved by the optimal pilot design, the best small-budget single-subpopulation design, and a representative sample.
    In budget-constrained settings, optimal pilots sample from a \emph{single homogeneous subpopulation}, while representative sampling is best for well-resourced trials.
  }
  \label{fig:small-budget-vs-proportional-vs-feasible}
  \vspace{-0.2\baselineskip}
\end{figure}

To gain insight into how treatment effect heterogeneity---the key motivation for representative samples---drives optimal pilot design, we specialize our small-budget optimal design to elliptical priors, which generalize multivariate normal distributions.
Here, we show that the single subpopulation that is optimal to sample under any sufficiently small budget is given by an index trading off
\begin{enumerate}[label=(\roman*)]
  \item a subpopulation's expected effect size,
  \item how predictive its effect is of the effect in the follow-up and target populations, and
  \item its recruitment cost.
\end{enumerate}

Why are homogeneous samples optimal under small budgets, while under large budgets the optimal pilot design approaches representative sampling?
\emph{For small budgets, the optimal design is driven by the signal-to-noise ratio of the test, while for large budgets it is driven by the mismatch between the pilot estimand (the sample average treatment effect) and the downstream estimand (the population average treatment effect).}

Our theoretical analysis of small-budget optimal design applies under mild assumptions to a much broader class of settings than pilot studies.
Homogeneous (single-subpopulation) sampling is small-budget optimal when significance testing is used to decide whether to receive \emph{any} downstream payoff, so long as that payoff does not depend directly on the pilot outcomes.
We show this general result via an asymptotic expansion of the objective that shows that for sufficiently small budgets, it is approximately equivalent to maximizing a signal that is linear in the subpopulation sample sizes divided by a noise scale equal to the square root of the total sample size.
This signal-to-noise ratio behaves roughly like a quasiconvex function (its positive part is quasiconvex), and therefore under a linear budget constraint it is maximized at an extreme point of the feasible design polytope.

In the other extreme where the pilot budget is very large, the pilot test can be made arbitrarily precise, and a representative sample will make almost no type I or type II errors.
Any non-representative sample, on the other hand, will make errors on some non-negligible fraction of the prior distribution of treatment effects---those for which the signs of the sample and population average treatment effects differ.

We conclude this section with a heuristic description of what kinds of subpopulations are optimal to sample under a tight budget.
Consider our main objective of maximizing the expected impact of a pilot used to determine whether to run a follow-up RCT, and compare it to two other potential objectives: maximizing the impact of adopting based on just a single RCT, or maximizing the probability of passing the pilot significance test.
Under a sufficiently tight budget, our analysis shows that sampling from a single subpopulation is optimal for all three objectives.
However, which population to sample from varies.

Consider the case when all subpopulations have the same sampling costs, and a potential follow-up RCT will use a very large and representative sample.
In this case, pilot false positives never make it to adoption---the follow-up RCT filters out any treatments with negative effects in the target population.
Since false positives are not penalized, the pilot can be more aggressive in sampling from subpopulations that have larger expected effects, even if they are less predictive of the target population effect.
However, it cannot focus only on passing the pilot, as there is still an important trade-off between passing the pilot and avoiding false negatives:  avoiding rejecting interventions with a positive effect in the target population.
Thus, pilot impact falls between single-RCT impact and pilot pass probability in terms of how it prioritizes how promising a subpopulation is versus how predictive it is of the target population effect.

The remainder of the paper is organized as follows.
\Cref{sec:Model} presents a model of the experimental pipeline, including the treatment effect structure, experimental design decisions, and significance testing protocols, and formalizes the optimization problem.
Sections \ref{sec:large-budget} and \ref{sec:small-budget} present our main theoretical results characterizing optimal pilot design in the large-budget and small-budget regimes.
\Cref{sec:case-study} presents a semi-synthetic case study calibrated to a real set of experiments in education research, demonstrating that real-world pilot studies can fall squarely within either of the budget regimes we characterize.
Finally, we situate our contributions in the literature on experimentation and external validity in \Cref{sec:related-work}, and discuss implications for future work in \Cref{sec:conclusion}.

%% file: diagrams/experiment_process_flow_chart.tex
\begin{tikzpicture}[scale=0.19]
  \tikzstyle{every node}+=[inner sep=0pt]
  \pgfmathsetmacro{\diagArrowLabelGap}{1.8}
  \newcommand{\dashedDecisionArrowLabel}[3]{%
    \path let
    \p1 = (#1),
    \p2 = (#2),
    \n1 = {atan2(\y2-\y1,\x2-\x1)}
    in node[
      rotate=\n1,
      text=red,
      fill=white,
      inner sep=1pt
    ] at ($($(#1)!0.5!(#2)$)+(\n1+90:\diagArrowLabelGap)$) {#3};%
  }

  \draw (14.2,-20.3) node {\textbf{Pilot RCT}};
  \draw (33.1,-20.3) node {Follow-up RCT};
  \draw (52.3,-20.3) node {Adoption};
  \draw (14.2,-31.9) node[text=red] {Statistical Test};
  \draw (33.1,-31.9) node[text=red] {Statistical Test};
  \draw (52.3,-31.9) node {Treatment Effect};

  \draw [black, dashed] (16.87,-29.54) -- (30.43,-22.66);
  \dashedDecisionArrowLabel{16.87,-29.54}{30.43,-22.66}{Continue?}
  \fill [black] (30.43,-22.66) -- (29.49,-22.58) -- (29.94,-23.47);
  \draw [black, dashed] (35.78,-29.56) -- (49.62,-22.64);
  \dashedDecisionArrowLabel{35.78,-29.56}{49.62,-22.64}{Adopt?}
  \fill [black] (49.62,-22.64) -- (48.68,-22.55) -- (49.12,-23.45);

  \draw [black] (14.2,-22.66) -- (14.2,-29.54);
  \fill [black] (14.2,-29.54) -- (14.7,-28.74) -- (13.7,-28.74);

  \draw [black] (33.1,-22.64) -- (33.1,-29.56);
  \fill [black] (33.1,-29.56) -- (33.6,-28.76) -- (32.6,-28.76);

  \draw [black] (52.3,-22.64) -- (52.3,-29.56);
  \fill [black] (52.3,-29.56) -- (52.8,-28.76) -- (51.8,-28.76);

\end{tikzpicture}

%% file: paper/02_model.tex
\section{Model}
\label{sec:Model}

In this section, we introduce the model of the main experimental process that we study in this paper.
Our goal is to design pilot studies under a budget constraint that maximize impact---the expected improvement in outcomes---through an experimental process consisting of a pilot RCT, a potential follow-up RCT, and finally a treatment adoption decision, as described in the introduction and shown in~\Cref{fig:exp_process_flow}.

We consider a setting where, following each RCT stage, a statistical significance test serves as the \emph{only} continuation criterion for whether the intervention moves from one stage to the next.
As described in the previous section, this reflects a common continuation criterion used in practice.

We further assume that the pilot affects the outcome of the experimental process \emph{only} through its significance test, i.e.\ whether the treatment continues to a follow-up RCT.
In other words, the potential follow-up design is fixed, e.g.\ to something like a representative sample, and the adoption decision only depends on the outcomes of the follow-up RCT.
This assumption is necessary for tractability since designing the follow-up is usually itself a non-convex optimization problem.

The remainder of this section gives a formal definition of the model and objective.

\subsection{Three-stage experiment pipeline}

We model the experimental process in terms of three stages: the pilot RCT, the follow-up RCT, and adoption, indexed by $t\in\{1,2,3\}$, respectively.

In the pilot and follow-up RCT stages, a sample of subjects is recruited into the study.
At the adoption stage, the intervention is rolled out to a predetermined target population.
Each stage has a corresponding \emph{average treatment effect},
\[
  \ATE_t\in\R,\quad t\in\{1,2,3\},
\]
and the two RCT stages each have an \emph{average treatment effect estimate} derived from the RCT sample,
\[
  \widehat{\ATE}_t\in\R,\quad t\in\{1,2\}.
\]

Following each RCT stage $t=1,2$, the intervention moves to the next stage if and only if the average treatment effect estimate clears a threshold given by a \emph{statistical significance test},
\[
  \Test_t\in\{0,1\},\quad t\in\{1,2\},
\]
where we write $\Test_t=1$ if the treatment continues to the next stage and $\Test_t=0$ otherwise.
Since a treatment is only adopted if it passes both significance tests, we can write the \emph{realized improvement in average outcomes} as
\[
  \Test_1 \cdot \Test_2 \cdot\, \ATE_3.
\]

\subsection{RCT samples and the target population}

Each stage of the experimental process has a sample or population of units that can vary in its composition.
We define discrete unit types to capture feature variability, and the composition of each stage in terms of the number of units of each type.

\paragraph{Unit types}
There are $D\ge 2$ types of units (e.g., patients, users, schools, etc.) that can be treated or observed, indexed by $d\in\{1,\ldots,D\}$.
These types can vary in their treatment effects and sampling costs.
For example, in an education study where schools are the units, a type could include schools with similar demographics, funding, and size---all of which could affect the treatment effect---and at a similar distance from the research institution running the pilot study---which could affect the cost of sampling.

\paragraph{Pilot RCT sample}
At the pilot stage $t=1$, the pilot designer chooses how many units of each type to include in the study.
Let $s_{1d}\geq 0$ denote the number of units sampled from each type $d$, and $\bm{s}_1\in\R_{\geq 0}^D$ the vector of units sampled across all $D$ types.
Let $n_1 \defeq \bm{1}^{\top}\bm{s}_1$ be the pilot sample size.
For tractability, we let $\bm{s}_1$ take on continuous values, rather than imposing an integer constraint.
We further constrain the sample to be non-empty, i.e., $\bm{s}_1 \neq \bm{0}$ and $n_1 > 0$, since the average treatment effect is undefined for an empty sample, and write this as $\bm{s}_1 \in \R_{\geq 0}^{D} \setminus \{\bm{0}\}$.

\paragraph{Follow-up RCT sample}
At the follow-up stage $t=2$, a sample of units is enrolled in the follow-up RCT.
Let $\bm{s}_2\in\R_{\geq 0}^{D} \setminus \{\bm{0}\}$ be the vector of units sampled across all $D$ types if the follow-up is run, and $n_2 \defeq \bm{1}^{\top}\bm{s}_2 > 0$ the total follow-up sample size.
We focus on designing the optimal pilot sample $\bm{s}_1$ for a fixed follow-up study sample $\bm{s}_2$, and leave their optimal joint design for future work.

\paragraph{Target population}
At the adoption stage $t=3$, a fixed set of units we call the target population is treated if the treatment is adopted.
Let $\bm{s}_3\in\R_{\geq 0}^{D} \setminus \{\bm{0}\}$ be the vector of population sizes for each type, and $n_3 \defeq \bm{1}^{\top}\bm{s}_3 > 0$ be the total target population size.

\begin{figure}[ht]
  \centering
  \input{diagrams/simplified_experiment_process_flow}
  \caption{
    \textbf{The three-stage experimental pipeline.}
    The pilot design $\bm{s}_1$ is highlighted in blue.
    The samples $\bm{s}_1$ and $\bm{s}_2$ in the first two stages determine the distribution of the average treatment effect estimates, and the population $\bm{s}_3$ in the third stage determines the change in outcomes if both significance tests pass.
  }
  \label{fig:experiment_process_flow}
\end{figure}
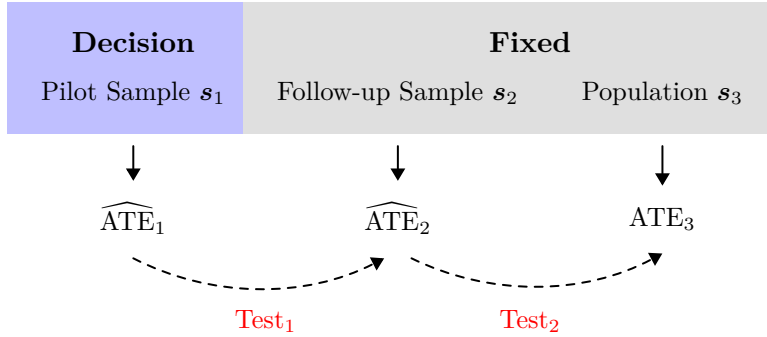

\subsection{Probabilistic model}

We define a probabilistic model of the experimental process.
The pilot designer has a prior over heterogeneous average treatment effects for each type.
Under this prior, the sample compositions of the pilot and follow-up RCTs and the composition of the target population induce a distribution over their average treatment effects.
Combined with a sampling distribution for sample average treatment effect estimates and a definition of the significance tests, this gives a full probability model for the experimental process.

We now define the heterogeneous treatment effect vector and how the sample and population compositions determine the distribution of average and sample average treatment effects.

\paragraph{Heterogeneous treatment effects}
We denote the conditional average treatment effect for type $d$ as $\tau_d$ and
let $\bm{\tau}\in\R^D$ be the vector of type average treatment effects.%
\footnote{
  Under the potential outcomes framework, these can be written as the mean
  difference in outcomes under treatment $Y_{i}(1)$ and control $Y_{i}(0)$ for a
  unit $i$ of type $d$, $\tau_d = \E[Y_{i}(1) - Y_{i}(0) \mid i \text{ is of
  type } d]$.
}
The experimenter has a prior over the conditional average treatment effects across types.
Our main results hold under general priors $\mathcal{P}$ on $\R^D$ satisfying moment and other mild regularity conditions,
\[
  \bm{\tau}\sim\mathcal{P}.
\]
The existence of fourth moments of $\bm\tau$ and a Lebesgue density positive on an open ball containing the origin are sufficient for all the main results; the appendix gives specific conditions for each.

We are able to give a more interpretable characterization of the optimal small-budget pilot design when further restricting the prior to the class of elliptical distributions~\cite{FangKo90}
\[
  \bm{\tau}\sim\mathrm{EC}_D(\bm{\mu},\Sigma,\phi),
\]
which includes and can be viewed as generalizing the multivariate normal distribution.

\paragraph{Average Treatment Effects}
Each stage $t=1,2,3$ and its associated sample $\bm{s}_t$ has an average treatment effect $\ATE_t$ given in terms of the type-level effects as
\[
  \ATE_t \defeq \frac{\bm{s}_t^{\top}\bm{\tau}}{n_t},
  \quad t\in\{1,2,3\}.
\]
These effects are not known to the pilot designer, but their prior distributions are induced by the designer's prior distribution on $\bm\tau$.

\paragraph{RCT Average Treatment Effect Estimates}
In the pilot and follow-up, the experimenter employs an RCT design paired with an appropriate estimator of the average treatment effect.
We approximate the sampling distributions of the estimators using independent normal distributions across stages, with precision proportional to the sample size $n_t$.%
\footnote{
  We set the variance to be exactly equal to the inverse sample size for notational simplicity; one could achieve this by scaling the outcome by a design- and estimator-dependent constant.
  If each type has different sampling noise, one can interpret the costs as costs-per-precision instead of costs-per-unit.
}
This approximation is motivated by central limit theorems that apply to many standard designs and estimators, such as completely- and Bernoulli-randomized designs as well as matched pairs and more general stratified designs~\cite{LiDi17,Ding23}, and is frequently used when modeling sampling noise~\cite{FrazierPo09,ChickGa22,AzevedoDe20}.
\[
  \widehat{\ATE}_t \mid \bm{\tau} \overset{\text{ind}}{\sim}
  \mathcal{N}\left(
    \ATE_t,
    \frac{1}{n_t}
  \right),
  \quad t\in\{1,2\}.
\]

\paragraph{Significance Tests as Continuation Criteria}
Following each RCT stage, the treatment advances to the next stage if the sample average treatment effect clears a threshold given by a statistical test with a specified significance level.

Formally, the continuation criteria are two Z-tests with one-sided%
\footnote{
  The choice of a one-sided test is without loss of generality, as in this case
  a one-sided test at level $\alpha$ is equivalent to a two-sided test at level
  $2\alpha$ when continuation requires a positive significant result.
}
significance levels $\alpha_{1}\in(0,1)$ and $\alpha_{2}\in(0,1)$.
Under the weak null that the treatment has no true average effect $\ATE_t=0$, the following test statistic has a standard normal sampling distribution by the distribution of the average treatment effect estimate:
\[
  \sqrt{n_t}\cdot\widehat{\ATE}_t
  =
  \sqrt{n_t}\cdot(\widehat{\ATE}_t - \ATE_t)
  \sim
  \mathcal{N}(0,1),
  \quad t\in\{1,2\}.
\]
Thus, one-sided tests at the $\alpha_{1}$ and $\alpha_{2}$ significance levels that allow continuation only for significant positive effects are given by the following decision rules,
\begin{align*}
  \Test_t &
  \defeq
  \indic{
    \sqrt{n_t}\cdot\widehat{\ATE}_t\ge
    z_{1-\alpha_t}
  } ,\quad z_{1-\alpha_t}\defeq\Phi^{-1}\left(1-\alpha_{t}\right)
\end{align*}
where we write $\Phi^{-1}$ for the normal inverse cumulative distribution function and $z_{1-\alpha_t}$ for the corresponding critical value of the test statistic.

\subsection{Pilot constraint, impact objective, and optimization problem}
\label{sec:pilot-budget-constraint-and-objective}

\paragraph{Pilot Sampling Costs}
For each type $d$, it costs $c_{d}>0$ to recruit a unit of type $d$ into the pilot RCT.
This cost can represent a monetary cost in enrolling the unit, a time cost in recruiting the unit or measuring its outcomes, etc.

\paragraph{Pilot Budget and Investment}
The pilot designer has a fixed budget $b > 0$ available to invest in running the pilot RCT.
When choosing the pilot design $\bm{s}_1$, the experimenter is subject to the constraint that the investment $I \defeq \bm{c}^{\top}\bm{s}_{1}$ is bounded by the budget $b$,
\[
  \bm{c}^{\top}\bm{s}_{1}\le b.
\]
We assume that the pilot budget constraint is the only binding constraint on the pilot sample.

\paragraph{Pilot RCT Impact}
Our main goal is to understand how to optimize the expected realized improvement in
outcomes in the target population---what we call the \emph{pilot RCT impact}.
The true average treatment effect is only realized in the target population if
the trial continues to a follow-up and results in adoption, so we can define the
pilot impact as%
\footnote{
  Note that writing $\ATE_3$ in the expression for the pilot impact, i.e.\ ignoring any finite sample variation in the target population average treatment effect, is without loss of generality since we are interested in the expected value.
  We could substitute some $\widehat{\ATE}_3$ for $\ATE_3$ in the above expression, where so long as $\E[\widehat{\ATE}_3\mid \bm{\tau}] = \ATE_3$ and the noise is independent given $\bm{\tau}$, the expected pilot impact would be unchanged.
}
\[
  V(\bm{s}_1) \defeq \E\left[
    \Test_1
    \cdot
    \Test_2
    \cdot\,
    \ATE_3
  \right]
\]
Other objectives can be formulated within our framework, such as maximizing the probability of passing the pilot statistical test $\E[\Test_1]$ or a sequence of tests $\E[\Test_1 \cdot \Test_2]$, or the impact of making an adoption decision based on a single RCT $\E[\Test_1 \cdot \ATE_3]$.
We compare small-budget optimal policies for these alternate objectives in \Cref{sec:alternate-objectives-appendix}.

We can formalize our mathematical goal as the following optimization problem.

\vspace{0.1cm}
\medskip
\noindent\textbf{Pilot RCT Impact Maximization Problem}
\vspace{0.1cm}
\begin{align*}
  \text{maximize}\quad  &
  V(\bm{s}_1)\\
  \text{s.t.}\quad &
  \begin{aligned}
    \bm{c}^{\top}\bm{s}_{1} & \le b &  & \text{(pilot
    budget constraint)}\\
    \bm{s}_{1} & \in\R_{\geq 0}^{D} \setminus \{\bm{0}\} &  &
    \text{(non-negative, non-zero
    sample)}
  \end{aligned}
\end{align*}

\vspace{0.2cm}
\noindent We note that in general this problem is non-convex, even under the assumption of a Gaussian prior on treatment effects that is commonly made in the literature on clinical trials and R\&D experiments in operations research~\cite{FrazierPo09,ChickGa22}.
For this reason, we will focus on the set of optimal solutions
\[
  \mathcal{S}_1^{\star}
  =
  \underset{\bm{s}_1}{\arg\max}
  \left\{
    V(\bm{s}_1) :
    \bm{c}^\top\bm{s}_1 \leq b,\
    \bm{s}_1 \in \R_{\geq 0}^{D} \setminus \{\bm{0}\}
  \right\}.
\]

In the following sections, we characterize this set of optimal solutions in two parameter regimes: for  large pilot budgets ($b\to\infty$) in \Cref{sec:large-budget} and small pilot budgets ($b$ smaller than a finite threshold) in \Cref{sec:small-budget}.
We write $\mathcal{S}_1^{\star}(b)$ when we wish to emphasize the dependence of the set of optimal solutions on the pilot budget $b$.
To start, we note that an optimal solution always exists.

\begin{proposition}[Existence of optimal
  solution]
  \label{prop:optimal-solution-exists}
  Suppose $\E[|\ATE_3|]<\infty$ and $\var(\ATE_3)>0$.
  Then the impact maximization problem has at least one optimal solution, i.e.\ $\mathcal{S}_1^{\star}\neq\emptyset$.
\end{proposition}

This follows from the observation that the objective is continuous and does not approach its supremum as $\bm{s}_1\to\bm{0}$, where $\bm{0}$ is the only limit point not in the feasible set since every valid design must have a positive sample size so that the objective is well-defined.
See \Cref{\sectionof{prop:optimal-solution-exists-appendix}} for the
\hyperref[prop:optimal-solution-exists-appendix]{proof}.

We will frequently reparametrize the decision variable $\bm{s}_1$ in terms of the pilot investment $\invst=\bm{c}^{\top}\bm{s}_1\in\R$ and a unit-investment normalized sample
\[
  \overline{\bm{s}}_1\in\{\bm{x}\in\R_{\geq 0}^D : \bm{c}^\top
  \bm{x}=1\}\eqdef\overline{\mathcal{S}}.
\]
By the definition of the feasible set above,
\[
  \bm{s}_1 = \invst\,\overline{\bm{s}}_1\ \ \text{is feasible if and
  only if}\ \ 0 < \invst \le
  b\ \ \text{and}\ \ \overline{\bm{s}}_1\in\overline{\mathcal{S}}.
\]

%% file: diagrams/simplified_experiment_process_flow.tex
\newcommand{\drawdownarc}[4][]{%
  \path let \p1=(#2), \p2=(#3) in
  \pgfextra{
    \pgfmathsetmacro{\dx}{\x2-\x1}
    \pgfmathsetmacro{\dy}{\y2-\y1}
    \pgfmathsetmacro{\L}{sqrt(\dx*\dx+\dy*\dy)}

    \pgfmathsetlengthmacro{\slen}{#4cm}
    \pgfmathsetmacro{\s}{\slen}

    \pgfmathsetmacro{\r}{(\L*\L)/(8*\s) + (\s/2)}
    \pgfmathsetmacro{\d}{\r-\s}

    \pgfmathsetmacro{\xm}{(\x1+\x2)/2}
    \pgfmathsetmacro{\ym}{(\y1+\y2)/2}

    \pgfmathsetmacro{\nx}{\dy/\L}
    \pgfmathsetmacro{\ny}{-\dx/\L}
    \pgfmathsetmacro{\flip}{ifthenelse(\ny<0,-1,1)}
    \pgfmathsetmacro{\nx}{\flip*\nx}
    \pgfmathsetmacro{\ny}{\flip*\ny}

    \pgfmathsetmacro{\xc}{\xm + \nx*\d}
    \pgfmathsetmacro{\yc}{\ym + \ny*\d}

    \pgfmathsetmacro{\angA}{atan2(\y1-\yc,\x1-\xc)}
    \pgfmathsetmacro{\theta}{2*asin(\L/(2*\r))}

    \pgfmathsetmacro{\sgn}{ifthenelse(\dx>0,1,-1)}
    \pgfmathsetmacro{\delta}{\sgn*\theta}

    \xdef\ArcStart{\angA}
    \xdef\ArcDelta{\delta}
    \xdef\ArcRadius{\r}
  };
  \draw[#1] (#2) arc[start angle=\ArcStart, delta angle=\ArcDelta,
  radius=\ArcRadius pt];
}

\begin{tikzpicture}[
    font=\small,
    >=Triangle,
    stage/.style={align=center},
    label/.style={align=center, font=\bfseries},
    flow/.style={->, line width=0.75pt, line cap=round},
    decisionflow/.style={flow, dashed}
  ]

  \def\colsep{3.5}
  \def\toplabsep{0.40}
  \def\rowsep{1.65}
  \def\arrowgap{0.15}

  \def\curvedrop{0.40}
  \def\curvegap{0.18}
  \def\midgap{0.18}

  \def\testdrop{1.35}

  \def\boxpadx{9pt}
  \def\boxpady{7pt}

  \coordinate (c1) at (0,0);
  \coordinate (c2) at (\colsep,0);
  \coordinate (c3) at (2*\colsep,0);

  \node[stage] (S1) at (c1) {Pilot Sample $\bm{s}_1$};
  \node[stage] (S2) at (c2) {Follow-up Sample $\bm{s}_2$};
  \node[stage] (S3) at (c3) {Population $\bm{s}_3$};

  \node[label] (DecLbl)     at ($(S1.north)+(0,\toplabsep)$) {Decision};
  \node[label] (FixTopLbl)  at
  ($($(S2.north)!0.5!(S3.north)$)+(0,\toplabsep)$) {Fixed};

  \node[stage] (H1) at ($(S1)+(0,-\rowsep)$) {$\widehat{\ATE}_1$};
  \node[stage] (H2) at ($(S2)+(0,-\rowsep)$) {$\widehat{\ATE}_2$};
  \node[stage] (Y)  at ($(S3)+(0,-\rowsep)$) {$\ATE_3$};

  \node[stage, text=red] (T1) at ($($(H1)!0.5!(H2)$)+(0,-\testdrop)$)
  {$\Test_1$};
  \node[stage, text=red] (T2) at ($($(H2)!0.5!(Y)$ )+(0,-\testdrop)$)
  {$\Test_2$};

  \begin{pgfonlayer}{background}
    \node (BoxDec) [
      fit=(DecLbl)(S1),
      fill=blue!25,
      inner xsep=\boxpadx,
      inner ysep=\boxpady
    ] {};

    \node (BoxFixTop) [
      fit=(FixTopLbl)(S2)(S3),
      fill=gray!25,
      inner xsep=\boxpadx,
      inner ysep=\boxpady
    ] {};

  \end{pgfonlayer}

  \coordinate (S2BoxSouth) at (S2.south |- BoxFixTop.south);
  \coordinate (S3BoxSouth) at (S3.south |- BoxFixTop.south);

  \draw[flow] ($(BoxDec.south)+(0,-\arrowgap)$)   --
  ($(H1.north)+(0,\arrowgap)$);
  \draw[flow] ($(S2BoxSouth)+(0,-\arrowgap)$)    --
  ($(H2.north)+(0,\arrowgap)$);
  \draw[flow] ($(S3BoxSouth)+(0,-\arrowgap)$)    -- ($(Y.north)+(0,\arrowgap)$);

  \coordinate (ArcAstart) at ($(H1.south)+(0,-\curvegap)$);
  \coordinate (ArcAend)   at ($(H2.south)+(-\midgap,-\curvegap)$);
  \coordinate (ArcBstart) at ($(H2.south)+(\midgap,-\curvegap)$);
  \coordinate (ArcBend)   at ($(Y.south)+(0,-\curvegap)$);

  \drawdownarc[decisionflow]{ArcAstart}{ArcAend}{\curvedrop}
  \drawdownarc[decisionflow]{ArcBstart}{ArcBend}{\curvedrop}

\end{tikzpicture}

%% file: paper/03_large_budget.tex
\section{Representative Samples are Optimal in Large-Budget Pilots}
\label{sec:large-budget}

We begin our analysis of optimal pilot design by considering the optimal pilot sample composition for large pilot budgets.
We show that as the pilot budget grows $b\to\infty$, the set of normalized optimal pilot designs converges to a representative sample---one that matches the composition of the target population $\bm{s}_1 \propto \bm{s}_3$, for any fixed follow-up sample $\bm{s}_2$.
Intuitively, a large pilot allows treatments to continue if and only if they have positive effects in the pilot sample, so we should want them to have positive effects in the pilot if and only if they do in the target population.
This is ensured uniquely by a large representative sample.

We first state our main result for the large-budget regime: as the pilot budget $b$ increases without bound, the investment of optimal pilots grows without bound ($I\to\infty$) and the normalized composition of any optimal pilot converges to the normalized representative design $\bm{s}_3 / \bm{c}^\top \bm{s}_3$.
\begin{theorem}[Large-budget Limiting Optimal
  Policy]
  \label{thm:large-budget-limiting-optimal-policy}
  Suppose the prior distribution of $\bm\tau$ has a density positive on an open ball containing the origin, and the average treatment effect on the target population has finite second moment under the prior.
  Then in the large-budget limit $b\to\infty$, the investment of optimal pilots diverges
  \[
    \inf_{\bm s_1^\star\in\mathcal{S}_1^{\star}(b)} \bm
    c^\top \bm s_1^\star
    \to \infty
  \]
  and the unit-investment normalized optimal designs converge to the normalized representative design,
  \[
    \sup_{\bm s_1^\star\in\mathcal{S}_1^{\star}(b)}
    \left\|
    \frac{\bm s_1^\star}{\bm c^\top \bm
    s_1^\star}-\frac{\bm s_3}{\bm c^\top \bm s_3}
    \right\|
    \to 0.
  \]
\end{theorem}

See \Cref{\sectionof{thm:large-budget-limiting-optimal-policy-appendix}} for the
\hyperref[thm:large-budget-limiting-optimal-policy-appendix]{proof}.
We prove the result by showing a more general large-budget result for objectives of the form
\[
  V_Z(\bm s_1)\defeq \E[\Test_1 \cdot Z\cdot \ATE_3],
\]
where $Z\geq 0$ is a bounded weight that is non-adaptive (independent of the pilot outcome conditional on the true treatment effect), has positive expectation conditional on the true treatment effect $\E[Z\mid\bm\tau]>0$ almost surely.
The pilot impact objective is the special case $Z=\Test_2$, while the impact of adopting based on a single RCT is the special case $Z=1$.
Therefore, representative sampling is the unique large-budget limiting impact-optimal design whether or not there is a follow-up RCT before adoption.

The proof proceeds in two stages.
First, we show that the oracle performance achieved by continuing exactly when the target population effect $\mathrm{ATE}_3$ is positive can be matched by a noise-free ``infinite-budget'' pilot if and only if the pilot sample is representative $\bm{s}_1 \propto \bm{s}_3$, since only then does the pilot not make any type I or type II errors, as visualized in \Cref{fig:large-budget-visualization}.
Second, we show that as the pilot investment grows, its performance is uniformly approximated by the performance of the noise-free pilot.
Performance is bounded away from the oracle performance both for any fixed non-representative pilot sample composition and if the pilot investment is capped, so the optimal pilot converges to a representative sample with ever-increasing investment.
This result and the idea of the noise-free pilot match one intuitive reason why representative samples might be desirable in pilot studies: we should want exactly those treatments that have positive effects in the target population to receive a follow-up study.
We can only approximate this rule when the pilot budget is large.
This sets the stage for a setting where representative pilots are not optimal.

\vspace{0.2cm}
\begin{figure}[ht]
  \centering
  \includegraphics[width=0.8\textwidth]{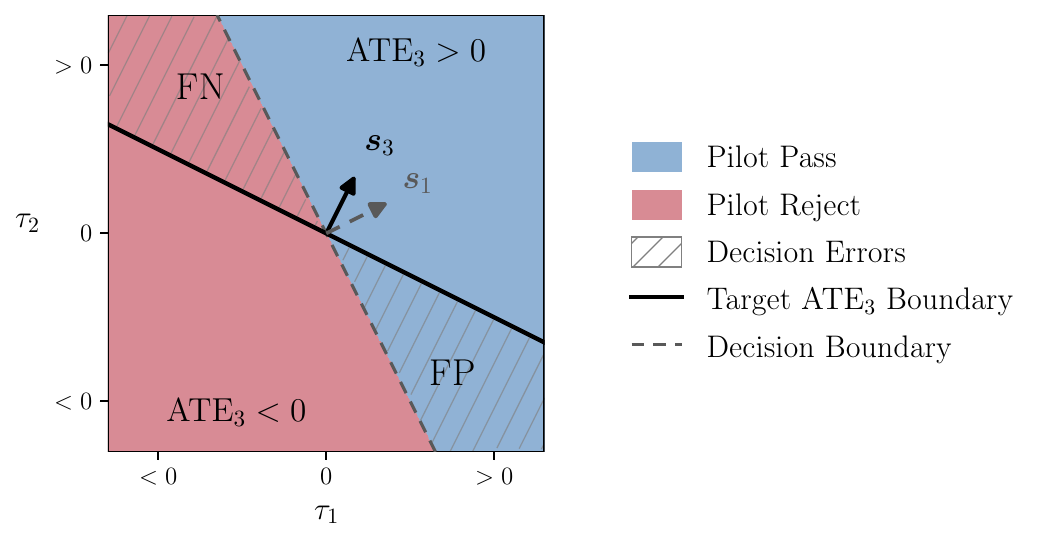}
  \caption{
    \textbf{Why a representative sample is optimal for large budgets.}
    When the sample is non-representative, i.e. $\bm{s}_1 \npropto \bm{s}_3$, the pilot decision boundary is misaligned with the boundary separating positive and negative average treatment effects in the target population.
    Even with an infinite sampling investment, there are conditional average treatment effects $\bm{\tau}$ for which the pilot incorrectly classifies the sign of the target population average treatment effect $\mathrm{ATE}_3$.
  }
  \label{fig:large-budget-visualization}
\end{figure}

%% file: paper/04_small_budget.tex
\section{Homogeneous Samples are Optimal for Small-Budget Pilots}
\label{sec:small-budget}

The key insight we derive in our model is that whether representative sampling is optimal depends critically on how resource-constrained the pilot study is.
While in the large-budget regime representative sampling is near-optimal, in this section we show that for sufficiently small pilot budgets, the optimal pilot design concentrates sampling on a single type.

We begin in \Cref{sec:small-budget-optimal-design} by deriving the optimal small-budget design when an RCT is used to determine whether to receive \emph{any} unknown downstream payoff under a general prior distribution with moment assumptions.
In \Cref{sec:small-budget-optimal-pilot} we specialize the result to give the optimal design for maximizing the pilot RCT impact in our model, where the payoff is the result of running a follow-up RCT and making an adoption decision based on its statistical significance.
We then discuss in \Cref{sec:small-budget-heterogeneity} how the optimal pilot design depends on the prior heterogeneity of the treatment effects under additional assumptions on the prior distribution.

\subsection{Small-budget optimal design for general downstream payoffs}
\label{sec:small-budget-optimal-design}

Consider an RCT sample $\bm{s}_1$ with total sample size $n_1=\bm{1}^\top\bm{s}_1$ and unknown vector of heterogeneous treatment effects $\bm{\tau}$.
As described in \Cref{sec:Model}, consider a statistical test of the form
\[
  \Test_1 \defeq \indic{\sqrt{n_1}\cdot \widehat{\ATE}_1\ge z_{1-\alpha_1}}.
\]
However, assume the test determines whether the decision maker receives a general downstream payoff $Y$, and that the objective is to maximize the expected realized payoff,
\[
  V_Y(\bm{s}_1)
  \defeq
  \E\left[
    \Test_1\cdot Y
  \right].
\]
We derive the optimal small-budget design for this objective, under the assumption that the payoff $Y$ is \emph{non-adaptive}, i.e. it does not depend directly on the RCT observations or, more formally, the payoff is independent of the pilot outcomes conditioned on the treatment effect vector $\bm{\tau}$.

We begin with a heuristic derivation.
The test statistic can be written as a linear function of the RCT sample divided by the square root of the total sample size, plus standard normal noise
\[
  \sqrt{n_1}\cdot\widehat{\ATE}_1
  =
  \sqrt{n_1}\cdot\frac{\bm{s}_1^\top\bm\tau}{n_1}+\varepsilon_1
  =
  \frac{\bm{\tau}^\top\bm{s}_1}{\sqrt{\bm{1}^\top\bm{s}_1}}+\varepsilon_1,
  \qquad
  \varepsilon_1\mid \bm\tau\sim\mathcal{N}(0,1).
\]
When samples are small, so is the expected value---the non-centrality---of this test statistic.
A Taylor expansion of the objective around zero non-centrality shows that in the small-budget regime, the objective is approximately equivalent to maximizing the expected product of the test statistic non-centrality and the downstream payoff.
\begin{alignat*}{2}
  V_Y(\bm{s}_1)
  &=
  \E\left[\indic{\frac{\bm{\tau}^\top\bm{s}_1}{\sqrt{\bm{1}^\top\bm{s}_1}}+\varepsilon_1 \ge z_{1-\alpha_1}} \cdot Y\right]
  &\qquad&
  (\text{definition of }\Test_1) \\
  &=
  \E\left[
    \Phi\!\left(
      \frac{\bm{\tau}^\top\bm{s}_1}{\sqrt{\bm{1}^\top\bm{s}_1}}
      - z_{1-\alpha_1}
    \right)
    \cdot Y
  \right]
  &\qquad&
  (\text{non-adaptivity, }Y\ind \varepsilon_1 \vert \bm{\tau}) \\
  &\approx
  \alpha_1\E[Y]
  +\varphi(z_{1-\alpha_1})\,
  \E\left[
    \frac{\bm{\tau}^\top\bm{s}_1}{\sqrt{\bm{1}^\top\bm{s}_1}}
    \cdot Y
  \right]
  &\qquad&
  (\text{expansion at zero non-centrality}) \\
  &=
  \alpha_1\E[Y]
  +\varphi(z_{1-\alpha_1})\,
  \frac{\E[Y \bm\tau]^\top\bm{s}_1}{\sqrt{\bm{1}^\top\bm{s}_1}}
  &\qquad&
  (\text{linearity of expectation})
\end{alignat*}

This means that for small budgets, maximizing the expected payoff is approximately equivalent to maximizing a linear signal divided by the square root of the sample size:
\begin{align*}
  \underset{\bm{s}_1}{\maximize}
  \quad
  \frac{\E[Y \bm\tau]^\top\bm{s}_1}{\sqrt{\bm{1}^\top\bm{s}_1}}
  \qquad
  \subjectto
  \quad
  \bm c^\top\bm{s}_1\le b,
  \quad
  \bm{s}_1\in\R_{\geq 0}^{D}\setminus\{\bm{0}\}.
\end{align*}
The ratio of a linear function to the square root of a linear function is quasiconvex when both linear functions are positive, and quasiconvex functions are maximized over polytopes at extreme points.
Thus, so long as there is at least one type with positive signal $\E[Y \tau_d]>0$, the approximate objective is maximized by investing the entire budget into sampling a single type: the one that maximizes the \emph{small-budget index}
\[
  \pi_d^Y
  \defeq
  \frac{\E[Y \tau_d]}{\sqrt{c_{d}}}.
\]

Formalizing the argument above requires care on two points: showing that the objective does not approach a maximizer as the sample approaches zero $\bm{s}_1\to 0$ where the objective is undefined, and establishing uniform control of the remainder of the Taylor expansion to show that the optimizer of the objective and its leading-order term are exactly equal for sufficiently small finite budgets.
We do so by writing the objective in terms of the investment $\invst=\bm{c}^\top \bm{s}_1$ and the unit-investment normalized sample composition $\overline{\bm{s}}_1$ as described in \Cref{sec:Model}, and taking a Taylor expansion in terms of $\sqrt{I}$.

\begin{proposition}[Small-investment expansion for non-adaptive payoffs]
  \label{prop:small-investment-expansion}
  Fix a prior distribution of $\bm\tau$, and let $Y$ be non-adaptive with $\E[|Y|]<\infty$ and $\E[|Y|\|\bm\tau\|_2^3]<\infty$.
  Then there exists an investment threshold $\overline{\invst}>0$ such that for all investments $0 < \invst \le \overline{\invst}$ and all unit-investment samples $\overline{\bm{s}}_1\in\overline{\mathcal{S}}$,
  \[
    V_Y(\invst\,\overline{\bm{s}}_1)
    =
    \alpha_1\, \E[Y]
    +\sqrt{\invst}\,\varphi(z_{1-\alpha_1})\,
    \frac{\E[Y\bm\tau]^{\!\top}\overline{\bm{s}}_1}{\sqrt{\bm{1}^{\top}\overline{\bm{s}}_1}}
    +R(\invst,\overline{\bm{s}}_1),
  \]
  where $R(\invst,\overline{\bm{s}}_1)$ is $O(\invst)$ and Lipschitz in $\invst$ (both uniformly in $\overline{\bm{s}}_1$) and $O(\invst)$-Lipschitz in $\overline{\bm{s}}_1$.
\end{proposition}

This result and the following small budget optimal policy are proved in \Cref{\sectionof{prop:small-investment-expansion-appendix}}.
The expansion of the expected payoff objective allows us to show that for all budgets~$b>0$ below a fixed budget threshold~$b<\overline{b}$ depending on the problem parameters, the optimal design is completely determined by the first non-constant term in the expansion,
\[
  \sqrt{\invst}\,\cdot\,
  \frac{\E[Y\bm\tau]^{\!\top}\overline{\bm{s}}_1}{\sqrt{\bm{1}^{\top}\overline{\bm{s}}_1}}.
\]
Similarly to our heuristic derivation above, we argue that the linear over square-root-of-linear term is maximized at an extreme point.
Here, however, it is a function on the set of unit-investment designs $\overline{\mathcal{S}}=\{\bm{x}\in\R^D_{\ge 0} : \bm{c}^{\top}\bm{x}=1\}$ which is a compact polytope and does not contain the zero sample for which our objective is undefined.
The term scales with the square root of the investment, and is therefore maximized by investing the entire budget.
Thus, the optimal design invests the entire budget in a single type, and the Lipschitz properties of the remainder allow us to extend from the term above to the original expected payoff objective under a sufficiently small budget.
This brings us to the general small-budget result.

\begin{theorem}[Optimal small-budget design for non-adaptive payoffs]
  \label{thm:optimal-small-budget-design-general}
  Fix a prior distribution of $\bm\tau$, and let $Y$ satisfy the assumptions of \Cref{prop:small-investment-expansion}.
  Suppose there exists a unique type $d^\star$ maximizing the \emph{small-budget index}
  \[
    \pi_d^Y
    \defeq
    \frac{\E[Y \tau_d]}{\sqrt{c_{d}}}
    \in
    \R.
  \]
  Suppose also that $\pi_{d^\star}^Y>0$.
  Then there exists a budget threshold $\overline{b}>0$ such that for all budgets $0 < b \le \overline{b}$, a unique optimal design exists for $V_Y$ and invests the entire budget in the type $d^\star$.
\end{theorem}

\subsection{Small-budget optimal pilot}
\label{sec:small-budget-optimal-pilot}

For the pilot impact objective, the downstream payoff is $Y=\Test_2\cdot\ATE_3$, so the small-budget index becomes
\[
  \idx_d
  \defeq
  \frac{\E[\tau_d \cdot \Test_2 \cdot \ATE_3]}{\sqrt{c_{d}}}
  \in
  \R.
\]

\begin{corollary}[Optimal small-budget pilot]
  \label{cor:optimal-small-budget-pilot}
  Suppose the prior distribution of $\bm\tau$ has finite fourth moments, and that there exists a unique type $d^\star$ maximizing the small-budget index $\idx_d$.
  Then there exists a budget threshold $\overline{b}>0$ such that for all pilot budgets $b<\overline{b}$, a unique optimal pilot exists and invests the entire budget in the type $d^\star$.
\end{corollary}

This follows from \Cref{thm:optimal-small-budget-design-general}.
The probability model and finite fourth moments imply the non-adaptivity and moment assumptions for $Y=\Test_2\cdot\ATE_3$, and the positivity condition is automatic for the pilot impact objective; see \Cref{\sectionof{prop:positive-impact-objectives-appendix}}.

Not only are we able to show that representative sampling is never optimal for pilots that screen treatments for follow-up study under a sufficiently small budget, in fact the other extreme is optimal here: sampling only from a subpopulation containing a single type.
This is counter to the many calls for the use of more representative samples in the social and health sciences~\cite{FDAHH20, DepartmentofEducationBo22,BryanTi21}, and reflects that modeling decision outcomes under a decision rule requiring statistical significance yields qualitatively different optimal designs than those usually found in work focused on externally valid inference~\cite{Tipton13,EgamiLe23,Bouyamourn25} and minimax-regret objectives~\cite{HuZh24,OleaPr25}.

Further, we are able to characterize exactly which type to sample.
Intuitively, large $\tau_d$ makes passing the pilot more likely, and we want this to happen when the target effect $\ATE_3$ is large under those realizations where the follow-up statistical test $\Test_2$ passes.
Thus, we trade off the expected value of $\tau_d \cdot \Test_2 \cdot \ATE_3$ against the square root of the sampling cost $c_d$, since the impact improves with the square root of the investment.
We now proceed to give a more interpretable characterization of the optimal type to sample under additional assumptions on the prior distribution of the treatment effects.

\subsection{The influence of treatment effect heterogeneity on
optimal pilot design}
\label{sec:small-budget-heterogeneity}

In addition to characterizing the optimal sampling type for the pilot, \Cref{cor:optimal-small-budget-pilot} also allows us to interpret how the optimal type to sample in the small-budget regime should depend on our treatment effect prior.
In the remainder of this section we analyze the form of $\pi_d$ by restricting the prior to the class of elliptical distributions~\cite{FangKo90}
\[
  \bm{\tau}\sim\mathrm{EC}_D(\bm{\mu},\Sigma,\phi),
\]
which includes and can be viewed as generalizing the multivariate normal distribution.
We begin by informally describing the dependence under typical experimentation conditions, before giving in \Cref{prop:marginal-impact-interpretation} a decomposition of the small-budget index from which this interpretation is derived.

To begin, assume the follow-up RCT uses a representative sample, captured in our model by assuming $\bm{s}_2\propto \bm{s}_3$.
Further, assume the prior expected target effect is not harmful, $\E[\ATE_3]\ge 0$, and that the target-population effect is nondegenerate, $\var(\ATE_3)>0$.

Heterogeneity in the conditional average treatment effects $\tau_d$ across types can come into the prior in two ways: \emph{expected} variation in the prior mean $\E[\tau_d]$ across types, and \emph{unknown} variation through the covariance matrix $\cov(\bm\tau)$.
Under the scenario described above, the small-budget index is proportional to a weighted sum of a type's prior mean and how predictive its treatment effect is of the target population effect, scaled by the square root of the sampling cost: for some positive weight $\overline W_1$,
\vspace{0.3cm}
\begin{equation}
  \idx_d \propto \frac{\overbrace{\overline W_1 \cdot \E[\tau_d]}^{\text{Expected
    heterogeneity}} \ + \
    \overbrace{\cov(\tau_d, \ATE_3)}^{\text{Unknown heterogeneity}}
  }{\underbrace{\sqrt{c_d}}_{\text{Cost heterogeneity}}}.
  \label{eq:small-budget-index-interpretation-special-case}
\end{equation}
\vspace{0.3cm}

To see how the prior covariance $\cov(\bm\tau)$ affects the expression above, consider the further specialization to a unit-variance exchangeable prior covariance with common correlation $\rho\in[0,1]$,
\[
  \cov(\bm\tau) =
  \begin{pmatrix}
    1 & \rho & \cdots & \rho \\
    \rho & 1 & \cdots & \rho \\
    \vdots & \vdots & \ddots & \vdots \\
    \rho & \rho & \cdots & 1
  \end{pmatrix}.
\]
In this case each type effect $\tau_d$ is predictive of the target population effect $\ATE_3$ by an amount depending on its share of the target population $s_{3d}/n_3$, attenuated by the common correlation $\rho$.
\begin{equation}
  \cov(\tau_d, \ATE_3)
  =
  \sum_{d'=1}^D \frac{s_{3d'}}{n_3}\cov(\tau_d, \tau_{d'})
  =
  \rho + (1-\rho)\frac{s_{3d}}{n_3}
  \label{eq:small-budget-index-interpretation-special-case-cov}
\end{equation}
These assumptions allow us to interpret the small-budget index under different kinds of elliptical priors:
\begin{enumerate}
  \item {
      When we do not expect much heterogeneity in the treatment effect across types, the prior means $\mu_d$ are roughly equal across types and the prior correlation $\rho$ is close to 1.
      In this case the numerator in \eqref{eq:small-budget-index-interpretation-special-case} is constant and the small-budget optimal policy depends only on cost: it reduces to sampling the cheapest type, i.e. convenience sampling.
    }
  \item {
      If we have prior knowledge that there is potential heterogeneity in the treatment effect but do not anticipate that it favors any specific type, the prior correlation $\rho$ is small and the prior means $\E[\tau_d]$ are still equal.
      In this case the optimal small-budget pilot skews towards types that make up a larger share of the target population.
      See \Cref{fig:exchangeable-2types-by-rho-zero-mean} for an example.
    }
  \item {
      Finally, if we have prior knowledge that the expected effect is larger in specific types than others, so $\E[\tau_d]$ is not equal across types, and we believe that the expected effect is not harmful $\E[\ATE_3]\ge 0$, the optimal small-budget pilot skews towards types with larger prior means.
      See \Cref{fig:exchangeable-2types-by-mu} for an example.
    }
\end{enumerate}
\begin{figure}[ht]
  \centering
  \begin{subfigure}[t]{0.48\textwidth}
    \centering
    \includegraphics[width=\linewidth]{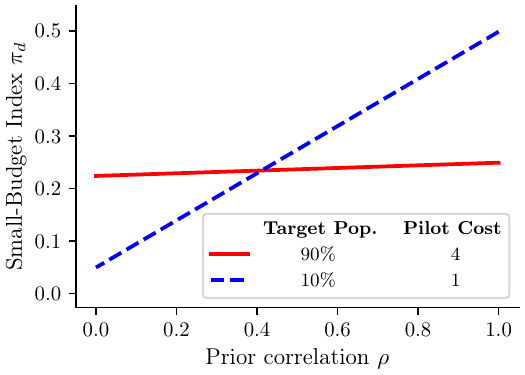}
    \caption{
      Decreasing $\rho$ (increasing heterogeneity) favors the type that has a larger target population share and is therefore more predictive.
      Here under $\bm\mu=\bm{0}$.
    }
    \label{fig:exchangeable-2types-by-rho-zero-mean}
  \end{subfigure}
  \hfill
  \begin{subfigure}[t]{0.48\textwidth}
    \centering
    \includegraphics[width=\linewidth]{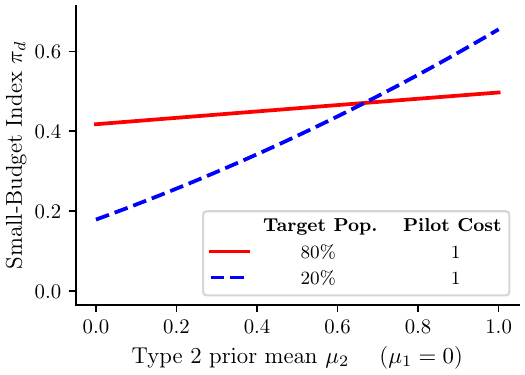}
    \caption{
      Increasing $\mu_2$ (the prior mean of type 2) while keeping $\mu_1 = 0$ shifts the index toward type 2, even though type 1 is more common.
      Here under $\rho=0.2$.
    }
    \label{fig:exchangeable-2types-by-mu}
  \end{subfigure}
  \caption{
    The small-budget index for two types under a large ($n_2 = 100$), representative follow-up and a multivariate normal prior with unit prior standard deviations.
    For any sufficiently small budget, the optimal pilot invests the entire budget in sampling the type with the larger index.
  }
  \label{fig:exchangeable-2types}
\end{figure}

\Needspace{12\baselineskip}
The preceding interpretations are formalized in the following general decomposition of the small-budget index, proved in \Cref{\sectionof{prop:marginal-impact-interpretation-appendix}}.

\begin{proposition}[The small-budget index under an elliptical prior]
  \label{prop:marginal-impact-interpretation}
  Suppose the treatment effect prior is an elliptical distribution,
  \[
    \bm\tau\sim\mathrm{EC}_D(\bm\mu,\Sigma,\phi),
  \]
  with finite second moments.
  Then there exist scalars
  \(W_1,W_2,W_3\in\R\), not depending on \(d\), such that for every
  type \(d\),
  \[
    \idx_d
    =
    \frac{
      W_1\,\E[\tau_d]
      +
      W_2\,\cov(\tau_d,\ATE_2)
      +
      W_3\,\cov(\tau_d,\ATE_3)
    }{\sqrt{c_d}}.
  \]

  If additionally \(\bm s_2\propto \bm s_3\), \(\E[\ATE_3]\ge0\), and \(\var(\ATE_3)>0\), then the decomposition reduces to
  \[
    \idx_d
    \propto
    \frac{
      \overline W_1\,\E[\tau_d]
      +
      \cov(\tau_d,\ATE_3)
    }{\sqrt{c_d}},\qquad \overline W_1 > 0
  \]
  with a positive constant of proportionality, where the weight on the prior mean is strictly positive.
\end{proposition}

This decomposition shows that the small-budget index $\idx_d$ is the cost-adjusted, weighted sum of three interpretable components: the prior mean $\E[\tau_d]$ of the type's treatment effect, how predictive it is of the follow-up average effect $\ATE_2$, and how predictive it is of the target population effect $\ATE_3$.
Furthermore, when the follow-up is representative and the treatment is not expected to harm the target population, the index prioritizes types with larger prior means.

%% file: paper/05_simulation_studies.tex
\section{Case Study}
\label{sec:case-study}

We now illustrate the optimal pilot study design in a realistic setting, using a simulated case study based on education RCT data.
While the theory from \Cref{sec:large-budget} and \Cref{sec:small-budget} holds only in the large budget regime $(b\to\infty)$ and the small budget regime $(b<\overline{b})$, in this section we show that both regimes can correspond to realistic sample sizes when we calibrate our model to real data.

For this case study, we focus on the trade-off between sampling from the most promising type and sampling a representative sample, and assume constant sampling costs across school types.
Over-sampling from the most promising type is done in practice for a variety of reasons, including making the potential treatment effect easier to detect and reducing the potential for causing harm~\cite{FDACD19}.
Our model captures the increased power of a targeted sample, and allows us to evaluate the trade-off between statistical power and generalizability.

We model a hypothetical pilot study preceding a real large-scale education RCT, the National Study of Learning Mindsets (NSLM) \cite{YeagerHa19}.
This experiment focused on a large representative sample of 6,320 lower achieving 9th grade students.
It scaled up the evaluation of an intervention designed to teach students a growth mindset, which had already been evaluated in prior RCTs.
Because there were theoretical reasons to believe that the treatment effect of the intervention would be moderated by at least two school-level features, we can model the decision of a pilot designer who could have chosen to recruit a representative sample or to focus on schools where the treatment had the best chance of success.

We begin by describing the details of our model calibration.
The follow-up is fixed to be representative with a sampling standard deviation corresponding to a sample size of 6,320 students.
This means we calibrate the \emph{effective} follow-up size $n_2$ used in the model by setting $\sqrt{1 / n_2} = 0.03$ to match the reported standard error of the NSLM study~\cite{YeagerHa19}.
We can then use the ratio between the number of students in a hypothetical pilot and the number of students in the NSLM study to calibrate the effective pilot size $n_1$.
We use one-sided significance tests at level $\alpha = 0.025$.
Since we assume costs are constant across types, the pilot budget is equivalent to the number of students sampled (the $x$-axis of \Cref{fig:simulation_2_policy_comparison}).

\begin{figure}[ht]
  \centering
  \includegraphics[width=0.78\textwidth]{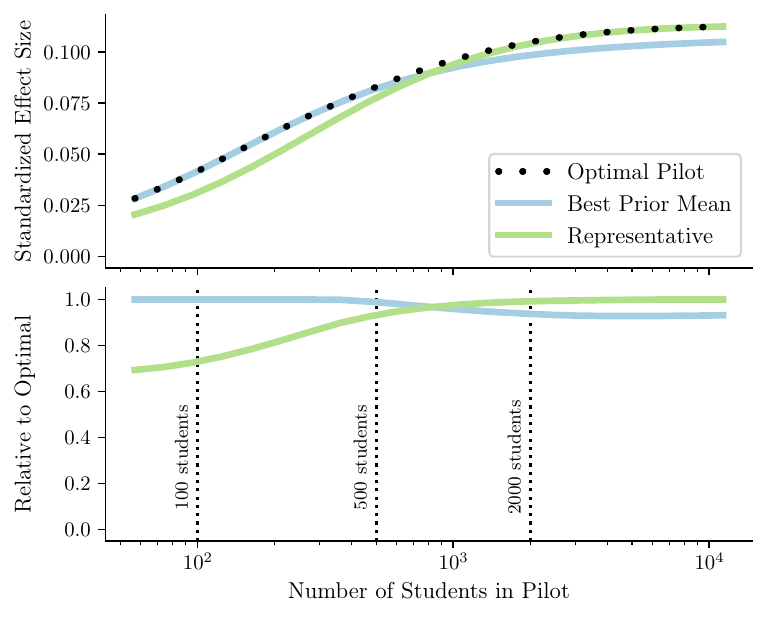}
  \caption{
    \textbf{Pilot RCT Impact.}
    Comparison of the small-budget and limiting large-budget optimal pilot designs for a range of sample sizes in a semi-synthetic scenario calibrated to the NSLM study.
    In this instance where costs across types are assumed constant, the budget is equivalent to a constraint on the number of students, and the small-budget optimal design is equivalent to sampling only from the schools with the highest prior mean.
    The limiting large-budget optimal design is representative sampling.
    The top panel shows the expected change in standardized effect size achieved by the small-budget and large-budget optimal designs as well as an optimal design, while the bottom panel shows each design's improvement over a baseline whose continuation decision is purely random with the same false-positive rate, as a fraction of the optimal design's improvement
  }
  \label{fig:simulation_2_policy_comparison}
\end{figure}

We model 6 types of schools based on the school-level moderators reported in \textcite{YeagerHa19}.
They describe that theory supported two school-level moderators of the treatment effect: the school's achievement level, and the supportiveness of peer norms at the school.
They divide the schools into three achievement levels: low ($<25$th), medium ($25$-$75$th), and high ($>75$th percentile), with medium achievement schools expected to benefit most from the intervention.
Similarly, they divide schools into two supportiveness levels: low ($<50$th) and high ($>50$th percentile), with high supportiveness schools expected to benefit more from the intervention.
We define our 6 school types as corresponding to each combination of achievement and supportiveness levels, and assume they occur with frequency proportional to the product of their group frequencies (since \textcite{YeagerHa19} did not report their joint frequency distribution).

We base our prior on the distribution of effects across education RCTs.
In meta-analyses of education studies it is common to report standardized effect sizes, scaling the treatment effect to the standard deviation of the outcome~\cite{Kraft20}.
\textcite{YeagerHa19} discuss that a standardized effect size of 0.20 would be on the large end of possible effects, so we set the prior mean for the most promising type---medium achievement, high peer supportiveness---to 0.10 as a reasonable expectation of a good treatment effect before the pilot stage.
We use a multivariate normal prior, and set the prior mean for types that have one lower-prior moderator to 0.05, and types that have both lower-prior moderators to 0.025.
We set the prior standard deviation for each type to 0.28, matching the empirical value reported in a review of 747 preK--12 education RCTs \cite{Kraft20}, and we fix a common correlation value of $0.5$ between all types to represent a moderate level of heterogeneity.

\Cref{fig:simulation_2_policy_comparison} visualizes the pilot impact in terms of standardized effect size after calibrating our model as described above.
For pilot samples below 323 students, sampling from only the most promising schools is optimal (in this instance, the small-budget optimal design is equivalent to sampling only the school-type with the highest prior mean).
For pilot samples above 823 students representative sampling outperforms the small-budget optimal design, and above about 2000 students, representative sampling is close to optimal.
At least some of the RCTs that evaluated the growth mindset intervention prior to the NSLM study had enrolled ``thousands of adolescents''~\cite{YeagerHa19}, so both ends of this range of pilot sample sizes are plausible.

%% file: paper/06_related_work.tex
\section{Related Work}
\label{sec:related-work}

In this section we review three areas of related work: (i)~treatment effect heterogeneity and external validity, (ii)~decision-focused experiment design, and (iii)~sequential experimentation.

\subsection{External validity and generalizability}

How well experimental findings generalize to broader populations has historically taken a back seat to concerns about internal statistical validity in a number of fields, including the social sciences~\cite{EgamiHa23} and medicine~\cite{AverittWe20}.
Growing evidence suggests this lack of attention to \emph{external validity} is consequential.
Treatment effects often vary substantially across subpopulations, settings, and time periods~\cite{BryanTi21,Tipton21,Meager19,WilliamsCi22}.
The replication crisis in psychology, economics, and biomedical research has revealed that interventions showing promise in initial studies frequently fail to produce similar effects in new contexts~\cite{OpenScienceCollaboration15,CamererDr16,BegleyEl12,CommitteeonReproducibilityandReplicabilityinScienceBo19}.

In resource-constrained settings, many experiments rely on convenience samples.
Driven by limited budgets, tight timelines, and difficulty in recruiting participants from hard-to-reach populations, experimenters select sites based on geographic proximity to their institution, existing relationships, and other factors affecting recruitment~\cite{TiptonFe16,TiptonSp21}.
Enrichment strategies are another common reason for using non-representative samples in clinical trials.
These strategies select a sample from a subpopulation that contains subjects with either an especially positive predicted response to the treatment or an especially poor prognosis, making it easier to detect statistically significant effects~\cite{FDACD19,Thall21,SteingrimssonBe21}.
A third reason for non-representative samples is that certain subpopulations may be more vulnerable or protected, and therefore not eligible for the study.
This is particularly common in biomedical research, where safety is a primary concern~\cite{FDAHH20}.

In response to concerns about external validity, researchers and regulators have advocated for the broader use where possible of representative sampling---constructing experimental samples that mirror the composition of the target population~\cite{Tipton13,FDAHH20}.
This design principle is realized in many large RCTs by attempting to draw a uniform random sample from the target population---for example, in \Cref{sec:case-study} we model the pilot preceding a large education RCT that attempted to draw a random sample of US public schools \cite{YeagerHa19}.
When such random sampling is not feasible, or when small sample sizes might randomly lead to non-representative samples, various methods of nonprobability sampling have been proposed with the goal of making generalizable inferences~\cite{Tipton13,EgamiLe23,Bouyamourn25}.
In contrast to our work, these methods are all focused on a single round of experimentation, and use objectives focused on the external validity of inferences rather than the quality of decision-making.

Another approach to achieving external validity is to attempt to reweight or extrapolate treatment effects from a study sample to a target population~\cite{TiptonOl18,ZhangHu24,LingMo23}.
In this work we model the continuation criterion using the standard approach to significance testing, where inference is conducted for the treatment effect in the study sample, and we leave to future work the design of optimal pilots when scale-up decisions can utilize out-of-sample prediction or extrapolation.

\subsection{Decision-focused experiment design}

While the notion of external validity is concerned with generalization, the question of \emph{what} to generalize is largely independent.
Much of the statistical analysis and design literature is focused on estimating quantities of interest~\cite{Fisher49,Pukelsheim06} or maximizing some measure of information gain~\cite{ChalonerVe95,RainforthFo23}.
In contrast, our work is part of a large body of work, dating back at least to~\textcite{Blackwell51} but especially rich in recent years, that focuses on the value of information in shaping the outcomes of \emph{decisions} instead of the accuracy of estimates, understanding that estimation and decision-making can require significantly different information~\cite{Fernandez-LoriaPr22,Eckles22}.

Much of the decision-focused experimental design literature is focused on Bayesian decision theory, where the design of the experiment and the final decision are jointly optimized over a prior belief distribution~\cite{RyanDr16,NgIm26}.
This approach has been used to derive structural insights into optimal experimentation policies~\cite{RadnerSt84,DeLaraGi07,AzevedoDe20}, and to make policy recommendations for allocating resources to experiments~\cite{FenwickSt20}.
It has also been used to develop specific experiment designs, as in~\textcite{Heath22,FlightJu22,ChenWi13,FrazierPo09,GechterHi24}.
\textcite{GechterHi24} is a work in this vein that, like ours, optimizes the sample composition of an experiment under a prior distribution on heterogeneous treatment effect.
It differs from our work by assuming that adoption decisions are made in a Bayes-optimal manner following the experiment based on the posterior distribution.
Unfortunately, this type of Bayesian decision-making is not feasible in the common scenarios where the experiment designer and decision maker have different information~\cite{BatesJo22,BatesJo23,Min23}, or when frequentist statistical methods are required for reasons of robustness, by scientific norms, or by regulations~\cite{WassersteinLa16,BanerjeeCh20,Recht25}.
Our work differs from standard Bayes-optimal in that it takes into account the widespread use of frequentist significance testing as an evaluation criterion, while still allowing for optimizing the experiment design using prior information.
We ask:~\emph{how do current statistical testing norms affect optimal experiment design?}

Unlike the Bayesian approach which requires a prior distribution, minimax-regret experiment design instead considers worst-case decision performance against some set of possible outcomes, and has similarly been used to derive optimal site selection or sample composition designs~\cite{OleaPr25,HuZh24}.
Similarly to our work, \textcite{HuZh24} studies how to allocate sampling effort across subpopulations in a randomized experiment to optimize downstream welfare under heterogeneous effects.
In contrast to our approach, they present a solution under a minimax-regret criterion, and for a single experiment stage where adoption may differ across groups. Their design does not depend qualitatively on the budget, which we believe reflects the fact that the minimax-optimal adoption criterion they model does not require statistical significance.

\subsection{Sequential experimentation}

Like decision-focused experiment design, the design of \emph{sequential} experiments has been studied for a long time~\cite{Wald47}, but has recently received special focus as computational methods have made their optimization more feasible~\cite{CheNa23,ChapelleLi11,TrellaZh22} and domains like online platforms have created more opportunities for real-world adaptive experimentation~\cite{AzevedoDe20,XuDu18}.
Such sequential or adaptive experiments are promising in their ability to improve over static designs in terms of the tradeoff between experiment size or duration and quality, as well as by controlling risk during the experiment and other objectives.

We focus on what we believe is a first-order effect in many multi-stage experiments: \emph{the screening effect of early-stage experiments on what interventions receive further study}.
This screening effect has been studied in the context of designing phase II / efficacy trials in drug development, with various works optimizing some combination of sample sizes for phase II and phase III, the decision rule to continue to phase III, and adaptive stopping and sample sizes within stage II~\cite{RossellMu07,DingRo08,KirchnerKi16,ErdmannKi20}.
We are not aware of any work optimizing the sample composition, i.e.\ the distribution of covariates, in a screening trial under a decision- or utility-focused objective.

While adaptivity is helpful in many settings, in many domains it is not clear how or even if the design of a follow-up RCT depends on the outcomes of a previous experiment.
For example, efficiency trials in medical and education research often must use representative or otherwise constrained designs, and are often run by different researchers and institutions than pilot studies, limiting the range of possible adaptivity~\cite{FDA19,DepartmentofEducationBo22}.

As we describe in \Cref{sec:Model}, we fix all the decision rules downstream of the pilot study including the continuation and adoption criteria (statistical significance tests) and the follow-up design (a fixed sampling design independent of the pilot outcome, for example a uniform random sample with a given size).
While we model a multi-stage experimentation process, our optimization problem only involves a single decision---the pilot design.

%% file: paper/07_conclusion.tex
\section{Conclusion}
\label{sec:conclusion}

A critical role played by pilot studies in the process of scientific discovery is helping to determine which interventions deserve further study for potential adoption.
Motivated by this screening function and the tight budgets that such pilots often face, we argue that while representative samples are optimal in the limit for large pilots, when pilots are very resource-constrained it is optimal to design the study around a very targeted sample.
In a model with a finite number of unit types, we prove under mild conditions that for all pilot budgets below a threshold, the optimal pilot design samples only a single type.
We show that while this rule reduces to convenience sampling when the unknown treatment effects are the same across all types, under treatment effect heterogeneity the optimal pilot design prioritizes sampling from types with larger prior mean and that are more predictive of downstream effects.
Finally, we simulate a pilot study calibrated to real education experiments to illustrate that both the large- and small-budget pilot regimes can arise at realistic sample sizes.
Taken together, these findings add nuance to the common narrative that pilot studies should strive to enroll more representative samples; in fact, this can be counterproductive when budgets are tight.

A number of the work's limitations are focused on the assumptions around the cost and feasibility of sampling.
In reality, samples cannot usually be drawn uniformly at random from different subpopulations~\cite{Tipton14}, and costs are more complex than a linear function of sample size~\cite{TiptonSp21}.
Modeling more realistic mechanisms like fixed costs for access to different subpopulations, or the dynamic process of enrolling units into the study, could potentially extend the insights a model like this could provide.

Another promising direction for future work is to consider alternate continuation criteria.
While fully Bayesian approaches to decision-making have been studied, we believe that restricting the class of decision rules to those that are implementable and acceptable in science is necessary when studying pilot studies in that context.
Frequentist decision rules that reweight or extrapolate treatment effects from the study sample to the target population are one promising option, as are combining ordinary significance tests with measures of generalizability~\cite{Tipton14}.

%% file: paper/08_acknowledgements.tex
\phantomsection
\section*{Acknowledgments}
\addcontentsline{toc}{section}{Acknowledgements}
\label{sec:acknowledgements}

We are grateful to Elizabeth Tipton, Naoki Egami, and Molly Offer-Westort for valuable discussions and insights into social science methodology and experimental practice.

%% file: paper/A1_alternate_objectives.tex
\section{Alternate Objectives}
\label{sec:alternate-objectives-appendix}

In this section, we briefly discuss the relationship between the small-budget optimal design under our primary impact objective---the improvement in outcomes over two experiment stages---and three alternative objectives given by either (i) assuming only a \emph{single stage}, and/or (ii) maximizing the \emph{probability of passing} the statistical test or tests instead of the improvement in outcomes.
See \Cref{tab:small-budget-alternate-objectives} for the objective definitions.
These objectives serve both to provide a sensitivity analysis of our optimal small-budget design to the specified objective, and to highlight how various objectives of experiment designers or stakeholders may be more or less aligned in terms of their optimal experiment designs.

\begin{table}[h]
  \centering
  \small
  \renewcommand{\arraystretch}{1.2}
  \begin{tabular}{@{} l @{\hspace{0.75em}} c @{\hspace{0.9em}} c @{\hspace{1.0em}} c c c @{}}
    \multicolumn{1}{c}{} &
    \multicolumn{1}{c}{Objective} &
    \multicolumn{1}{c}{Small-budget index} &
    \multicolumn{3}{c}{Components in elliptical decomposition} \\
    \multicolumn{1}{c}{} &
    \multicolumn{1}{c}{$V_Y(\bm{s}_1)$} &
    \multicolumn{1}{c}{$\pi^Y_d$} &
    $\mu_d$ &
    $\cov(\tau_d,\ATE_2)$ &
    $\cov(\tau_d,\ATE_3)$ \\
    \noalign{\vskip 2pt}
    \hline
    \noalign{\vskip 4pt}
    \makecell[l]{Single-stage\\success} &
    $\E\left[\Test_1\right]$ &
    $\dfrac{\E\left[\tau_d\right]}{\sqrt{c_d}}$ &
    $\checkmark$ & & \\[12pt]
    \makecell[l]{Two-stage\\success} &
    $\E\left[\Test_1 \cdot \Test_2\right]$ &
    $\dfrac{\E\left[\tau_d \cdot \Test_2\right]}{\sqrt{c_d}}$ &
    $\checkmark$ & $\checkmark$ & \\[12pt]
    \makecell[l]{Single-stage\\impact} &
    $\E\left[\Test_1 \cdot \ATE_3\right]$ &
    $\dfrac{\E\left[\tau_d \cdot \ATE_3\right]}{\sqrt{c_d}}$ &
    $\checkmark$ & & $\checkmark$ \\[12pt]
    \makecell[l]{Two-stage\\impact} &
    $\E\left[\Test_1 \cdot \Test_2 \cdot \ATE_3\right]$ &
    $\dfrac{\E\left[\tau_d \cdot \Test_2 \cdot \ATE_3\right]}{\sqrt{c_d}}$ &
    $\checkmark$ & $\checkmark$ & $\checkmark$
  \end{tabular}
  \caption{Objective definitions, their small-budget indices, and the mean and covariance-predictivity components that appear when the corresponding coefficient is decomposed under an elliptical prior.
    A checkmark indicates structural dependence in the decomposition, not a fixed sign or a guarantee that the coefficient is nonzero for every parameter value.
  If an aggregate effect is degenerate, the corresponding covariance coordinate is identically zero for all types.}
  \label{tab:small-budget-alternate-objectives}
\end{table}

Analogues of the small-budget optimal pilot design \Cref{cor:optimal-small-budget-pilot} hold under general priors with moment constraints for any of the alternate objectives in \Cref{tab:small-budget-alternate-objectives}: that is, whenever the corresponding small-budget index has a unique strictly positive maximizer, the optimal small-budget design invests the entire budget in a single type given by the index.
The general index formula can be derived from \linkedCref{thm:optimal-small-budget-design-general-appendix}.
Further, when the prior is elliptical with finite second moments, these small-budget indices can be represented as cost-adjusted combinations of the type prior mean and the type's predictivity for the follow-up and target-population average effects.
\Cref{cor:named-objective-elliptical-decompositions-appendix} gives the coefficient decompositions underlying the checkmarks in \Cref{tab:small-budget-alternate-objectives}.
The checkmarks give the structure of the result---at special parameter values a listed component can have zero coefficient.

Note that the two-stage impact objective we focus on is the only one whose elliptical-prior decomposition involves both follow-up-effect predictivity and target-population-effect predictivity.
These two covariance terms are equal when the follow-up is representative, so in this case two-stage impact is similar to both the single-stage impact and two-stage success objectives, although the weights are different.
Meanwhile, optimizing for passing the pilot statistical test only depends on each type's cost and prior mean

%% file: paper/A2_notation_for_proofs.tex
\section{Notation, Definitions, and General Assumptions for Proofs}
\label{app:notation-for-proofs}

We recall some definitions and define additional notation used in the proofs.
Additional terms are defined in the appendix sections in which they are used.

The heterogeneous treatment effects $\bm\tau$ are drawn from the prior distribution $\mathcal{P}$.

Conditional on the heterogeneous treatment effects $\bm\tau$, the stage-$t$ average treatment effect is
\[
  \ATE_t = \frac{\bm s_t^\top \bm\tau}{n_t},
  \qquad t\in\{1,2,3\},
\]
and the pilot and follow-up estimates satisfy
\[
  \widehat{\ATE}_t = \ATE_t+\frac{\varepsilon_t}{\sqrt{n_t}},
  \qquad
  \varepsilon_t\sim\mathcal N(0,1),
  \qquad
  t\in\{1,2\},
\]
with the stage-specific sampling noises mutually independent of each other and of $\bm\tau$.
We also write the conditional pass probability for the significance test at stage $t$ as
\[
  p_t(x)\defeq \Phi\!\big(\sqrt{n_t}\,x-z_{1-\alpha_t}\big)\in(0,1),
  \qquad t\in\{1,2\},
\]
so that $p_t(\ATE_t)=\P(\Test_t=1\mid \bm\tau)$ almost surely. Here $\Phi$ represents the normal cumulative distribution function. We will write $\varphi$ for the normal density, and $\phi$ for the generating characteristic function of a generic elliptical distribution.

Define the four first-stage objectives discussed in \Cref{sec:alternate-objectives-appendix} as
\begin{align*}
  V_{\singlesuccess}(\bm s_1) &\defeq \E[\Test_1],\\
  V_{\twosuccess}(\bm s_1) &\defeq \E[\Test_1\cdot\Test_2],\\
  V_{\singleimpact}(\bm s_1) &\defeq \E[\Test_1\cdot \ATE_3],\\
  V_{\twoimpact}(\bm s_1) &\defeq \E[\Test_1\cdot\Test_2\cdot \ATE_3].
\end{align*}
where throughout much of the paper body we write the main two-stage impact objective as
\[
  V(\bm s_1)\defeq V_{\twoimpact}(\bm s_1).
\]
For the two-stage impact objective define the oracle impact $V_{\oracle}$, the noise-free pilot impact $V_{\noisefree}$, and the unconditional follow-up impact $V_{\followup}$:
\begin{align*}
  V_{\oracle}
  &\defeq \E\!\left[\indic{\ATE_3>0}\cdot\Test_2\cdot \ATE_3\right],\\
  V_{\noisefree}(\bm s_1)
  &\defeq \E\!\left[\indic{\ATE_1>0}\cdot\Test_2\cdot \ATE_3\right],\\
  V_{\followup}
  &\defeq \E\!\left[\Test_2\cdot \ATE_3\right].
\end{align*}

For a random downstream payoff $Y$, define the corresponding first-stage objective
\[
  V_Y(\bm s_1)\defeq \E[\Test_1\cdot\, Y].
\]
We will say that $Y$ is \emph{non-adaptive} if it is independent of the pilot-stage noise conditional on $\bm\tau$:
\[
  Y \ind \varepsilon_1 \mid \bm\tau.
\]
For such $Y$, define
\[
  \theta_d^Y\defeq \E[\tau_d\cdot Y],
  \qquad
  \pi_d^Y\defeq \frac{\theta_d^Y}{\sqrt{c_d}}
\]
and write $\bm\theta^Y = (\theta_d^Y)_{d=1}^D$.

The four named objectives correspond to the downstream payoffs
\[
  Y_{\singlesuccess}\defeq 1,
  \qquad
  Y_{\twosuccess}\defeq \Test_2,
  \qquad
  Y_{\singleimpact}\defeq \ATE_3,
  \qquad
  Y_{\twoimpact}\defeq \Test_2\cdot \ATE_3,
\]
all of which are non-adaptive.
For the main two-stage impact objective we retain the shorthand
\[
  \theta_d \defeq \theta_d^{\twoimpact}
  = \E[\tau_d\cdot\Test_2\cdot \ATE_3],
  \qquad
  \idx_d = \pi_d^{\twoimpact} = \frac{\theta_d}{\sqrt{c_d}},
\]
and write $\bm\theta = (\theta_d)_{d=1}^D$.

Recall the definition of the set of unit-investment normalized pilot designs from \Cref{sec:Model},
\[
  \overline{\mathcal{S}} \defeq \{\bm x\in\R_{\ge0}^D:\bm c^\top\bm x=1\},
\]
and define the unit-investment representative pilot design as
\[
  \overline{\bm s}_1^{\,\infty}
  \defeq
  \frac{\bm s_3}{\bm c^\top \bm s_3}\in\overline{\mathcal{S}}.
\]

%% file: paper/A3_small_budget_first_stage_objectives.tex
\section{Small-Budget Theory for Non-Adaptive Payoffs}
\label{sec:small-budget-first-stage-objectives-appendix}

This section treats the small-budget problem, in \Cref{sec:generic-first-stage-results-appendix} for non-adaptive first-stage objectives of the form $V_Y(\bm s_1)=\E[\Test_1\cdot\, Y]$, and in \Cref{sec:named-first-stage-results-appendix} for the named objectives from the paper by specialization.

\subsection{Generic first-stage results}
\label{sec:generic-first-stage-results-appendix}

The following gives the proof of the generic small-investment expansion.
\refstepcounter{proposition}
\begin{linkedresult}{proof}{prop:small-investment-expansion}
  \label{prop:small-investment-expansion-appendix}
  Fix a prior distribution of $\bm\tau$, and let $Y$ be non-adaptive with $\E[|Y|]<\infty$ and $\E[|Y|\|\bm\tau\|_2^3]<\infty$.
  Then there exist constants $\overline{\invst}>0$ and $0<C_0,C_1<\infty$ such that for all pilot investments $\invst\in(0,\overline{\invst}]$ and all unit-investment samples $\overline{\bm{s}}_1\in\overline{\mathcal{S}}$,
  \[
    V_Y(\invst\,\overline{\bm{s}}_1)
    =
    \alpha_1 \E[Y]
    +\sqrt{\invst}\,\varphi(z_{1-\alpha_1})\,
    \frac{(\bm{\theta}^Y)^{\!\top}\overline{\bm{s}}_1}{\sqrt{\bm{1}^{\top}\overline{\bm{s}}_1}}
    +R(\invst,\overline{\bm{s}}_1).
  \]

  The remainder satisfies the following for all $\invst, \invst'\in(0,\overline{\invst}]$ and all $\overline{\bm{s}}_1,\overline{\bm{s}}_1'\in\overline{\mathcal{S}}$.
  It is bounded as
  \[
    \left|R(\invst,\overline{\bm{s}}_1)\right|
    \le C_0\,\invst,
  \]
  it is Lipschitz in the investment $\invst$
  \[
    \left|R(\invst,\overline{\bm{s}}_1)-R(\invst',\overline{\bm{s}}_1)\right|
    \le C_0\,|\invst-\invst'|,
  \]
  and it is $O(\invst)$-Lipschitz in the unit-investment design
  \[
    \left|R(\invst,\overline{\bm{s}}_1)-R(\invst,\overline{\bm{s}}_1')\right|
    \le C_1\,\invst\,\|\overline{\bm{s}}_1-\overline{\bm{s}}_1'\|_1.
  \]
\end{linkedresult}

\begin{proof}
  First, we give a smooth extension and reparametrization of the first-stage objective around $\invst=0$.
  Then we use a Taylor expansion to calculate the first two terms.
  Finally, we bound the remainder term and show it has the desired properties.

  \emph{Step 1: A $C^{\infty}$ reparametrization around $\invst=0$.}
  Reparametrize the pilot sample $\bm{s}_1 = j^2 \overline{\bm{s}}_1$ with $j\in\R$ and $\overline{\bm{s}}_1\in \overline{\overline{\mathcal{S}}^{\epsilon}}$ where $\overline{\overline{\mathcal{S}}^{\epsilon}}\supset \overline{\mathcal{S}}^{\epsilon}$ are closed and open $\epsilon$-neighborhoods respectively of the set of unit-investment samples $\overline{\mathcal{S}}$ within $\R_{\ge 0}^D$, with $\epsilon = 1 / (2 \norm{\bm{c}}_2) > 0$.
  For $\overline{\bm{s}}_{1}\in\overline{\overline{\mathcal{S}}^{\epsilon}}$ define
  \[
    \overline{n}_{1} \defeq \bm{1}^{\top}\overline{\bm{s}}_{1}.
  \]
  Note that for all $\overline{\bm{s}}_{1}\in\overline{\overline{\mathcal{S}}^{\epsilon}}$ we can write $\overline{\bm{s}}_1 = \overline{\bm{s}}_1' + \bm{\epsilon}$ for some $\overline{\bm{s}}_1'\in\overline{\mathcal{S}}$ and $\bm{\epsilon}$ with $\norm{\bm{\epsilon}}_2 \le \epsilon$.
  Therefore
  \[
    \bm{c}^\top \overline{\bm{s}}_1 = 1 + \bm{c}^\top \bm{\epsilon}
    \ge 1 - \norm{\bm{c}}_2 \epsilon = 1/2.
  \]
  Since $\overline{\bm{s}}_1\in\R_{\ge 0}^D$, Hölder's inequality gives
  \[
    \left|\bm{c}^\top \overline{\bm{s}}_1\right| \le
    \norm{\bm{c}}_\infty
    \norm{\overline{\bm{s}}_1}_1,
    \qquad
    \left|\bm{c}^\top \overline{\bm{s}}_1\right| \le \norm{\bm{c}}_2
    \norm{\overline{\bm{s}}_1}_2.
  \]
  Combining these, we have
  \[
    1/4 \le \norm{\bm{c}}_2^2 \norm{\overline{\bm{s}}_1}_2^2,
    \qquad
    1/2 \le \norm{\bm{c}}_\infty \norm{\overline{\bm{s}}_1}_1.
  \]
  Thus
  \[
    \overline{n}_{1}
    =
    \norm{\overline{\bm{s}}_{1}}_1
    \ge
    \frac{1}{2 \norm{\bm{c}}_\infty}
    > 0.
  \]

  Define $\xi:\overline{\overline{\mathcal{S}}^{\epsilon}}\times\R^D\to\R$ and $q:\R\times\overline{\overline{\mathcal{S}}^{\epsilon}}\times\R^D\to[0,1]$ as
  \[
    \xi(\overline{\bm{s}}_1,\bm\tau)
    \defeq
    \frac{\overline{\bm{s}}_1^\top\bm\tau}{\sqrt{\overline n_1}},
    \qquad
    q(j,\overline{\bm{s}}_1,\bm\tau)
    \defeq
    \Phi\big(j\,\xi(\overline{\bm{s}}_1,\bm\tau)-z_{1-\alpha_1}\big),
  \]
  such that if $j>0$ and $\bm s_1=j^2\overline{\bm s}_1$, then
  \[
    q(j,\overline{\bm{s}}_1,\bm\tau)
    =
    \P(\Test_1=1\mid \bm\tau).
  \]
  Since $Y$ is assumed to be non-adaptive---that is, independent of the pilot-stage sampling noise conditional on $\bm\tau$---iterated expectations gives
  \begin{align*}
    V_Y(j^2\overline{\bm s}_1)
    & =
    \E\!\left[\Test_1 \cdot \, Y\right] \\
    & =
    \E\!\left[\E\!\left[\Test_1 \mid \bm\tau\right]\cdot \, Y\right] \\
    & =
    \E\!\left[q(j,\overline{\bm{s}}_1,\bm\tau)\,Y\right].
  \end{align*}
  Define the composite map $g:\R\times\overline{\overline{\mathcal{S}}^{\epsilon}}\to\R$ as
  \[
    g(j,\overline{\bm s}_1)
    \defeq
    \E\!\left[q(j,\overline{\bm{s}}_1,\bm\tau)\,Y\right]
  \]
  so that $V_Y(j^2\overline{\bm s}_1)=g(j,\overline{\bm s}_1)$ for $j>0$.

  Since $\overline{\overline{\mathcal{S}}^\epsilon}$ is compact, there exists $M<\infty$ such that
  \[
    \|\overline{\bm s}_1\|_2\le M
    \qquad\forall\overline{\bm s}_1\in \overline{\overline{\mathcal{S}}^\epsilon}.
  \]
  Therefore
  \[
    |\xi(\overline{\bm{s}}_1,\bm\tau)|
    \le
    \frac{M}{\sqrt{\overline n_1}}\|\bm\tau\|_2
    \le
    M\sqrt{2\|\bm c\|_\infty}\,\|\bm\tau\|_2.
  \]
  Also, differentiating $\xi$ with respect to $\overline{\bm s}_1$ gives
  \[
    \nabla_{\overline{\bm s}_1}\xi(\overline{\bm s}_1,\bm\tau)
    =
    \frac{\bm\tau}{\sqrt{\bm 1^\top\overline{\bm s}_1}}
    -
    \frac{\overline{\bm s}_1^\top\bm\tau}{2(\bm 1^\top\overline{\bm s}_1)^{3/2}}\,\bm 1,
  \]
  and hence, using $\|\bm 1\|_\infty=1$,
  \[
    \left\|\nabla_{\overline{\bm s}_1}\xi(\overline{\bm s}_1,\bm\tau)\right\|_\infty
    \le
    \left(\sqrt{2\|\bm c\|_\infty}+\sqrt{2}M\|\bm c\|_\infty^{3/2}\right)\|\bm\tau\|_2.
  \]

  Fix arbitrary $\overline j>0$, and write $\overline{\overline{\mathcal{S}}^{\epsilon/2}}$ for the closed $\epsilon/2$-neighborhood of $\overline{\mathcal{S}}$.
  The first two derivatives of $q$ with respect to $j$ are
  \[
    \frac{\partial}{\partial j}q(j,\overline{\bm s}_1,\bm\tau)
    =
    \xi(\overline{\bm s}_1,\bm\tau)\,
    \varphi\!\big(j\,\xi(\overline{\bm s}_1,\bm\tau) - z_{1-\alpha_1}\big),
  \]
  and
  \[
    \frac{\partial^2}{\partial j^2}q(j,\overline{\bm s}_1,\bm\tau)
    =
    - \xi(\overline{\bm s}_1,\bm\tau)^2
    \big(j\,\xi(\overline{\bm s}_1,\bm\tau) - z_{1-\alpha_1}\big)
    \varphi\!\big(j\,\xi(\overline{\bm s}_1,\bm\tau) - z_{1-\alpha_1}\big).
  \]
  Since $\varphi$, $x\varphi(x)$, and $x^2\varphi(x)$ are bounded on $\R$ ($\varphi$ is bounded and its tails decay exponentially, i.e. faster than any polynomial), there exists a constant $C<\infty$ depending only on $\overline j$, $\bm c$, $M$, and $z_{1-\alpha_1}$ such that for all $(j,\overline{\bm s}_1,\bm\tau)\in[-\overline j,\overline j]\times \overline{\overline{\mathcal{S}}^{\epsilon}}\times\R^D$,
  \[
    \left|\frac{\partial}{\partial j}q(j,\overline{\bm s}_1,\bm\tau)\right|
    \le
    C\,\|\bm\tau\|_2,
    \qquad
    \left|\frac{\partial^2}{\partial j^2}q(j,\overline{\bm s}_1,\bm\tau)\right|
    \le
    C\,\|\bm\tau\|_2^2.
  \]
  Writing
  \[
    h_j(x)\defeq - x^2(jx - z_{1-\alpha_1})\varphi(jx - z_{1-\alpha_1}),
  \]
  we have
  \[
    \frac{\partial^2}{\partial j^2}q(j,\overline{\bm s}_1,\bm\tau)
    =
    h_j\!\big(\xi(\overline{\bm s}_1,\bm\tau)\big),
  \]
  and the chain rule gives
  \[
    \nabla_{\overline{\bm s}_1}\frac{\partial^2}{\partial j^2}q(j,\overline{\bm s}_1,\bm\tau)
    =
    h_j'\!\big(\xi(\overline{\bm s}_1,\bm\tau)\big)\,
    \nabla_{\overline{\bm s}_1}\xi(\overline{\bm s}_1,\bm\tau).
  \]
  Since
  \[
    h_j'(x)
    =
    -2x(jx-z_{1-\alpha_1})\varphi(jx-z_{1-\alpha_1})
    -jx^2\varphi(jx-z_{1-\alpha_1})
    +jx^2(jx-z_{1-\alpha_1})^2\varphi(jx-z_{1-\alpha_1}),
  \]
  and $\varphi$, $x\varphi(x)$, and $x^2\varphi(x)$ are bounded, there exists $C'<\infty$ such that
  \[
    |h_j'(x)|\le C'(|x|+x^2)
    \qquad
    \forall (j,x)\in[-\overline j,\overline j]\times\R.
  \]
  Combining this with the bounds on $\xi$ and $\nabla_{\overline{\bm s}_1}\xi$ yields
  \[
    \left\|\nabla_{\overline{\bm s}_1}\frac{\partial^2}{\partial j^2}q(j,\overline{\bm s}_1,\bm\tau)\right\|_\infty
    \le
    C''\left(\|\bm\tau\|_2^2+\|\bm\tau\|_2^3\right)
  \]
  for a constant $C''<\infty$.

  By the integrability conditions assumed on $Y$ and $\bm{\tau}$,
  \[
    2\left(\mathbb{E}[|Y|]+\mathbb{E}\!\left[|Y|\|\boldsymbol{\tau}\|_2^3\right]\right)
    =
    \E\left[2|Y|\left(1+\|\boldsymbol{\tau}\|_2^3\right)\right]
    <
    \infty.
  \]
  Also, for \(x\ge 0\),
  \[
    x,\ x^2,\ x^2+x^3 \le 2(1+x^3).
  \]
  Thus, applying this with \(x=\|\boldsymbol{\tau}\|_2\) and multiplying by \(|Y|\), the random variables
  \[
    |Y|\,\|\boldsymbol{\tau}\|_2,
    \qquad
    |Y|\,\|\boldsymbol{\tau}\|_2^2,
    \qquad
    |Y|\left(\|\boldsymbol{\tau}\|_2^2+\|\boldsymbol{\tau}\|_2^3\right)
  \]
  are integrable.

  Moreover, for each fixed $\bm\tau$, the maps
  \[
    (j,\overline{\bm s}_1)\mapsto q(j,\overline{\bm s}_1,\bm\tau),
    \quad
    \frac{\partial}{\partial j}q(j,\overline{\bm s}_1,\bm\tau),
    \quad
    \frac{\partial^2}{\partial j^2}q(j,\overline{\bm s}_1,\bm\tau),
    \quad
    \nabla_{\overline{\bm s}_1}\frac{\partial^2}{\partial j^2}q(j,\overline{\bm s}_1,\bm\tau)
  \]
  are continuous on $[-\overline j,\overline j]\times \overline{\overline{\mathcal{S}}^\epsilon}$.
  Therefore, we can differentiate
  \[
    g(j,\overline{\bm{s}}_1) = \E[q(j,\overline{\bm{s}}_1,\bm{\tau}) Y]
  \]
  under the expectation~\cite[Theorem~2.27]{Folland99}, first in $j$ and then componentwise in $\overline{\bm s}_1$,
  \[
    \frac{\partial}{\partial j}g(j,\overline{\bm s}_1)
    =
    \E\!\left[
      Y\,
      \frac{\partial}{\partial j}q(j,\overline{\bm s}_1,\bm\tau)
    \right],
    \qquad
    \frac{\partial^2}{\partial j^2}g(j,\overline{\bm s}_1)
    =
    \E\!\left[
      Y\,
      \frac{\partial^2}{\partial j^2}q(j,\overline{\bm s}_1,\bm\tau)
    \right]
  \]
  for $(j,\overline{\bm s}_1)\in(-\overline j,\overline j)\times \overline{\mathcal{S}}^\epsilon$, and
  \[
    \nabla_{\overline{\bm s}_1}\frac{\partial^2}{\partial j^2}g(j,\overline{\bm s}_1)
    =
    \E\!\left[
      Y\,
      \nabla_{\overline{\bm s}_1}\frac{\partial^2}{\partial j^2}q(j,\overline{\bm s}_1,\bm\tau)
    \right]
  \]
  on the same set.
  All of these derivatives are continuous by dominated convergence: the preceding bounds give parameter-uniform integrable envelopes for the integrands, while the lower bound on \(\bm 1^\top\overline{\bm s}_1\) ensures pointwise continuity in \((j,\overline{\bm s}_1)\).

  Finally, since
  \[
    \mathcal D\defeq[0,\overline j / 2]\times \overline{\overline{\mathcal{S}}^{\epsilon/2}}
  \]
  is a compact subset of $(-\overline j,\overline j)\times\overline{\mathcal{S}}^{\epsilon}$, we may define
  \[
    C_0\defeq\frac{1}{2}\sup_{(j,\overline{\bm s}_1)\in\mathcal D}\left|
    \frac{\partial^2}{\partial j^2}g(j,\overline{\bm s}_1)
    \right|<\infty,
    \qquad
    C_1\defeq\frac{1}{2}\sup_{(j,\overline{\bm s}_1)\in\mathcal D}\left\|
    \nabla_{\overline{\bm s}_1}\frac{\partial^2}{\partial j^2}g(j,\overline{\bm s}_1)
    \right\|_\infty<\infty.
  \]

  \emph{Step 2: Constant term and first derivative at $j=0$.}
  At $j=0$, for any $\overline{\bm s}_1\in\overline{\mathcal{S}}$, we have $q(0,\overline{\bm s}_1,\bm\tau)=1 - \Phi(z_{1-\alpha_1})=\alpha_1$, so
  \[
    g(0,\overline{\bm s}_1)=\alpha_1\,\E[Y].
  \]
  Evaluating the formula above for $\partial g/\partial j$ at $j=0$ and plugging in the definition of $\bm{\theta}^Y=\E[\bm\tau Y]$ gives
  \begin{align*}
    \frac{\partial}{\partial j}g(0,\overline{\bm s}_1)
    &=
    \E\!\left[
      \frac{\partial}{\partial j}q(j,\overline{\bm s}_1,\bm\tau)\Big|_{j=0}\,Y
    \right] \\
    &=
    \varphi(z_{1-\alpha_1})\,
    \frac{1}{\sqrt{\bm 1^\top\overline{\bm s}_1}}\,
    \E\!\left[
      (\overline{\bm s}_1^\top\bm\tau)\,Y
    \right] \\
    &=
    \varphi(z_{1-\alpha_1})\,
    \frac{(\bm{\theta}^Y)^\top\overline{\bm s}_1}{\sqrt{\bm 1^\top\overline{\bm s}_1}}.
  \end{align*}

  \emph{Step 3: Taylor expansion with an integral remainder and remainder bounds.}
  For any $(j,\overline{\bm s}_1)\in\mathcal D$, Taylor's theorem with an integral remainder~\cite[Theorem~2.55]{Folland23} gives
  \[
    g(j,\overline{\bm{s}}_1)
    =
    g(0,\overline{\bm{s}}_1)
    +j\,\frac{\partial}{\partial j}g(0,\overline{\bm{s}}_1)
    +\int_0^j (j-x)\,\frac{\partial^2}{\partial j^2}
    g(x,\overline{\bm{s}}_1)\,dx.
  \]
  Let $\overline{\invst}\defeq (\overline{j} / 2)^2$.
  Plugging in the values computed above, for any $\invst\in(0,\overline{\invst}]$ and $\overline{\bm s}_1\in\overline{\mathcal{S}}$,
  \[
    V_Y(\invst\overline{\bm{s}}_1)
    =
    g(\sqrt{\invst},\overline{\bm{s}}_1)
    =
    \alpha_1 \E[Y]
    +\sqrt{\invst}\,\varphi(z_{1-\alpha_1})\,
    \frac{(\bm{\theta}^Y)^{\!\top}\overline{\bm{s}}_1}{\sqrt{\bm{1}^{\top}\overline{\bm{s}}_1}}
    +R(\invst,\overline{\bm{s}}_1),
  \]
  where
  \[
    R(\invst,\overline{\bm{s}}_1)
    \defeq
    \int_0^{\sqrt{\invst}} (\sqrt{\invst}-x)\,\frac{\partial^2}{\partial j^2}
    g(x,\overline{\bm{s}}_1)\,dx.
  \]

  Since $|\frac{\partial^2}{\partial j^2} g|\le 2 C_0$ on $\mathcal D$, the remainder is bounded as
  \[
    |R(\invst,\overline{\bm{s}}_1)|
    \le
    2 C_0 \int_0^{\sqrt{\invst}}(\sqrt{\invst}-x)\,dx
    =
    C_0\,\invst.
  \]

  To show the Lipschitz bound in the investment, fix $\overline{\bm{s}}_1\in\overline{\overline{\mathcal{S}}^{\epsilon/2}}$ and $\invst,\invst'\in(0,\overline{\invst}]$.
  Without loss of generality assume $\invst'\le \invst$.
  Subtracting and splitting the integral at $\sqrt{\invst'}$ gives
  \begin{align*}
    R(\invst,\overline{\bm{s}}_1)-R(\invst',\overline{\bm{s}}_1)
    &=
    (\sqrt{\invst}-\sqrt{\invst'})\int_0^{\sqrt{\invst'}}\frac{\partial^2}{\partial
    j^2}g(x,\overline{\bm{s}}_1)\,dx
    \;+\;
    \int_{\sqrt{\invst'}}^{\sqrt{\invst}}(\sqrt{\invst}-x)\,\frac{\partial^2}{\partial
    j^2}g(x,\overline{\bm{s}}_1)\,dx.
  \end{align*}
  Since $\big|\frac{\partial^2}{\partial j^2}g(x,\overline{\bm{s}}_1)\big|\le 2C_0$ uniformly on $\mathcal{D}$, we can bound
  \begin{align*}
    \big|R(\invst,\overline{\bm{s}}_1)-R(\invst',\overline{\bm{s}}_1)\big|
    &\le
    2C_0\Big((\sqrt{\invst}-\sqrt{\invst'})\sqrt{\invst'}+\int_{\sqrt{\invst'}}^{\sqrt{\invst}}(\sqrt{\invst}-x)\,dx\Big)
    \\
    &=
    2C_0\Big((\sqrt{\invst}-\sqrt{\invst'})\sqrt{\invst'}+\tfrac{1}{2}(\sqrt{\invst}-\sqrt{\invst'})^2\Big)
    \\
    &=
    C_0\,(\invst-\invst').
  \end{align*}
  Therefore, for all $\invst,\invst'\in(0,\overline{\invst}]$ and all $\overline{\bm{s}}_1\in\overline{\overline{\mathcal{S}}^{\epsilon/2}}$,
  \[
    \big|R(\invst,\overline{\bm{s}}_1)-R(\invst',\overline{\bm{s}}_1)\big|
    \le C_0\,|\invst-\invst'|.
  \]

  To show the Lipschitz bound in the unit-investment design, fix $\overline{\bm{s}}_1,\overline{\bm{s}}_1'\in\overline{\overline{\mathcal{S}}^{\epsilon/2}}$ and $x\in(0,\overline{j} / 2]$.
  Note that $\overline{\overline{\mathcal{S}}^{\epsilon/2}}$ is convex because it is the closed $\epsilon/2$-neighborhood of the convex polytope $\overline{\mathcal{S}}$ within $\R_{\ge 0}^D$.
  Therefore, since $\frac{\partial^2}{\partial j^2} g(x,\cdot)$ is continuously differentiable on $\overline{\mathcal{S}}^{\epsilon}\supset \overline{\overline{\mathcal{S}}^{\epsilon/2}}$, the fundamental theorem of calculus along the line segment gives
  \[
    \frac{\partial^2}{\partial j^2}
    g(x,\overline{\bm{s}}_1)-\frac{\partial^2}{\partial j^2}
    g(x,\overline{\bm{s}}_1')
    =
    \int_0^1
    \nabla_{\overline{\bm{s}}_1}\frac{\partial^2}{\partial j^2}
    g\big(x,\overline{\bm{s}}_1'+\lambda(\overline{\bm{s}}_1-\overline{\bm{s}}_1')\big)^{\!\top}
    (\overline{\bm{s}}_1-\overline{\bm{s}}_1')\,d\lambda.
  \]
  Thus, by Hölder's inequality and the bound $\|\nabla_{\overline{\bm{s}}_1}\frac{\partial^2}{\partial j^2} g\|_\infty\le 2 C_1$ on $\mathcal D$, we can bound
  \[
    \left|\frac{\partial^2}{\partial j^2}
    g(x,\overline{\bm{s}}_1)-\frac{\partial^2}{\partial j^2}
    g(x,\overline{\bm{s}}_1')\right|
    \le
    2 C_1\,\|\overline{\bm{s}}_1-\overline{\bm{s}}_1'\|_1.
  \]
  Therefore, for $\invst \le \overline{\invst}$,
  \begin{align*}
    |R(\invst,\overline{\bm{s}}_1)-R(\invst,\overline{\bm{s}}_1')|
    &\le
    \int_0^{\sqrt{\invst}} (\sqrt{\invst}-x)\,
    \left|\frac{\partial^2}{\partial j^2}
    g(x,\overline{\bm{s}}_1)-\frac{\partial^2}{\partial j^2}
    g(x,\overline{\bm{s}}_1')\right|\,dx\\
    &\le
    2 C_1\,\|\overline{\bm{s}}_1-\overline{\bm{s}}_1'\|_1
    \int_0^{\sqrt{\invst}} (\sqrt{\invst}-x)\,dx
    \;=\;
    C_1\,\invst\,\|\overline{\bm{s}}_1-\overline{\bm{s}}_1'\|_1.
  \end{align*}
\end{proof}

The following gives the proof of the generic small-budget optimal design result.

\refstepcounter{theorem}
\begin{linkedresult}{proof}{thm:optimal-small-budget-design-general}
  \label{thm:optimal-small-budget-design-general-appendix}
  Fix a prior distribution of $\bm\tau$ and a payoff $Y$ satisfying the assumptions of
  \linkedCref{prop:small-investment-expansion-appendix}.
  Suppose there exists a unique type $d^\star$ maximizing the small-budget index
  \[
    \pi_d^Y
    \defeq
    \frac{\theta_d^Y}{\sqrt{c_{d}}},
  \]
  and that $\pi_{d^\star}^Y>0$.
  Then there exists $\overline{b}>0$ such that for every budget $b\in(0,\overline{b}]$, a unique optimal design exists for
  \[
    \max_{\substack{\bm s_1\in\R_{\geq 0}^{D} \setminus \{\bm{0}\}\\ \bm c^\top \bm s_1\le b}}
    V_Y(\bm s_1),
  \]
  and it invests the entire budget in the type $d^\star$,
  \[
    \bm{s}_1^\star(b) =
    \frac{b}{c_{d^\star}}\,\bm{e}_{d^\star}.
  \]
\end{linkedresult}

\begin{proof}
  \emph{Basic definitions and outline.}
  Recall from \Cref{sec:Model} that the budget-constrained problem is equivalent to maximizing $V_Y(\invst\,\overline{\bm{s}}_1)$ over $\invst\in(0,b]$ and $\overline{\bm{s}}_1\in\overline{\mathcal{S}}$.
  We proceed in three steps.
  First, we show that the leading-order function
  \[
    L(\overline{\bm{s}}_1)\defeq
    \frac{(\bm{\theta}^Y)^\top
    \overline{\bm{s}}_1}{\sqrt{\bm{1}^\top\overline{\bm{s}}_1}},
    \qquad \overline{\bm{s}}_1\in\overline{\mathcal{S}},
  \]
  in \linkedCref{prop:small-investment-expansion-appendix} is uniquely maximized over $\overline{\mathcal{S}}$ at the vertex corresponding to $d^\star$.
  Second, we show that for sufficiently small fixed investment $\invst$, the unique optimal unit-investment design is the same vertex.
  Third, we show that the optimal design for sufficiently small budgets $b$ invests the entire budget there.

  \emph{Step 1: The leading-order functional $L$ is uniquely maximized over $\overline{\mathcal{S}}$ at the vertex corresponding to $d^\star$ and decreases at least linearly away from it.}
  Define the corresponding unit-investment vertex
  \[
    \overline{\bm{v}}^\star \defeq
    \frac{1}{c_{d^\star}}\,\bm{e}_{d^\star}\in\overline{\mathcal{S}}.
  \]
  Then $L(\overline{\bm{v}}^\star)=\pi_{d^\star}^Y$.
  Define the index gap
  \[
    \Delta\defeq \pi_{d^\star}^Y-\max_{k\neq d^\star}\pi_k^Y \;>\;0,
  \]
  and the cost extrema $c_{\min}\defeq\min_d c_{d}$ and $c_{\max}\defeq\max_d c_{d}$.
  For any $\overline{\bm{s}}_1\in\overline{\mathcal{S}}$ define $\lambda_d\defeq c_{d}\overline{s}_{1d}$, so $\lambda_d\ge0$ and $\sum_d\lambda_d=1$.
  Also, write $w_d\defeq 1/c_{d}$, so that $\overline{s}_{1d}=\lambda_d w_d$.

  Under this new parametrization, we have
  \[
    \bm{1}^\top\overline{\bm{s}}_1
    =
    \sum_{d=1}^D \overline{s}_{1d}
    =
    \sum_{d=1}^D \lambda_d w_d,
    \qquad
    (\bm{\theta}^Y)^\top\overline{\bm{s}}_1
    =
    \sum_{d=1}^D \theta_d^Y\,\overline{s}_{1d}
    =
    \sum_{d=1}^D \lambda_d\,\theta_d^Y\,w_d.
  \]
  Using $\theta_d^Y = \pi_d^Y\sqrt{c_{d}}$ and $w_d=1/c_{d}$ gives $\theta_d^Y w_d = \pi_d^Y\sqrt{w_d}$, hence
  \[
    (\bm{\theta}^Y)^\top\overline{\bm{s}}_1
    =
    \sum_{d=1}^D \lambda_d\,\pi_d^Y\,\sqrt{w_d},
    \qquad\text{and therefore}\qquad
    L(\overline{\bm{s}}_1)
    =
    \frac{\sum_{d=1}^D
    \lambda_d\,\pi_d^Y\,\sqrt{w_d}}{\sqrt{\sum_{d=1}^D \lambda_d w_d}}.
  \]

  By the definition of $\Delta$, for every $d\neq d^\star$ we have $\pi_d^Y\le \pi_{d^\star}^Y-\Delta$, so
  \begin{align*}
    \sum_{d=1}^D \lambda_d\,\pi_d^Y\,\sqrt{w_d}
    &=
    \lambda_{d^\star} \pi_{d^\star}^Y\sqrt{w_{d^\star}}
    +\sum_{d\neq d^\star}\lambda_d\,\pi_d^Y\,\sqrt{w_d} \\
    &\le
    \lambda_{d^\star} \pi_{d^\star}^Y\sqrt{w_{d^\star}}
    +\sum_{d\neq d^\star}\lambda_d\,(\pi_{d^\star}^Y-\Delta)\,\sqrt{w_d} \\
    &=
    \pi_{d^\star}^Y\sum_{d=1}^D \lambda_d\sqrt{w_d}
    -\Delta\sum_{d\neq d^\star}\lambda_d\sqrt{w_d}.
  \end{align*}
  Dividing by $\sqrt{\sum_d \lambda_d w_d}$ yields
  \[
    L(\overline{\bm{s}}_1)
    \le
    \pi_{d^\star}^Y\,
    \frac{\sum_{d=1}^D \lambda_d\sqrt{w_d}}{\sqrt{\sum_{d=1}^D \lambda_d w_d}}
    \;-\;
    \Delta\,
    \frac{\sum_{d\neq
    d^\star}\lambda_d\sqrt{w_d}}{\sqrt{\sum_{d=1}^D \lambda_d w_d}}.
  \]
  By Cauchy--Schwarz,
  \[
    \sum_{d=1}^D \lambda_d\sqrt{w_d}
    =
    \sum_{d=1}^D \sqrt{\lambda_d}\,\sqrt{\lambda_d w_d}
    \le
    \sqrt{\sum_{d=1}^D \lambda_d}\,\sqrt{\sum_{d=1}^D \lambda_d w_d}
    =
    \sqrt{\sum_{d=1}^D \lambda_d w_d},
  \]
  hence the ratio in the first term is at most $1$, and since $\pi_{d^\star}^Y>0$ we obtain
  \[
    L(\overline{\bm{s}}_1)
    \le
    \pi_{d^\star}^Y
    \;-\;
    \Delta\,
    \frac{\sum_{d\neq
    d^\star}\lambda_d\sqrt{w_d}}{\sqrt{\sum_{d=1}^D \lambda_d w_d}}.
  \]
  Noting that $L(\overline{\bm{v}}^\star)=\pi_{d^\star}^Y$, we can rearrange this as
  \[
    L(\overline{\bm{v}}^\star)-L(\overline{\bm{s}}_1)
    \ge
    \Delta\,
    \frac{\sum_{d\neq
    d^\star}\lambda_d\sqrt{w_d}}{\sqrt{\sum_{d=1}^D \lambda_d w_d}}.
  \]
  Since $\min_d \sqrt{w_d}=1/\sqrt{c_{\max}}$ and $\max_d w_d = 1/c_{\min}$,
  \[
    \sum_{d\neq d^\star}\lambda_d\sqrt{w_d}
    \ge
    \frac{1-\lambda_{d^\star}}{\sqrt{c_{\max}}},
    \qquad
    \sqrt{\sum_{d=1}^D \lambda_d w_d}
    \le
    \frac{1}{\sqrt{c_{\min}}}.
  \]
  Combining these gives the uniform bound
  \[
    L(\overline{\bm{v}}^\star)-L(\overline{\bm{s}}_1)
    \ge
    \Delta\sqrt{\frac{c_{\min}}{c_{\max}}}\,(1-\lambda_{d^\star}),
    \qquad \forall\overline{\bm{s}}_1\in\overline{\mathcal{S}}.
  \]
  Finally, we can relate $1-\lambda_{d^\star}$ to the $\ell_1$-distance from the vertex:
  \begin{align*}
    \norm{\overline{\bm{s}}_1-\overline{\bm{v}}^\star}_1
    &=
    \left|\overline{s}_{1d^\star}-\overline{v}_{1d^\star}^\star\right|
    +\sum_{d\neq d^\star}\left|\overline{s}_{1d}-\overline{v}_{1d}^\star\right| \\
    &=
    \left|\overline{s}_{1d^\star}-\frac{1}{c_{d^\star}}\right|
    +\sum_{d\neq d^\star}\left|\overline{s}_{1d}-0\right|
    && \text{since }\overline{\bm{v}}^\star=\bm{e}_{d^\star}/c_{d^\star} \\
    &=
    \frac{1}{c_{d^\star}}-\overline{s}_{1d^\star}
    +\sum_{d\neq d^\star}\overline{s}_{1d}
    && \text{since }c_{d^\star}\overline{s}_{1d^\star}=\lambda_{d^\star}\le 1 \\
    &=
    \frac{1}{c_{d^\star}}-\frac{\lambda_{d^\star}}{c_{d^\star}}
    +\sum_{d\neq d^\star}\frac{\lambda_d}{c_d} \\
    &\le
    \frac{1-\lambda_{d^\star}}{c_{\min}}
    +\sum_{d\neq d^\star}\frac{\lambda_d}{c_{\min}} \\
    &=
    \frac{2}{c_{\min}}\,(1-\lambda_{d^\star}).
  \end{align*}
  Combining this with the previous inequality gives
  \begin{equation}
    \label{eq:generic-leading-gap}
    L(\overline{\bm{v}}^\star)-L(\overline{\bm{s}}_1)
    \ge
    q\,\|\overline{\bm{s}}_1-\overline{\bm{v}}^\star\|_1,
    \qquad
    q\defeq \frac{\Delta
    \left(c_{\min}\right)^{3/2}}{2\left(c_{\max}\right)^{1/2}}.
  \end{equation}
  In particular, if $\overline{\bm{s}}_1\neq \overline{\bm{v}}^\star$ then the bound above is strict, so $\overline{\bm{v}}^\star$ is the unique maximizer of $L$ on $\overline{\mathcal{S}}$.

  \emph{Step 2: For fixed investment $\bm{c}^\top\bm{s}_1=\invst$ small enough, the unique optimal unit-investment design is $\overline{\bm{v}}^\star$.}
  By \linkedCref{prop:small-investment-expansion-appendix}, there exist $\overline{\invst}>0$ and $0<C_1<\infty$ such that for all $\invst\in(0,\overline{\invst}]$ and $\overline{\bm{s}}_1\in\overline{\mathcal{S}}$,
  \[
    V_Y(\invst\overline{\bm{s}}_1)
    =\alpha_1
    \E[Y]+\sqrt{\invst}\,\varphi(z_{1-\alpha_1})\,L(\overline{\bm{s}}_1)+R(\invst,\overline{\bm{s}}_1),
  \]
  with the $O(\invst)$-Lipschitz remainder bound
  \[
    |R(\invst,\overline{\bm{s}}_1)-R(\invst,\overline{\bm{s}}_1')|
    \le C_1\,\invst\,\|\overline{\bm{s}}_1-\overline{\bm{s}}_1'\|_1
    \qquad\forall\overline{\bm{s}}_1,\overline{\bm{s}}_1'\in\overline{\mathcal{S}}.
  \]
  Therefore, for any $\overline{\bm{s}}_1\in\overline{\mathcal{S}}$ and any $\invst\in(0,\overline{\invst}]$,
  \begin{align*}
    V_Y(\invst\overline{\bm{v}}^\star)-V_Y(\invst\overline{\bm{s}}_1)
    &=
    \sqrt{\invst}\,\varphi(z_{1-\alpha_1})\Big(L(\overline{\bm{v}}^\star)-L(\overline{\bm{s}}_1)\Big)
    +\Big(R(\invst,\overline{\bm{v}}^\star)-R(\invst,\overline{\bm{s}}_1)\Big)\\
    &\ge
    \sqrt{\invst}\,\varphi(z_{1-\alpha_1})\Big(L(\overline{\bm{v}}^\star)-L(\overline{\bm{s}}_1)\Big)
    -C_1\,\invst\,\|\overline{\bm{s}}_1-\overline{\bm{v}}^\star\|_1.
  \end{align*}
  Combining this with \eqref{eq:generic-leading-gap} gives
  \[
    V_Y(\invst\overline{\bm{v}}^\star)-V_Y(\invst\overline{\bm{s}}_1)
    \ge
    \Big(\sqrt{\invst}\,\varphi(z_{1-\alpha_1})\,q-C_1
    \invst\Big)\,\|\overline{\bm{s}}_1-\overline{\bm{v}}^\star\|_1.
  \]
  Let
  \[
    \invst_+ \defeq
    \min\!\left\{\overline{\invst},\ \Big(\frac{\varphi(z_{1-\alpha_1})q}{2C_1}\Big)^2\right\}.
  \]
  Then for all $\invst\in(0,\invst_+]$,
  \[
    \sqrt{\invst} \le \sqrt{\invst_+}
    \le \frac{\varphi(z_{1-\alpha_1})q}{2C_1}
  \]
  and the coefficient in parentheses satisfies
  \[
    \sqrt{\invst}\,\varphi(z_{1-\alpha_1})q-C_1 \invst
    \ge
    \sqrt{\invst}\,\varphi(z_{1-\alpha_1})q-C_1 \sqrt{\invst}
    \frac{\varphi(z_{1-\alpha_1})q}{2C_1}
    =
    \tfrac{1}{2}\sqrt{\invst}\,\varphi(z_{1-\alpha_1})q>0,
  \]
  so $V_Y(\invst\overline{\bm{v}}^\star)>V_Y(\invst\overline{\bm{s}}_1)$ for every $\overline{\bm{s}}_1\neq \overline{\bm{v}}^\star$ in $\overline{\mathcal{S}}$.
  Hence $\overline{\bm{v}}^\star$ is the unique maximizer of $\overline{\bm{s}}_1\mapsto V_Y(\invst\overline{\bm{s}}_1)$ on $\overline{\mathcal{S}}$ for all $\invst\in(0,\invst_+]$.

  \emph{Step 3: The optimal spending decision for sufficiently small budgets.}
  Let $\invst_+>0$ be as in \emph{Step 2}.
  Recall that $L(\overline{\bm{v}}^\star)=\pi_{d^\star}^Y>0$, and write
  \[
    A \;\defeq\; \varphi(z_{1-\alpha_1})\,L(\overline{\bm{v}}^\star)
    \;=\;\varphi(z_{1-\alpha_1})\,\pi_{d^\star}^Y \;>\;0.
  \]
  By \linkedCref{prop:small-investment-expansion-appendix}, the remainder is Lipschitz in the investment: for all $\invst,\invst'\in(0,\overline{\invst}]$,
  \[
    \big|R(\invst,\overline{\bm{v}}^\star)-R(\invst',\overline{\bm{v}}^\star)\big|
    \le C_0\,|\invst-\invst'|.
  \]
  Define
  \[
    \overline{b} \defeq
    \min\!\left\{\,\invst_+,\ \overline{\invst},\ \left(\frac{A}{2C_0}\right)^2\right\}.
  \]
  The map $\invst\mapsto V_Y(\invst\,\overline{\bm{v}}^\star)$ is strictly increasing on $(0,\overline{b}]$.
  To see this, fix any $0<\invst<\invst'\le \overline{b}$.
  Directly applying \linkedCref{prop:small-investment-expansion-appendix} with $\overline{\bm{s}}_1=\overline{\bm{v}}^\star$ gives
  \[
    V_Y(\invst'\overline{\bm{v}}^\star)-V_Y(\invst\overline{\bm{v}}^\star)
    =
    A(\sqrt{\invst'}-\sqrt{\invst})
    +\Big(R(\invst',\overline{\bm{v}}^\star)-R(\invst,\overline{\bm{v}}^\star)\Big).
  \]
  By the Lipschitz-in-$\invst$ remainder bound,
  \begin{align*}
    V_Y(\invst'\overline{\bm{v}}^\star)-V_Y(\invst\overline{\bm{v}}^\star)
    &\ge
    A(\sqrt{\invst'}-\sqrt{\invst}) - C_0(\invst'-\invst)
    = (\invst'-\invst)\left(\frac{A}{\sqrt{\invst'}+\sqrt{\invst}}-C_0\right).
  \end{align*}
  Since $\invst<\invst'$ we have $\sqrt{\invst'}+\sqrt{\invst}<2\sqrt{\invst'}$, hence
  \[
    \frac{A}{\sqrt{\invst'}+\sqrt{\invst}}
    >
    \frac{A}{2\sqrt{\invst'}}
    \ge
    \frac{A}{2\sqrt{\overline{b}}}
    \ge C_0,
  \]
  where the last inequality uses $\overline{b}\le (A/(2C_0))^2$.
  Therefore $V_Y(\invst'\overline{\bm{v}}^\star)>V_Y(\invst\overline{\bm{v}}^\star)$.

  Now fix any budget $b\in(0,\overline{b}]$.
  By \emph{Step 2}, for every fixed investment $\invst\in(0,b]$ the unique maximizer of $\overline{\bm{s}}_1\mapsto V_Y(\invst\overline{\bm{s}}_1)$ over $\overline{\mathcal{S}}$ is $\overline{\bm{v}}^\star$.
  Since $\invst\mapsto V_Y(\invst\overline{\bm{v}}^\star)$ is strictly increasing on $(0,b]$, the unique maximizer over $\invst\in(0,b]$ is $\invst=b$.
  Hence the unique optimizer is
  \[
    \bm{s}_1^\star(b)=b\,\overline{\bm{v}}^\star
    =
    \frac{b}{c_{d^\star}}\,\bm{e}_{d^\star}.
  \]
\end{proof}

\subsection{Named objectives under a general prior}
\label{sec:named-first-stage-results-appendix}

\begin{proposition}[Finite fourth moments imply the generic small-investment assumptions for the two-stage impact objective]
  \label{prop:fourth-moment-small-investment-two-impact-appendix}
  Suppose the prior distribution of $\bm\tau$ has finite fourth moments.
  Then the downstream payoff
  \[
    Y_{\twoimpact}=\Test_2\cdot \ATE_3
  \]
  is non-adaptive and satisfies $\E[|Y_{\twoimpact}|]<\infty$ and $\E[|Y_{\twoimpact}|\|\bm\tau\|_2^3]<\infty$.
  Consequently, the non-adaptivity and moment assumptions in \linkedCref{prop:small-investment-expansion-appendix} and \linkedCref{thm:optimal-small-budget-design-general-appendix} hold with $Y=Y_{\twoimpact}$.
\end{proposition}

\begin{proof}
  $Y_{\twoimpact}=\Test_2\cdot \ATE_3$ is non-adaptive because $\Test_2$ and $\ATE_3$ are independent of the pilot-stage sampling noise conditional on $\bm\tau$ by the probability model described in \Cref{sec:Model}.

  The first integrability condition is satisfied because
  \[
    |Y_{\twoimpact}|
    \le
    |\ATE_3|
    =
    \frac{|\bm s_3^\top\bm\tau|}{n_3}
    \le
    \frac{\|\bm s_3\|_2}{n_3}\|\bm\tau\|_2,
  \]
  and the second is satisfied because
  \[
    |Y_{\twoimpact}|\|\bm\tau\|_2^3
    \le
    |\ATE_3|\|\bm\tau\|_2^3
    =
    \frac{|\bm s_3^\top\bm\tau|}{n_3}\|\bm\tau\|_2^3
    \le
    \frac{\|\bm s_3\|_2}{n_3}\|\bm\tau\|_2^4,
  \]
  where we use $0\le\Test_2\le1$, the definition
  $\ATE_3=\bm s_3^\top\bm\tau/n_3$, and Cauchy--Schwarz.
  Since $D<\infty$,
  \[
    \|\bm\tau\|_2^4
    =
    \left(\sum_{d=1}^D \tau_d^2\right)^2
    \le
    D\sum_{d=1}^D \tau_d^4,
  \]
  so finite fourth moments of $\bm\tau$ imply $\E[\|\bm\tau\|_2^4]<\infty$, and hence also $\E[\|\bm\tau\|_2]<\infty$.
  Therefore the right-hand sides are integrable, which proves the claim.
\end{proof}

Applying \linkedCref{prop:small-investment-expansion-appendix} and \linkedCref{thm:optimal-small-budget-design-general-appendix} with $Y=Y_{\twoimpact}=\Test_2 \cdot \ATE_3$ nearly recovers the small-budget results for the paper's main objective, and applying them with $Y=1, \Test_2$, and $\ATE_3$ gives the corresponding results for the alternate objectives.
The only missing component is that \linkedCref{thm:optimal-small-budget-design-general-appendix} requires the maximizing type to have a positive index.
We now show that this positivity condition holds automatically for the two impact objectives $V_{\singleimpact}$ and $V_{\twoimpact}$.

\begin{proposition}[Automatic positivity for the impact objectives]
  \label{prop:positive-impact-objectives-appendix}
  Suppose the prior distribution of $\bm\tau$ has finite fourth moments.
  If there exists a unique type maximizing the small-budget index for $V_{\singleimpact}$, then there exists at least one type $d$ with $\theta_d^{\singleimpact}>0$.
  If there exists a unique type maximizing the small-budget index for $V_{\twoimpact}$, then there exists at least one type $d$ with $\theta_d = \theta^{\twoimpact}_d>0$.
  Consequently, the positivity condition in \linkedCref{thm:optimal-small-budget-design-general-appendix} is automatic for either impact objective whenever its small-budget index has a unique maximizer.

\end{proposition}

\begin{proof}
  To begin, note that if there exists a unique type $d^\star$ maximizing the small-budget index for a payoff $Y$, then
  \[
    \pi_d^Y = \frac{\theta_d^Y}{\sqrt{c_d}} = \frac{\E[\tau_d \cdot Y]}{\sqrt{c_d}}
  \]
  is not constant, and thus $Y$ cannot be almost surely zero.
  Also, the assumption that the prior has finite fourth moments implies that $\E[\ATE_3^2]<\infty$.

  For the single-impact objective, $Y_{\singleimpact}=\ATE_3$, so a unique small-budget index maximizer implies that $\ATE_3$ is not almost surely zero.
  Taking the inner product with $\bm s_3$ gives
  \[
    \bm s_3^\top\bm\theta^{\singleimpact}
    =
    \sum_{d=1}^D s_{3d}\,\E[\tau_d \ATE_3]
    =
    \E[(\bm s_3^\top\bm\tau)\,\ATE_3]
    =
    n_3\,\E[\ATE_3^2].
  \]
  Since $\ATE_3$ is not almost surely zero, we have $\E[\ATE_3^2]>0$, so $\bm s_3^\top\bm\theta^{\singleimpact}>0$.
  Therefore at least one coordinate of $\bm\theta^{\singleimpact}$ is strictly positive.

  For the two-impact objective, $Y_{\twoimpact}=\Test_2\cdot \ATE_3$, so a unique small-budget index maximizer implies that $\Test_2\cdot \ATE_3$ is not almost surely zero, and therefore $\ATE_3$ is not almost surely zero.
  Similarly,
  \[
    \bm s_3^\top\bm\theta^{\twoimpact}
    =
    \sum_{d=1}^D s_{3d}\,\E[\tau_d\,\cdot\,\Test_2\,\cdot\, \ATE_3]
    =
    \E[(\bm s_3^\top\bm\tau)\,\cdot\,\Test_2\,\cdot\, \ATE_3]
    =
    n_3\,\E[\Test_2\,\cdot\, \ATE_3^2].
  \]
  Conditional on $\bm\tau$, $\Test_2$ is Bernoulli with success probability $p_2(\ATE_2)\in(0,1)$.
  Hence
  \[
    \E[\Test_2\,\cdot\, \ATE_3^2]
    =
    \E[p_2(\ATE_2)\,\ATE_3^2]
    >
    0,
  \]
  because $p_2(\ATE_2)>0$ and $\ATE_3$ not almost surely zero implies $\P(\ATE_3\neq0)>0$.
  Thus $\bm s_3^\top\bm\theta^{\twoimpact}>0$, so at least one coordinate of $\bm\theta^{\twoimpact}$ is strictly positive.
\end{proof}

Combining \linkedCref{thm:optimal-small-budget-design-general-appendix} with \Cref{prop:fourth-moment-small-investment-two-impact-appendix} and \Cref{prop:positive-impact-objectives-appendix} gives \Cref{cor:optimal-small-budget-pilot}.

The small-budget indices can be written in the following form for general priors.
In \Cref{sec:elliptical-decompositions-appendix} we show that the two-stage objectives can be further simplified under an elliptical prior.

\begin{proposition}[Representations for the named coefficients under a general prior]
  \label{prop:first-stage-coefficient-formulas-appendix}
  Assume the prior on $\bm\tau$ has finite second moments.
  Then for every type $d$,
  \begin{align*}
    \theta_d^{\singlesuccess}
    &=
    \E[\tau_d] \\
    \theta_d^{\twosuccess}
    &=
    \E[\tau_d\,p_2(\ATE_2)],\\
    \theta_d^{\singleimpact}
    &=
    \E[\tau_d \ATE_3]
    =
    \E[\ATE_3]\,\E[\tau_d]
    +
    \cov(\tau_d,\ATE_3),\\
    \theta_d
    =
    \theta_d^{\twoimpact}
    &=
    \E[\tau_d\,p_2(\ATE_2)\,\ATE_3].
  \end{align*}

\end{proposition}

\begin{proof}
  The single-stage success formula is immediate from the definition of $\theta_d^{\singlesuccess}$.

  Since $p_2(\ATE_2)=\P(\Test_2=1\mid\bm\tau)\in(0,1)$ and $\ATE_3$ is measurable with respect to $\bm\tau$, by iterated expectations we have
  \[
    \theta_d^{\twosuccess}
    =
    \E[\tau_d\,\cdot\,\Test_2]
    =
    \E[\tau_d\,p_2(\ATE_2)],
    \qquad
    \theta_d
    =
    \E[\tau_d\,\cdot\,\Test_2\,\cdot\, \ATE_3]
    =
    \E[\tau_d\,p_2(\ATE_2)\,\ATE_3].
  \]
  For the single-stage impact objective,
  \[
    \theta_d^{\singleimpact}
    =
    \E[\tau_d \ATE_3]
    =
    \E[\tau_d]\E[\ATE_3]+\cov(\tau_d,\ATE_3)
  \]
  by the linearity of expectation and the definition of covariance.
\end{proof}

%% file: paper/A4_elliptical_priors_and_small_budget_index_decompositions.tex
\section{Elliptical Priors and Small-Budget Index Decompositions}
\label{sec:elliptical-decompositions-appendix}

The previous section developed the optimal small-budget design in terms of small-budget indices under an arbitrary prior.
In this section we show that the small-budget indices can be written in a more interpretable form under an elliptical prior.
Only this section restricts the prior to the class of elliptical distributions; the general small-budget and large-budget results do not.

We refer to \citeauthor{FangKo90}'s \citetitle{FangKo90}~\cite{FangKo90} for the definition of elliptical distributions and an extensive discussion of their properties.
Note that when $\mathrm{EC}_D(\bm\mu,\Sigma,\phi)$ has finite second moments, the location $\bm\mu$ and scale $\Sigma$ parameters have the following relationship to the mean and covariance (Theorem 2.17 of~\cite{FangKo90}):
\[
  \E[\bm\tau]=\bm\mu \quad\text{and}\quad \cov(\bm\tau) \propto \Sigma.
\]

The key property of elliptical distributions $\mathrm{EC}_D(\bm\mu,\Sigma,\phi)$ that we will use is that $\bm\tau$ has an affine conditional mean given any linear summary of $\bm\tau$.

\begin{lemma}[Linear-summary conditional mean]
  \label{lem:elliptical-conditional-mean-singular-appendix}
  Let \(\bm\tau\sim\mathrm{EC}_D(\bm\mu,\Sigma,\phi)\) have finite second moments.
  Let \(A\in\mathbb R^{D\times k}\), and set \(L=A^\top\bm\tau\).
  Then
  \[
    \E[\bm\tau\mid L]
    =
    \bm\mu
    +
    \cov(\bm\tau,L)\var(L)^\dagger(L-\E[L])
    \qquad\text{a.s.},
  \]
  where \({}^\dagger\) denotes the Moore--Penrose inverse.
\end{lemma}

\begin{proof}
  Let
  \[
    \bm X\coloneqq \bm\tau-\bm\mu,
    \qquad
    C\coloneqq \cov(\bm\tau).
  \]
  By Theorem 2.17 of~\cite{FangKo90}, \(\E[\bm\tau]=\bm\mu\) and \(C\propto\Sigma\).
  If \(C=0\), then
  \(\bm\tau=\bm\mu\) a.s., and the claim is immediate.

  Hence assume \(C\neq 0\).
  Let \(q=\operatorname{rank}(\Sigma)\).
  Then \(q>0\), and \(C=c\Sigma\) for some \(c>0\).
  By Corollary 1 to Theorem 2.14 of~\cite{FangKo90}, there exist
  \(\Gamma\in\mathbb R^{q\times D}\), a nonnegative random variable \(r\),
  and \(\bm u^{(q)}\) uniformly distributed on the unit sphere in
  \(\mathbb R^q\), with \(r\) independent of \(\bm u^{(q)}\), such that
  \[
    \Gamma^\top\Gamma=\Sigma,
    \qquad
    \Gamma^\top(r\bm u^{(q)})\stackrel d=\bm X .
  \]
  Since the desired conditional-mean identity is determined by the joint
  law of \((\bm X,A^\top\bm X)\), it suffices to prove it for this equal-in-distribution copy.
  Thus, writing
  \[
    \bm S\coloneqq r\bm u^{(q)},
  \]
  we may assume, without loss of generality, that
  \[
    \bm X=\Gamma^\top\bm S.
  \]
  The finite-second-moment assumption ensures \(\bm S\) is integrable.

  We first record a simple spherical projection identity.
  For any \(M\in\mathbb R^{q\times k}\), let
  \[
    P\coloneqq M(M^\top M)^\dagger M^\top .
  \]
  Then \(P\) is the orthogonal projection onto \(\operatorname{Im}(M)\), and
  \[
    \E[\bm S\mid M^\top\bm S]=P\bm S
    \qquad\text{a.s.}
  \]
  Indeed, \(P\bm S=M(M^\top M)^\dagger M^\top\bm S\) is measurable with respect to \(M^\top\bm S\).
  Let \(J=2P-I_q\).
  Then \(J\) is orthogonal,
  \(M^\top J\bm S=M^\top\bm S\), and
  \((I_q-P)J\bm S=-(I_q-P)\bm S\).
  Since \(\bm u^{(q)}\) is uniform on the sphere, \(J\bm S\stackrel d=\bm S\).
  Thus, for every bounded measurable
  \(g\),
  \[
    \begin{aligned}
      \E\!\left[g(M^\top\bm S)(I_q-P)\bm S\right]
      &=
      \E\!\left[g(M^\top J\bm S)(I_q-P)J\bm S\right] \\
      &=
      -\E\!\left[g(M^\top\bm S)(I_q-P)\bm S\right].
    \end{aligned}
  \]
  Hence \(\E[(I_q-P)\bm S\mid M^\top\bm S]=0\), proving the identity.

  Apply this identity with \(M=\Gamma A\).
  Since
  \(L=A^\top\bm\mu+A^\top\bm X\), conditioning on \(L=A^\top\bm X\) is the same as
  conditioning on \(A^\top\bm X=M^\top\bm S\).
  Therefore
  \[
    \begin{aligned}
      \E[\bm X\mid L]
      &=
      \Gamma^\top\E[\bm S\mid M^\top\bm S] \\
      &=
      \Gamma^\top M(M^\top M)^\dagger M^\top\bm S \\
      &=
      \Sigma A(A^\top\Sigma A)^\dagger A^\top\bm X .
    \end{aligned}
  \]
  Since \(C=c\Sigma\),
  \[
    \Sigma A(A^\top\Sigma A)^\dagger
    =
    C A(A^\top C A)^\dagger .
  \]
  Finally,
  \[
    \cov(\bm\tau,L)=CA,
    \qquad
    \var(L)=A^\top C A,
    \qquad
    L-\E[L]=A^\top\bm X.
  \]
  Substitution gives the claim.
\end{proof}

For the single-stage objectives, the general-prior formulas from \Cref{prop:first-stage-coefficient-formulas-appendix} already give
\[
  \theta_d^{\singlesuccess}=\E[\tau_d]=\mu_d,
  \qquad
  \theta_d^{\singleimpact}
  =
  \E[\ATE_3] \, \mu_d+\cov(\tau_d,\ATE_3).
\]
The next lemma records the elliptical-prior calculation that will be used for the two-stage objectives.

\begin{lemma}[Weighted linear-summary decomposition]
  \label{lem:elliptical-weighted-linear-summary-appendix}
  Let \(\bm\tau\sim\mathrm{EC}_D(\bm\mu,\Sigma,\phi)\) have finite second moments.
  Let \(L=A^\top\bm\tau\in\mathbb R^k\) be any linear summary.
  Let \(H\) be real-valued and measurable with respect to \(L\), and suppose
  \[
    \E[|H|]<\infty
    \qquad\text{and}\qquad
    \E[|H|\,\|\bm\tau\|_2]<\infty .
  \]
  Then, for every type \(d\),
  \[
    \E[\tau_d H]
    =
    \E[H]\,\mu_d
    +
    \cov(\tau_d,L)^\top\var(L)^\dagger\cov(L,H).
  \]
\end{lemma}

\begin{proof}
  By \Cref{lem:elliptical-conditional-mean-singular-appendix},
  \[
    \E[\bm\tau\mid L]
    =
    \bm\mu
    +
    \cov(\bm\tau,L)\var(L)^\dagger(L-\E[L])
    \qquad\text{a.s.}
  \]
  Taking the \(d\)-th coordinate and using that \(H\) is measurable with
  respect to \(L\),
  \[
    \begin{aligned}
      \E[\tau_d H]
      &=
      \E\!\left[H\,\E[\tau_d\mid L]\right] \\
      &=
      \E[H]\,\mu_d
      +
      \cov(\tau_d,L)^\top\var(L)^\dagger
      \E\!\left[(L-\E[L])H\right] \\
      &=
      \E[H]\,\mu_d
      +
      \cov(\tau_d,L)^\top\var(L)^\dagger\cov(L,H).
    \end{aligned}
  \]
\end{proof}

We now use this result to decompose the small-budget indices for the two-stage success objective and the pilot (two-stage) impact objective into more interpretable components.
We begin with the two-stage success objective.

\begin{proposition}[Two-stage success coefficient under an elliptical prior]
  \label{prop:two-stage-success-elliptical-appendix}
  Assume the prior distribution of \(\bm\tau\) is elliptical with finite second moments.
  Define
  \[
    W_2^{\twosuccess}
    \defeq
    \begin{cases}
      \dfrac{\cov(\ATE_2,p_2(\ATE_2))}{\var(\ATE_2)}, & \text{if }\var(\ATE_2)>0,\\[8pt]
      0, & \text{if }\var(\ATE_2)=0.
    \end{cases}
  \]
  Then, for every type \(d\),
  \[
    \theta_d^{\twosuccess}
    =
    \P(\Test_2=1)\,\mu_d
    +
    W_2^{\twosuccess}\,\cov(\tau_d,\ATE_2),
  \]
  and consequently
  \[
    \pi_d^{\twosuccess}
    =
    \frac{
      \P(\Test_2=1)\,\E[\tau_d]
      +
      W_2^{\twosuccess}\,\cov(\tau_d,\ATE_2)
    }{\sqrt{c_d}}.
  \]
  Moreover, if \(\var(\ATE_2)>0\), then \(W_2^{\twosuccess}>0\).
\end{proposition}

\begin{proof}
  Apply \Cref{lem:elliptical-weighted-linear-summary-appendix} with \(L=\ATE_2\) and \(H=p_2(\ATE_2)\).
  Since \(\E[p_2(\ATE_2)]=\P(\Test_2=1)\), this gives the displayed decomposition.
  When \(\var(\ATE_2)=0\), the covariance term is identically zero and the displayed definition sets \(W_2^{\twosuccess}=0\).
  When \(\var(\ATE_2)>0\), the generalized inverse reduces to ordinary division, giving the displayed ratio.

  It remains to prove positivity in the nondegenerate case.
  Let \(X\) and \(X'\) be independent copies of \(\ATE_2\).
  Since
  \[
    p_2(x)=\Phi(\sqrt{n_2}x-z_{1-\alpha_2})
  \]
  is strictly increasing,
  \[
    (X-X')\bigl(p_2(X)-p_2(X')\bigr)\ge0
    \qquad\text{a.s.},
  \]
  and the inequality is strict on \(\{X\ne X'\}\).
  If \(\var(\ATE_2)>0\), then \(\P(X\ne X')>0\).
  Therefore
  \[
    2\cov(\ATE_2,p_2(\ATE_2))
    =
    \E\!\left[(X-X')\bigl(p_2(X)-p_2(X')\bigr)\right]
    >0.
  \]
  Dividing by \(\var(\ATE_2)>0\) proves \(W_2^{\twosuccess}>0\).
\end{proof}

Now we turn to the main result of this section, the decomposition of the small-budget index for the pilot (two-stage) impact objective.

\refstepcounter{proposition}
\begin{linkedresult}{proof}{prop:marginal-impact-interpretation}
  \label{prop:marginal-impact-interpretation-appendix}
  Assume the prior distribution of \(\bm\tau\) is elliptical with finite second moments.
  Then there exist scalars \(W_1,W_2,W_3\in\R\), not depending on \(d\), such that, for every type \(d\),
  \[
    \theta_d
    =
    W_1\,\mu_d
    +
    W_2\,\cov(\tau_d,\ATE_2)
    +
    W_3\,\cov(\tau_d,\ATE_3).
  \]
  Consequently,
  \[
    \idx_d
    =
    \frac{
      W_1\,\E[\tau_d]
      +
      W_2\,\cov(\tau_d,\ATE_2)
      +
      W_3\,\cov(\tau_d,\ATE_3)
    }{\sqrt{c_d}}.
  \]

  If additionally \(\bm s_2\propto\bm s_3\), \(\E[\ATE_3]\ge0\), and \(\var(\ATE_3)>0\), then
  \[
    \theta_d
    \propto
    \overline W_1\,\mu_d+
    \cov(\tau_d,\ATE_3)
  \]
  for some \(\overline W_1>0\).
  Equivalently, the small-budget index is a positive multiple of
  \[
    \idx_d
    \propto
    \frac{
      \overline W_1\,\E[\tau_d]
      +
      \cov(\tau_d,\ATE_3)
    }{\sqrt{c_d}},
    \qquad
    \overline W_1>0.
  \]
\end{linkedresult}

\begin{proof}
  Let
  \[
    L\defeq
    \begin{pmatrix}
      \ATE_2\\
      \ATE_3
    \end{pmatrix},
    \qquad
    H\defeq p_2(\ATE_2)\ATE_3.
  \]
  Since \(0<p_2(\ATE_2)<1\) and \(\ATE_3\) is a linear summary of \(\bm\tau\), finite second moments imply \(\E[|H|\,\|\bm\tau\|_2]<\infty\).
  Applying \Cref{lem:elliptical-weighted-linear-summary-appendix} gives
  \[
    \theta_d
    =
    \E[\tau_dH]
    =
    \E[H] \, \mu_d
    +
    \cov(\tau_d,L)^\top\var(L)^\dagger\cov(L,H).
  \]
  Define
  \[
    W_1\defeq \E[H],
    \qquad
    \begin{pmatrix}
      W_2\\
      W_3
    \end{pmatrix}
    \defeq
    \var(L)^\dagger\cov(L,H).
  \]
  Since
  \[
    \cov(\tau_d,L)
    =
    \begin{pmatrix}
      \cov(\tau_d,\ATE_2)\\
      \cov(\tau_d,\ATE_3)
    \end{pmatrix},
  \]
  this is exactly
  \[
    \theta_d
    =
    W_1\mu_d
    +
    W_2\cov(\tau_d,\ATE_2)
    +
    W_3\cov(\tau_d,\ATE_3).
  \]
  Dividing by \(\sqrt{c_d}\) gives the displayed index decomposition.

  It remains to prove the representative-follow-up reduction.
  Suppose \(\bm s_2\propto\bm s_3\), which implies \(\ATE_2=\ATE_3\) a.s.
  Write their common value as
  \[
    X\defeq \ATE_3.
  \]
  Then
  \[
    \theta_d
    =
    \E[\tau_d p_2(X)X].
  \]
  Since \(\var(X)=\var(\ATE_3)>0\), applying \Cref{lem:elliptical-weighted-linear-summary-appendix} with \(L=X\) and \(H=p_2(X)X\) gives
  \[
    \theta_d
    =
    A_0\mu_d+B_0\cov(\tau_d,X),
  \]
  where
  \[
    A_0\defeq \E[p_2(X)X],
    \qquad
    B_0\defeq \frac{\cov(X,p_2(X)X)}{\var(X)}.
  \]
  We show that \(A_0>0\) and \(B_0>0\) when \(\E[X]\ge0\).

  Let \(m\defeq\E[X]\).
  As in the proof of \Cref{prop:two-stage-success-elliptical-appendix}, strict monotonicity of \(p_2\) and \(\var(X)>0\) imply
  \[
    \cov(X,p_2(X))>0.
  \]
  Hence
  \[
    A_0
    =
    \E[p_2(X)]\,m+
    \cov(X,p_2(X))
    >0
  \]
  whenever \(m\ge0\).
  Also,
  \begin{align*}
    \cov(X,p_2(X)X)
    &=
    \E[(X-m)p_2(X)X] \\
    &=
    m\,\cov(X,p_2(X))
    +
    \E[p_2(X)(X-m)^2]
    >0,
  \end{align*}
  because \(m\ge0\), \(\cov(X,p_2(X))>0\), \(p_2(X)>0\) a.s., and \(\var(X)>0\).
  Thus \(B_0>0\).
  Therefore
  \[
    \theta_d
    =
    A_0\mu_d+B_0\cov(\tau_d,X)
    \propto
    \frac{A_0}{B_0}\mu_d+
    \cov(\tau_d,X),
  \]
  with \(\overline W_1\defeq A_0/B_0>0\).
  Substituting back \(X=\ATE_3\) proves the claimed special case.
\end{proof}

\begin{corollary}[Named small-budget coefficient decompositions under an elliptical prior]
  \label{cor:named-objective-elliptical-decompositions-appendix}
  Assume the prior distribution of \(\bm\tau\) is elliptical with finite second moments.
  Then, for every type \(d\), the four named small-budget coefficients satisfy
  \begin{align*}
    \theta_d^{\singlesuccess}
    &=
    \mu_d, \\
    \theta_d^{\twosuccess}
    &=
    \P(\Test_2=1)\,\mu_d
    +
    W_2^{\twosuccess}\,\cov(\tau_d,\ATE_2), \\
    \theta_d^{\singleimpact}
    &=
    \E[\ATE_3] \, \mu_d+
    \cov(\tau_d,\ATE_3), \\
    \theta_d^{\twoimpact}
    &=
    W_1\,\mu_d
    +
    W_2\,\cov(\tau_d,\ATE_2)
    +
    W_3\,\cov(\tau_d,\ATE_3)
  \end{align*}
  for some scalars \(W_1,W_2,W_3\in\R\), not depending on \(d\).
  The coefficient \(W_2^{\twosuccess}\) is defined in \Cref{prop:two-stage-success-elliptical-appendix}; in particular, \(W_2^{\twosuccess}>0\) whenever \(\var(\ATE_2)>0\).
\end{corollary}

\begin{proof}
  The single-stage success and single-stage impact formulas are the general-prior identities in \Cref{prop:first-stage-coefficient-formulas-appendix}.
  The two-stage success formula is \Cref{prop:two-stage-success-elliptical-appendix}.
  The two-stage impact formula is \linkedCref{prop:marginal-impact-interpretation-appendix}.
\end{proof}

%% file: paper/A5_large_budget_two_stage_impact_general_prior.tex
\section{Large-Budget Theory for Weighted Target-Impact Objectives}
\label{sec:large-budget-general-prior-appendix}

We now turn to the large-budget behavior of the paper's main objective $V=V_{\twoimpact}$.
As with the small-budget results, we provide proofs in a more general form and then specialize it to the named objectives in the paper.
The relevant large-budget class consists of weighted target-impact objectives
\[
  V_Z(\bm s_1)
  \defeq
  \E[\Test_1\cdot Z\cdot \ATE_3],
\]
where $Z$ is a downstream non-negative weight that is non-adaptive with respect to the pilot-stage sampling noise, meaning $Z\ind \varepsilon_1\mid\bm\tau$.
The paper's main objective $V_\twoimpact$ takes this form for $Z=\Test_2$, and the single-stage objective $V_\singleimpact$ takes this form for $Z=1$.

Throughout this section, we will use the unit-investment design set
\[
  \overline{\mathcal S}
  =
  \{\overline{\bm s}_1\in\R_{\ge0}^D:\bm c^\top\overline{\bm s}_1=1\}
\]
and the representative unit-investment design
\[
  \overline{\bm s}_1^{\,\infty}
  \defeq
  \frac{\bm s_3}{\bm c^\top\bm s_3}.
\]

We will rely on three prior conditions for the proofs.
First, $\ATE_3$ has a finite second moment and is \emph{non-degenerate}, i.e. has non-zero variance.
Second, the prior satisfies what we call a \emph{uniform boundary anti-concentration} condition.
With
\[
  \xi(\overline{\bm s}_1,\bm\tau)
  \defeq
  \frac{\overline{\bm s}_1^\top\bm\tau}
  {\sqrt{\bm1^\top\overline{\bm s}_1}},
\]
we assume
\begin{equation}
  \omega(\delta)
  \defeq
  \sup_{\overline{\bm s}_1\in\overline{\mathcal S}}
  \P\!\left(|\xi(\overline{\bm s}_1,\bm\tau)|\le \delta\right)
  \xrightarrow[\delta\downarrow0]{}0.
  \label{eq:uniform-boundary-anti-concentration-large-budget-appendix}
\end{equation}
Third, the prior satisfies the \emph{sign-separation condition}
\begin{equation}
  \forall \overline{\bm s}_1\in\overline{\mathcal{S}}\text{ not proportional to }\bm s_3:\qquad
  \P\!\big((\overline{\bm s}_1^\top\bm\tau)(\bm s_3^\top\bm\tau)<0\big)>0.
  \label{eq:ordinary-sign-separation-large-budget-appendix}
\end{equation}
Distributions with a Lebesgue density positive on an open ball containing the origin satisfy the latter two regularity conditions, and they also imply positive variance of $\ATE_3$ whenever $\E[\ATE_3^2]<\infty$; see \Cref{prop:density-large-budget-regularity-appendix}.

For any non-adaptive downstream non-negative weight $Z$, write
\[
  m_Z(\bm\tau)
  \defeq
  \E[Z\mid\bm\tau].
\]
We define the corresponding follow-up, oracle, and noise-free values as
\begin{align*}
  V_{Z,\followup}
  &\defeq
  \E[Z\cdot \ATE_3],\\
  V_{Z,\oracle}
  &\defeq
  \E[\indic{\ATE_3>0}\cdot Z\cdot \ATE_3],\\
  V_{Z,\noisefree}(\overline{\bm s}_1)
  &\defeq
  \E[\indic{\overline{\bm s}_1^\top\bm\tau>0}\cdot Z\cdot \ATE_3].
\end{align*}
For a budget $b>0$, define the budget-constrained value and solution set
\begin{align*}
  V_Z^\star(b)
  &\defeq
  \sup_{\substack{\bm s_1\in\R_{\geq 0}^{D} \setminus \{\bm{0}\}\\ \bm c^\top\bm s_1\le b}}
  V_Z(\bm s_1),\\
  \mathcal S_{1,Z}^\star(b)
  &\defeq
  \arg\max_{\substack{\bm s_1\in\R_{\geq 0}^{D} \setminus \{\bm{0}\}\\ \bm c^\top\bm s_1\le b}}
  V_Z(\bm s_1).
\end{align*}
When $Z=\Test_2$, these reduce to the paper's main value $V^\star(b)$ and solution set $\mathcal S_1^\star(b)$.

The weighted objective requires the following two analogues of nondegeneracy and sign separation.
First, we assume \emph{weighted nondegeneracy}:
\begin{equation}
  \E[m_Z(\bm\tau)|\ATE_3|]>0.
  \label{eq:weighted-nondegeneracy}
\end{equation}
Second, we assume \emph{weighted sign separation}:
\begin{equation}
  \forall \overline{\bm s}_1\in\overline{\mathcal{S}}\text{ not proportional to }\bm s_3:\qquad
  \E\!\left[
    m_Z(\bm\tau)|\ATE_3|
    \indic{(\overline{\bm s}_1^\top\bm\tau)(\bm s_3^\top\bm\tau)<0}
  \right]>0.
  \label{eq:weighted-sign-separation}
\end{equation}
These conditions are automatic for the paper's objective once the ordinary sign-separation and non-zero-variance conditions hold, because $\E[\Test_2\mid\bm\tau]=p_2(\ATE_2)>0$ almost surely; this verification is given in \Cref{lem:test-two-satisfies-weighted-conditions-appendix}.

The main generic result of this section is the following theorem.
The rest of the section proves it and then specializes it to $Z=\Test_2$.
The proof relies on three key facts: finite-investment pilots are strictly below the oracle value, large-investment pilots converge uniformly to a noise-free pilot rule, and the unique noise-free design that attains the oracle value is the representative design.

\begin{theorem}[Large-budget limit for weighted target-impact objectives]
  \label{thm:large-budget-weighted-target-impact-appendix}
  Let $Z$ be non-adaptive and suppose $0\le Z\le M<\infty$ almost surely.
  Suppose $\E[\ATE_3^2]<\infty$, the uniform boundary anti-concentration condition \Cref{eq:uniform-boundary-anti-concentration-large-budget-appendix} holds, and the weighted non-degeneracy and weighted sign-separation conditions \Cref{eq:weighted-nondegeneracy,eq:weighted-sign-separation} hold.
  Then $\mathcal S_{1,Z}^\star(b)\neq\emptyset$ for all $b>0$ and:
  \begin{enumerate}
    \item[\rm (i)] $V_Z^\star(b)\le V_{Z,\oracle}$ for all $b>0$.
    \item[\rm (ii)] $V_Z^\star(b)\to V_{Z,\oracle}$ as $b\to\infty$.
    \item[\rm (iii)] The optimal investment diverges,
      \[
        \inf_{\bm s_1^\star\in\mathcal S_{1,Z}^\star(b)} \bm c^\top\bm s_1^\star
        \to\infty
        \qquad\text{as }b\to\infty.
      \]
    \item[\rm (iv)] The set of unit-investment normalized optimal designs converges to the unit-investment representative design,
      \[
        \sup_{\bm s_1^\star\in\mathcal S_{1,Z}^\star(b)}
        \left\|
        \frac{\bm s_1^\star}{\bm c^\top\bm s_1^\star}
        -
        \overline{\bm s}_1^{\,\infty}
        \right\|
        \to0
        \qquad\text{as }b\to\infty.
      \]
  \end{enumerate}
\end{theorem}

\subsection*{Proof ingredients for the generic theorem}

\begin{lemma}[Boundary and noise-free limits]
  \label{prop:limiting-impact-values}
  Let $Z$ be non-adaptive and suppose $\E[|Z \cdot \ATE_3|]<\infty$.
  Then $V_Z$ is continuous on $\R_{\geq 0}^{D} \setminus \{\bm{0}\}$ and
  \[
    \lim_{\bm s_1\to\bm0}V_Z(\bm s_1)
    =
    \alpha_1\,V_{Z,\followup}.
  \]
  If, in addition, the uniform boundary anti-concentration condition \Cref{eq:uniform-boundary-anti-concentration-large-budget-appendix} holds, then for every $\overline{\bm s}_1\in\overline{\mathcal{S}}$,
  \[
    \lim_{\invst\to\infty}V_Z(\invst\,\overline{\bm s}_1)
    =
    V_{Z,\noisefree}(\overline{\bm s}_1).
  \]
  Under the same condition, $V_{Z,\noisefree}$ is continuous on $\overline{\mathcal{S}}$.
\end{lemma}

\begin{proof}
  \emph{Step 1: Continuity of $V_Z$ on $\R_{\geq 0}^{D} \setminus \{\bm{0}\}$ and the small-sample limit.}
  For $\bm s_1\in\R_{\geq 0}^{D} \setminus \{\bm{0}\}$, write $n_1=\bm1^\top\bm s_1$ and
  \[
    q(\bm s_1,\bm\tau)
    \defeq
    \Phi\!\left(
      \frac{\bm s_1^\top\bm\tau}{\sqrt{n_1}}
      -
      z_{1-\alpha_1}
    \right).
  \]
  Since $q(\bm s_1,\bm\tau)=\P(\Test_1=1\mid\bm\tau)$ and $Z$ is independent of the pilot-stage sampling noise conditional on $\bm\tau$,
  \[
    V_Z(\bm s_1)
    =
    \E[q(\bm s_1,\bm\tau)\cdot Z\cdot \ATE_3].
  \]
  Fix any sequence $\bm s_1^{(k)}\to \bm s_1$ in $\R_{\geq 0}^{D} \setminus \{\bm{0}\}$.
  Then $q(\bm s_1^{(k)},\bm\tau)\to q(\bm s_1,\bm\tau)$ pointwise in $\bm\tau$ by continuity of the Gaussian distribution function.
  Since $0\le q(\bm s_1^{(k)},\bm\tau)\le 1$ and $|Z \cdot \ATE_3|$ is integrable, dominated convergence gives
  \[
    V_Z(\bm s_1^{(k)})\to V_Z(\bm s_1),
  \]
  proving continuity on $\R_{\geq 0}^{D} \setminus \{\bm{0}\}$.

  If $\bm s_1\to\bm0$, then $n_1\to0$ and
  \[
    \left|\frac{\bm s_1^\top\bm\tau}{\sqrt{n_1}}\right|
    \le
    \sqrt{n_1}\,\|\bm\tau\|_1
    \to 0
    \qquad\text{for each fixed }\bm\tau,
  \]
  so $q(\bm s_1,\bm\tau)\to\Phi(-z_{1-\alpha_1})=\alpha_1$ pointwise.
  Dominated convergence therefore gives
  \[
    \lim_{\bm s_1\to\bm0}V_Z(\bm s_1)
    =
    \alpha_1\,\E[Z\cdot \ATE_3]
    =
    \alpha_1\,V_{Z,\followup}.
  \]

  \emph{Step 2: Pointwise large-investment limit.}
  Fix $\overline{\bm s}_1\in\overline{\mathcal{S}}$ and write
  \[
    q_{\invst}(\bm\tau)
    \defeq
    \Phi\!\left(
      \sqrt{\invst}\,\xi(\overline{\bm s}_1,\bm\tau)
      -
      z_{1-\alpha_1}
    \right).
  \]
  The uniform boundary anti-concentration condition implies
  \[
    \P\!\big(
      \xi(\overline{\bm s}_1,\bm\tau)=0
    \big)=0.
  \]
  Hence $q_{\invst}(\bm\tau)\to \indic{\xi(\overline{\bm s}_1,\bm\tau)>0}$ almost surely as $\invst\to\infty$.
  Since $\xi(\overline{\bm s}_1,\bm\tau)>0$ if and only if $\overline{\bm s}_1^\top\bm\tau>0$, dominated convergence gives
  \[
    \lim_{\invst\to\infty}V_Z(\invst\,\overline{\bm s}_1)
    =
    \E[\indic{\overline{\bm s}_1^\top\bm\tau>0}\cdot Z\cdot \ATE_3]
    =
    V_{Z,\noisefree}(\overline{\bm s}_1).
  \]

  \emph{Step 3: Continuity of $V_{Z,\noisefree}$ on $\overline{\mathcal{S}}$.}
  Let $\overline{\bm s}_1^{(k)}\to\overline{\bm s}_1$ in $\overline{\mathcal{S}}$.
  Then $\xi(\overline{\bm s}_1^{(k)},\bm\tau)\to\xi(\overline{\bm s}_1,\bm\tau)$ pointwise in $\bm\tau$.
  By the previous argument,
  \[
    \P\!\big(\xi(\overline{\bm s}_1,\bm\tau)=0\big)=0,
  \]
  so
  \[
    \indic{\xi(\overline{\bm s}_1^{(k)},\bm\tau)>0}
    \to
    \indic{\xi(\overline{\bm s}_1,\bm\tau)>0}
    \qquad\text{almost surely.}
  \]
  Since $|\indic{\xi>0}\cdot Z\cdot \ATE_3|\le |Z \cdot \ATE_3|$, dominated convergence gives
  \[
    V_{Z,\noisefree}(\overline{\bm s}_1^{(k)})
    \to
    V_{Z,\noisefree}(\overline{\bm s}_1),
  \]
  proving continuity on $\overline{\mathcal{S}}$.
\end{proof}

\begin{proposition}[Uniform convergence to the noise-free pilot]
  \label{prop:large-investment-expansion-weighted-appendix}
  Let $Z$ be non-adaptive and suppose $0\le Z\le M<\infty$ almost surely.
  Suppose $\E[\ATE_3^2]<\infty$ and the uniform boundary anti-concentration condition \Cref{eq:uniform-boundary-anti-concentration-large-budget-appendix} holds.
  Then as $\invst\to\infty$,
  \[
    \sup_{\overline{\bm s}_1\in\overline{\mathcal{S}}}
    \big|V_Z(\invst\,\overline{\bm s}_1)-V_{Z,\noisefree}(\overline{\bm s}_1)\big|
    \to 0.
  \]
\end{proposition}

\begin{proof}
  Fix $\delta>0$ and bound the tail error by
  \[
    r_{\invst}(\delta)
    \defeq
    \max\!\left\{
      \Phi(z_{1-\alpha_1}-\sqrt{\invst}\,\delta),
      \Phi(-z_{1-\alpha_1}-\sqrt{\invst}\,\delta)
    \right\}.
  \]
  For every $\overline{\bm s}_1\in\overline{\mathcal{S}}$ and every realization of $\bm\tau$, bound the approximation error between the finite and infinite investment decisions
  \[
    e_{\invst}(\overline{\bm s}_1,\bm\tau)
    \defeq
    \left|
    \Phi\!\left(
      \sqrt{\invst}\,\xi(\overline{\bm s}_1,\bm\tau)-z_{1-\alpha_1}
    \right)
    -
    \indic{\xi(\overline{\bm s}_1,\bm\tau)>0}
    \right|.
  \]
  If $|\xi(\overline{\bm s}_1,\bm\tau)|>\delta$, the error is bounded by $r_{\invst}(\delta)$; otherwise it is bounded by $1$.
  Hence
  \[
    e_{\invst}(\overline{\bm s}_1,\bm\tau)
    \le
    \indic{\big|\xi(\overline{\bm s}_1,\bm\tau)\big|\le\delta}
    +
    r_{\invst}(\delta).
  \]
  Therefore, applying non-adaptivity of $Z$,
  \begin{align*}
    \big|V_Z(\invst\,\overline{\bm s}_1)-V_{Z,\noisefree}(\overline{\bm s}_1)\big|
    &=
    \left|
    \E\!\left[
      Z\cdot \ATE_3\cdot
      \left(
        \Phi\!\left(
          \sqrt{\invst}\,\xi(\overline{\bm s}_1,\bm\tau)-z_{1-\alpha_1}
        \right)
        -
        \indic{\xi(\overline{\bm s}_1,\bm\tau)>0}
      \right)
    \right]
    \right| \\
    &\le
    \E\!\left[
      Z\cdot |\ATE_3|\cdot e_{\invst}(\overline{\bm s}_1,\bm\tau)
    \right] \\
    &\le
    \E\!\left[
      Z \cdot |\ATE_3|\cdot
      \indic{\big|\xi(\overline{\bm s}_1,\bm\tau)\big|\le\delta}
    \right]
    +
    r_{\invst}(\delta)\,\E[Z\cdot |\ATE_3|].
  \end{align*}
  Using $0\le Z\le M$ and Cauchy--Schwarz,
  \[
    \E\!\left[
      Z\cdot |\ATE_3|\cdot
      \indic{\big|\xi(\overline{\bm s}_1,\bm\tau)\big|\le\delta}
    \right]
    \le
    M\sqrt{\E[\ATE_3^2]}
    \P\!\big(\big|\xi(\overline{\bm s}_1,\bm\tau)\big|\le\delta\big)^{1/2}
    \le
    M\sqrt{\E[\ATE_3^2] \omega(\delta)}.
  \]
  Taking the supremum over $\overline{\bm s}_1\in\overline{\mathcal{S}}$ gives
  \[
    \sup_{\overline{\bm s}_1\in\overline{\mathcal{S}}}
    \big|V_Z(\invst\,\overline{\bm s}_1)-V_{Z,\noisefree}(\overline{\bm s}_1)\big|
    \le
    M\sqrt{\E[\ATE_3^2]\omega(\delta)}
    +
    r_{\invst}(\delta)\,M\E[|\ATE_3|].
  \]

  Let $\Delta_{\invst}$ denote the left-hand side of the preceding inequality, and fix any $\delta>0$.
  Since $\alpha_1\in(0,1)$, $z_{1-\alpha_1}$ is finite, and hence both $z_{1-\alpha_1}-\sqrt{\invst}\delta$ and $-z_{1-\alpha_1}-\sqrt{\invst}\delta$ tend to $-\infty$ as $\invst\to\infty$.
  Therefore $r_{\invst}(\delta)\to0$ for each fixed $\delta>0$.
  Taking $\limsup_{\invst\to\infty}$ in the preceding bound gives
  \[
    \limsup_{\invst\to\infty}\Delta_{\invst}
    \le
    M\sqrt{\E[\ATE_3^2]\omega(\delta)} .
  \]
  Since this holds for every $\delta>0$, we may let $\delta\downarrow0$; the uniform boundary anti-concentration assumption gives $\omega(\delta)\to0$, so $\limsup_{\invst\to\infty}\Delta_{\invst}\le0$.
  As $\Delta_{\invst}\ge0$, it follows that $\Delta_{\invst}\to0$.
  This proves uniform convergence.
\end{proof}

\begin{lemma}[Existence of optimal designs]
  \label{prop:optimal-solution-exists-weighted-appendix}
  Let $Z$ be non-adaptive and suppose $\E[|Z \cdot \ATE_3|]<\infty$, $m_Z(\bm\tau)\ge0$ almost surely, and \Cref{eq:weighted-nondegeneracy} holds.
  Then $\mathcal S_{1,Z}^\star(b)\neq\emptyset$ for every $b>0$.
\end{lemma}

\begin{proof}
  Fix $b>0$ and define the extended reparametrized objective
  \[
    f_Z(\invst,\overline{\bm s}_1)\defeq
    \begin{cases}
      V_Z(\invst\,\overline{\bm s}_1), & \invst>0,\\
      \alpha_1\,V_{Z,\followup}, & \invst=0,
    \end{cases}
    \qquad
    (\invst,\overline{\bm s}_1)\in[0,b]\times\overline{\mathcal{S}}.
  \]
  By \Cref{prop:limiting-impact-values}, $f_Z$ is continuous on the compact set $[0,b]\times\overline{\mathcal{S}}$, so Weierstrass' theorem implies that $f_Z$ attains its maximum there.

  It remains to show that no maximizer has $\invst^\star=0$.
  Let $\overline n_1^\infty\defeq \bm1^\top\overline{\bm s}_1^{\,\infty}>0$ and define
  \[
    q_\invst(x)
    \defeq
    \Phi\!\big(\sqrt{\invst\,\overline n_1^\infty}\,x-z_{1-\alpha_1}\big).
  \]
  Because the representative pilot satisfies $\ATE_1=\ATE_3$, we have
  \[
    f_Z(\invst,\overline{\bm s}_1^{\,\infty})-
    \alpha_1V_{Z,\followup}
    =
    \E\!\left[
      m_Z(\bm\tau)\,\ATE_3\big(q_\invst(\ATE_3)-\alpha_1\big)
    \right].
  \]
  Since $q_\invst(0)=\alpha_1$ and $q_\invst(\cdot)$ is strictly increasing, the product
  \[
    x\big(q_\invst(x)-\alpha_1\big)
  \]
  is strictly positive for every $x\neq0$.
  Since $m_Z(\bm\tau)\ge0$ almost surely and \Cref{eq:weighted-nondegeneracy} holds, it follows that
  \[
    f_Z(\invst,\overline{\bm s}_1^{\,\infty})>
    \alpha_1V_{Z,\followup}
    =
    f_Z(0,\overline{\bm s}_1^{\,\infty})
    \qquad
    \forall \invst>0.
  \]
  Hence the value attained at any positive investment along the representative design strictly improves on $\invst=0$.
  No maximizer of $f_Z$ can therefore satisfy $\invst^\star=0$.

  Therefore there exists a maximizer $(\invst^\star,\overline{\bm s}_1^\star)$ with $\invst^\star>0$, and then $\invst^\star\overline{\bm s}_1^\star$ is an optimal feasible solution.
\end{proof}

\begin{lemma}[Oracle upper bound]
  \label{lem:oracle-upper-bound-weighted}
  Suppose $Z\ge0$ almost surely and $\E[Z\cdot |\ATE_3|]<\infty$.
  Define the oracle random payoff and oracle value by
  \[
    Y_Z
    \defeq
    \indic{\ATE_3>0}\cdot Z\cdot \ATE_3,
    \qquad
    V_{Z,\oracle}
    \defeq
    \E[Y_Z].
  \]
  Define realized payoff random variables under the noisy and noise-free pilot decisions as
  \[
    X_Z(\bm s_1)\defeq \Test_1\cdot Z\cdot \ATE_3,
    \qquad
    X_{Z,\infty}(\overline{\bm s}_1)
    \defeq
    \indic{\overline{\bm s}_1^\top\bm\tau>0}\cdot Z\cdot \ATE_3.
  \]
  Then for every pilot design $\bm s_1\in\R_{\geq 0}^{D} \setminus \{\bm{0}\}$ and every $\overline{\bm s}_1\in\overline{\mathcal{S}}$,
  \[
    X_Z(\bm s_1)\le Y_Z
    \qquad\text{and}\qquad
    X_{Z,\infty}(\overline{\bm s}_1)\le Y_Z
    \qquad\text{almost surely.}
  \]
  Consequently,
  \[
    V_Z(\bm s_1)\le V_{Z,\oracle}
    \qquad\text{and}\qquad
    V_{Z,\noisefree}(\overline{\bm s}_1)\le V_{Z,\oracle}.
  \]
\end{lemma}

\begin{proof}
  If $\ATE_3>0$, then $Y_Z=Z\cdot \ATE_3\ge0$ and any pilot decision only multiplies this payoff by a number in $\{0,1\}$.
  If $\ATE_3\le0$, then $Y_Z=0$ and any decision to continue yields the non-positive payoff $Z\cdot \ATE_3\le0$.
  Thus both the noisy and noise-free pilot payoffs are bounded above by $Y_Z$ almost surely.
  Taking expectations proves the result.
\end{proof}

\begin{proposition}[Unique optimal noise-free design]
  \label{prop:optimal-noise-free-pilot-design-weighted-appendix}
  Suppose $Z\ge0$ almost surely and $\E[Z\cdot|\ATE_3|]<\infty$.
  Then the unit-investment representative pilot design $\overline{\bm s}_1^{\,\infty}$ attains the oracle bound:
  \[
    V_{Z,\noisefree}(\overline{\bm s}_1^{\,\infty})
    =
    V_{Z,\oracle}.
  \]
  If the weighted sign-separation condition \Cref{eq:weighted-sign-separation} holds, then $\overline{\bm s}_1^{\,\infty}$ is the unique maximizer of $V_{Z,\noisefree}$ over $\overline{\mathcal{S}}$.
\end{proposition}

\begin{proof}
  Under the representative pilot,
  \[
    \ATE_1
    =
    \frac{\left(\overline{\bm s}_1^{\,\infty}\right)^\top\bm\tau}
    {\bm1^\top\overline{\bm s}_1^{\,\infty}}
    =
    \frac{\bm s_3^\top\bm\tau}{\bm1^\top\bm s_3}
    =
    \ATE_3.
  \]
  Hence $\indic{\ATE_1>0}=\indic{\ATE_3>0}$ almost surely under $\overline{\bm s}_1^{\,\infty}$, so
  \[
    V_{Z,\noisefree}(\overline{\bm s}_1^{\,\infty})
    =
    \E[\indic{\ATE_3>0}\cdot Z\cdot \ATE_3]
    =
    V_{Z,\oracle}.
  \]

  Now let $\overline{\bm s}_1\in\overline{\mathcal{S}}$ be any noise-free design.
  The oracle gap can be written as
  \[
    V_{Z,\oracle}-V_{Z,\noisefree}(\overline{\bm s}_1)
    =
    \E\!\left[
      \Big(
        \indic{\ATE_3>0}
        -
        \indic{\overline{\bm s}_1^\top\bm\tau>0}
      \Big)\cdot Z \cdot \ATE_3
    \right],
  \]
  and the integrand is nonnegative almost surely by \Cref{lem:oracle-upper-bound-weighted}.
  On the strict sign-disagreement event
  \[
    A(\overline{\bm s}_1)
    \defeq
    \{(\overline{\bm s}_1^\top\bm\tau)(\bm s_3^\top\bm\tau)<0\},
  \]
  the integrand is exactly $Z\cdot|\ATE_3|$.
  Therefore, conditioning on $\bm\tau$,
  \[
    V_{Z,\oracle}-V_{Z,\noisefree}(\overline{\bm s}_1)
    \ge
    \E\!\left[
      m_Z(\bm\tau)\cdot |\ATE_3|\cdot \indic{A(\overline{\bm s}_1)}
    \right].
  \]
  If $\overline{\bm s}_1$ is not proportional to $\bm s_3$, the right-hand side is strictly positive by \Cref{eq:weighted-sign-separation}.
  Thus no non-proportional design attains the oracle bound.
  Since the representative design does attain the oracle bound, every maximizer must be proportional to $\bm s_3$.
  The unit-investment constraint $\bm c^\top\overline{\bm s}_1=1$ then identifies the unique maximizer as
  \[
    \overline{\bm s}_1
    =
    \frac{\bm s_3}{\bm c^\top\bm s_3}
    =
    \overline{\bm s}_1^{\,\infty}.
  \]
\end{proof}

\begin{lemma}[Strict finite-investment oracle gap]
  \label{lem:strict-finite-investment-oracle-gap-weighted}
  Let $Z$ be non-adaptive and suppose $Z\ge0$ almost surely, $\E[Z\cdot|\ATE_3|]<\infty$, and \Cref{eq:weighted-nondegeneracy} holds.
  Then every finite-investment pilot design satisfies
  \[
    V_Z(\bm s_1)<V_{Z,\oracle}
    \qquad
    \forall \bm s_1\in\R_{\geq 0}^{D}\setminus\{\bm0\}.
  \]
  Consequently, if $\mathcal S_{1,Z}^\star(b)\neq\emptyset$ for a finite budget $b$, then
  \[
    V_Z^\star(b)<V_{Z,\oracle}.
  \]
\end{lemma}

\begin{proof}
  Fix any finite-investment design $\bm s_1\in\R_{\geq 0}^{D}\setminus\{\bm0\}$.
  Conditional on $\bm\tau$, the pilot test is independent of $Z$ and has pass probability
  \[
    p_1(\bm\tau)
    \defeq
    \P(\Test_1=1\mid\bm\tau)
    \in(0,1).
  \]
  With
  \[
    X_Z\defeq \Test_1\cdot Z\cdot \ATE_3,
    \qquad
    Y_Z\defeq \indic{\ATE_3>0}\cdot Z\cdot \ATE_3,
  \]
  the conditional oracle gap is
  \[
    \E[Y_Z-X_Z\mid\bm\tau]
    =
    m_Z(\bm\tau)|\ATE_3|
    \Big(
      (1-p_1(\bm\tau))\indic{\ATE_3>0}
      +
      p_1(\bm\tau)\indic{\ATE_3<0}
    \Big).
  \]
  This quantity is nonnegative everywhere and strictly positive wherever $m_Z(\bm\tau)|\ATE_3|>0$.
  By \Cref{eq:weighted-nondegeneracy}, the latter event has positive weighted mass, so
  \[
    V_Z(\bm s_1)=\E[X_Z]<\E[Y_Z]=V_{Z,\oracle}.
  \]
  If $\mathcal S_{1,Z}^\star(b)$ is nonempty, an optimizer under budget $b$ is a finite-investment design, so the strict inequality also holds for $V_Z^\star(b)$.
\end{proof}

Finally, we will use the following standard result, which follows, for example, from \cite[Prop.~7.15 and Thm.~7.33]{RockafellarWe09}.

\begin{lemma}[Uniform convergence transfers unique maximizers]
  \label{lem:uniform-argmax-transfer-large-budget}
  Let $K\subset \R^D$ be compact, let $f:K\to\R$ be continuous with unique maximizer $x^\star$, and let $f_\lambda:K\to\R$ be functions indexed by $\lambda\to\infty$ such that
  \[
    \sup_{x\in K}|f_\lambda(x)-f(x)|\to0.
  \]
  Then for every $\epsilon>0$ and any (equivalent) norm $\|\cdot\|$, there exists $\lambda_\epsilon<\infty$ such that, for all $\lambda\ge\lambda_\epsilon$, every maximizer of $f_\lambda$ over $K$ lies in $\{x\in K:\|x-x^\star\|<\epsilon\}$.
\end{lemma}

We now have all the components necessary to proceed with a proof of the main large-budget result.

\begin{proof}[Proof of \Cref{thm:large-budget-weighted-target-impact-appendix}]
  Non-emptiness of $\mathcal S_{1,Z}^\star(b)$ follows from \Cref{prop:optimal-solution-exists-weighted-appendix}, since $0\le Z\le M$ and $\E[\ATE_3^2]<\infty$ imply $\E[|Z \cdot \ATE_3|]<\infty$.

  Part (i) follows immediately from the oracle upper bound in \Cref{lem:oracle-upper-bound-weighted}.
  For the matching lower bound in part (ii), \Cref{prop:optimal-noise-free-pilot-design-weighted-appendix} gives
  \[
    V_{Z,\noisefree}(\overline{\bm s}_1^{\,\infty})=V_{Z,\oracle}.
  \]
  Therefore
  \[
    V_Z^\star(b)
    \ge
    V_Z(b\,\overline{\bm s}_1^{\,\infty})
    \ge
    V_{Z,\noisefree}(\overline{\bm s}_1^{\,\infty})
    -
    \sup_{\overline{\bm s}_1\in\overline{\mathcal{S}}}
    \big|V_Z(b\,\overline{\bm s}_1)-V_{Z,\noisefree}(\overline{\bm s}_1)\big|.
  \]
  The supremum term converges to $0$ by \Cref{prop:large-investment-expansion-weighted-appendix}, while part (i) supplies the upper bound $V_Z^\star(b)\le V_{Z,\oracle}$.
  Hence $V_Z^\star(b)\to V_{Z,\oracle}$.

  To prove part (iii), fix any finite $\invst_0>0$.
  By \Cref{lem:strict-finite-investment-oracle-gap-weighted},
  \[
    \delta_{\invst_0}
    \defeq
    V_{Z,\oracle}-V_Z^\star(\invst_0)
    >0.
  \]
  By part (ii), choose $B_{\invst_0}$ such that $V_Z^\star(b)>V_Z^\star(\invst_0)$ whenever $b\ge B_{\invst_0}$.
  If an optimal design under such a budget had investment at most $\invst_0$, it would be feasible for the budget-$\invst_0$ problem and could not exceed $V_Z^\star(\invst_0)$.
  This contradicts $V_Z^\star(b)>V_Z^\star(\invst_0)$.
  Since $\invst_0$ was arbitrary, the optimal investment diverges.

  For part (iv), define, for each investment $\invst>0$,
  \[
    g_\invst(\overline{\bm s}_1)
    \defeq
    V_Z(\invst\,\overline{\bm s}_1),
    \qquad
    g_\infty(\overline{\bm s}_1)
    \defeq
    V_{Z,\noisefree}(\overline{\bm s}_1).
  \]
  By \Cref{prop:large-investment-expansion-weighted-appendix}, $g_\invst\to g_\infty$ uniformly on $\overline{\mathcal S}$.
  By \Cref{prop:optimal-noise-free-pilot-design-weighted-appendix}, $g_\infty$ has the unique maximizer $\overline{\bm s}_1^{\,\infty}$.
  Applying \Cref{lem:uniform-argmax-transfer-large-budget} with $K=\overline{\mathcal S}$ shows that every maximizer of $g_\invst$ is arbitrarily close to $\overline{\bm s}_1^{\,\infty}$ once $\invst$ is large enough, where we note that by equivalence of norms in $\R^D$ the $1$-norm can be replaced with any other norm.

  Finally, let $\bm s_1^\star\in\mathcal S_{1,Z}^\star(b)$ and write
  \[
    \invst^\star\defeq \bm c^\top\bm s_1^\star,
    \qquad
    \overline{\bm s}_1^\star
    \defeq
    \frac{\bm s_1^\star}{\bm c^\top\bm s_1^\star}.
  \]
  For its fixed investment $\invst^\star$, the normalized design $\overline{\bm s}_1^\star$ must maximize $g_{\invst^\star}$ over $\overline{\mathcal S}$; otherwise replacing it by a better unit-investment composition at the same investment would produce a feasible design with a larger value.
  Part (iii) implies that $\invst^\star\to\infty$ uniformly over all optimal designs as $b\to\infty$.
  Combining this investment divergence with the argmax-transfer conclusion proves
  \[
    \sup_{\bm s_1^\star\in\mathcal S_{1,Z}^\star(b)}
    \left\|
    \frac{\bm s_1^\star}{\bm c^\top\bm s_1^\star}
    -
    \overline{\bm s}_1^{\,\infty}
    \right\|
    \to0.
  \]
\end{proof}

\subsection*{Specialization to the two-stage impact objective}

We next verify that the generic weighted assumptions hold for the paper's main two-stage impact objective. This follows because strictly positive conditional weights transfer ordinary sign separation to weighted sign separation.

\begin{lemma}[Positive weights transfer sign separation]
  \label{prop:positive-conditional-weight-appendix}
  Let $Z$ be non-adaptive and suppose $m_Z(\bm\tau)\ge0$ almost surely and $\E[m_Z(\bm\tau)|\ATE_3|]<\infty$.
  If $m_Z(\bm\tau)>0$ almost surely and $\var(\ATE_3)>0$, then \Cref{eq:weighted-nondegeneracy} holds.
  If $m_Z(\bm\tau)>0$ almost surely and the ordinary sign-separation condition \Cref{eq:ordinary-sign-separation-large-budget-appendix} holds, then the weighted sign-separation condition \Cref{eq:weighted-sign-separation} holds.
\end{lemma}

\begin{proof}
  Since $m_Z(\bm\tau)>0$ almost surely and $\var(\ATE_3)>0$, we have $m_Z(\bm\tau)|\ATE_3|>0$ on a set of positive probability.
  This proves \Cref{eq:weighted-nondegeneracy}.

  Now fix any $\overline{\bm s}_1\in\overline{\mathcal{S}}$ not proportional to $\bm s_3$, and let
  \[
    A(\overline{\bm s}_1)
    \defeq
    \{(\overline{\bm s}_1^\top\bm\tau)(\bm s_3^\top\bm\tau)<0\}.
  \]
  By the ordinary sign-separation condition, $\P(A(\overline{\bm s}_1))>0$.
  On this event, $|\ATE_3|>0$, and by assumption $m_Z(\bm\tau)>0$ almost surely.
  Hence the nonnegative random variable
  \[
    m_Z(\bm\tau)|\ATE_3|\indic{A(\overline{\bm s}_1)}
  \]
  is strictly positive on a set of positive probability, so its expectation is strictly positive.
  This is exactly \Cref{eq:weighted-sign-separation}.
\end{proof}

\begin{lemma}[The follow-up test satisfies the weight assumptions]
  \label{lem:test-two-satisfies-weighted-conditions-appendix}
  Let $Z=\Test_2$.
  Then $Z$ is non-adaptive, $0\le Z\le1$ almost surely, and
  \[
    m_Z(\bm\tau)
    =
    \E[\Test_2\mid\bm\tau]
    =
    p_2(\ATE_2)
    \in(0,1)
    \qquad\text{almost surely.}
  \]
  Consequently, if $\var(\ATE_3)>0$ and the ordinary sign-separation condition \Cref{eq:ordinary-sign-separation-large-budget-appendix} holds, then \Cref{eq:weighted-nondegeneracy,eq:weighted-sign-separation} hold for $Z=\Test_2$.
\end{lemma}

\begin{proof}
  The follow-up test $\Test_2$ depends only on $\bm\tau$ and the follow-up-stage sampling noise $\varepsilon_2$.
  Since $\varepsilon_2$ is independent of the pilot-stage sampling noise $\varepsilon_1$ conditional on $\bm\tau$, $\Test_2$ is non-adaptive.
  The bound $0\le\Test_2\le1$ is immediate.
  Finally, by the definition of the stage-2 pass probability,
  \[
    \E[\Test_2\mid\bm\tau]
    =
    \P(\Test_2=1\mid\bm\tau)
    =
    p_2(\ATE_2),
  \]
  and $p_2(\ATE_2)\in(0,1)$ for every finite $\ATE_2$.
  The final claim follows from \Cref{prop:positive-conditional-weight-appendix}.
\end{proof}

\refstepcounter{proposition}
\begin{linkedresult}{proof}{prop:optimal-solution-exists}[Existence of optimal solution]
  \label{prop:optimal-solution-exists-appendix}
  Suppose $\E[|\ATE_3|]<\infty$ and $\var(\ATE_3)>0$.
  Then the pilot impact maximization problem defined in \Cref{sec:pilot-budget-constraint-and-objective} has at least one optimal solution, i.e.\ $\mathcal S_1^\star(b)\neq\emptyset$ for all $b>0$.
\end{linkedresult}

\begin{proof}
  Apply \Cref{prop:optimal-solution-exists-weighted-appendix} with $Z=\Test_2$.
  By \Cref{lem:test-two-satisfies-weighted-conditions-appendix}, $Z$ is non-adaptive, $m_Z(\bm\tau)=p_2(\ATE_2)>0$ almost surely, and $0\le Z\le1$.
  Since $\var(\ATE_3)>0$, \Cref{eq:weighted-nondegeneracy} holds by \Cref{prop:positive-conditional-weight-appendix}.
  Finally, $\E[|Z \cdot \ATE_3|]\le\E[|\ATE_3|]<\infty$.
\end{proof}

The next result is a slightly more general version of the large-budget theorem for the pilot impact.
That the density condition used in the body theorem implies the conditions here follows from \Cref{prop:density-large-budget-regularity-appendix} below.

\refstepcounter{theorem}
\begin{linkedresult}{extended}{thm:large-budget-limiting-optimal-policy}
  \label{thm:large-budget-limiting-optimal-policy-appendix}
  Suppose $\E[\ATE_3^2]<\infty$, $\var(\ATE_3)>0$, and the uniform boundary anti-concentration and ordinary sign-separation conditions \Cref{eq:uniform-boundary-anti-concentration-large-budget-appendix,eq:ordinary-sign-separation-large-budget-appendix} hold.
  Then:
  \begin{enumerate}
    \item[\rm (i)] $V^\star(b)\le V_{\oracle}$ for all $b>0$.
    \item[\rm (ii)] $V^\star(b)\to V_{\oracle}$ as $b\to\infty$.
    \item[\rm (iii)] The optimal investment diverges,
      \[
        \inf_{\bm s_1^\star\in\mathcal S_1^\star(b)} \bm c^\top\bm s_1^\star
        \to\infty
        \qquad\text{as }b\to\infty.
      \]
    \item[\rm (iv)] The set of unit-investment normalized optimal designs converges to the unit-investment representative design,
      \[
        \sup_{\bm s_1^\star\in\mathcal S_1^\star(b)}
        \left\|
        \frac{\bm s_1^\star}{\bm c^\top\bm s_1^\star}
        -
        \overline{\bm s}_1^{\,\infty}
        \right\|
        \to0
        \qquad\text{as }b\to\infty.
      \]
  \end{enumerate}
\end{linkedresult}

\begin{proof}
  Apply \Cref{thm:large-budget-weighted-target-impact-appendix} with $Z=\Test_2$.
  By \Cref{lem:test-two-satisfies-weighted-conditions-appendix}, $Z$ is non-adaptive, $0\le Z\le1$, and $m_Z(\bm\tau)=p_2(\ATE_2)>0$ almost surely.
  Since $\var(\ATE_3)>0$ and the ordinary sign-separation condition holds, \Cref{prop:positive-conditional-weight-appendix} verifies the weighted non-degeneracy and weighted sign-separation conditions.
  For $Z=\Test_2$, the generic value, oracle, noise-free value, and optimal solution set are exactly $V$, $V_{\oracle}$, $V_{\noisefree}$, and $\mathcal S_1^\star(b)$.
\end{proof}

\begin{proposition}[Density implies large-budget regularity]
  \label{prop:density-large-budget-regularity-appendix}
  Suppose the prior distribution of $\bm\tau$ is absolutely continuous with respect to Lebesgue measure on $\R^D$, with density $f$.
  Assume there exists $r>0$ such that
  \[
    f(x)>0
    \qquad
    \text{for Lebesgue-a.e. }x\in B_r\defeq\{x\in\R^D:\|x\|_2<r\}.
  \]
  Then $\ATE_3$ is non-atomic.
  In particular, if $\E[\ATE_3^2]<\infty$, then $\var(\ATE_3)>0$.
  Moreover, the uniform boundary anti-concentration condition \Cref{eq:uniform-boundary-anti-concentration-large-budget-appendix} and the ordinary sign-separation condition \Cref{eq:ordinary-sign-separation-large-budget-appendix} hold.
  Consequently, under the finite-second-moment assumption in the body theorem, the hypotheses of \linkedCref{thm:large-budget-limiting-optimal-policy-appendix} are satisfied.
\end{proposition}

\begin{proof}
  \emph{Step 1: Non-atomicity and positive variance of $\ATE_3$.}
  The target sample is nonzero, i.e. $\bm s_3\neq\bm0$, and for any $a\in\R$,
  \[
    \left\{\bm{\tau}\in\R^D:\ATE_3=a\right\}
    =
    \left\{\bm{\tau}\in\R^D:\bm s_3^\top \bm{\tau}=a\,\bm1^\top\bm s_3\right\}.
  \]
  is a hyperplane and therefore has Lebesgue measure zero.
  Absolute continuity of the law of $\bm\tau$ implies $\P(\ATE_3=a)=0$ for every $a\in\R$.
  Thus $\ATE_3$ has no atoms.
  If $\E[\ATE_3^2]<\infty$ and $\var(\ATE_3)=0$, then $\ATE_3$ would be almost surely equal to the constant $\E[\ATE_3]$, contradicting non-atomicity.
  Hence $\var(\ATE_3)>0$.

  \emph{Step 2: Uniform boundary anti-concentration condition.}
  For $\overline{\bm s}_1\in\overline{\mathcal{S}}$, define
  \[
    \bm u(\overline{\bm s}_1)
    \defeq
    \frac{\overline{\bm s}_1}{\sqrt{\bm 1^\top \overline{\bm s}_1}},
    \qquad
    \xi(\overline{\bm s}_1,\bm\tau)
    =
    \bm u(\overline{\bm s}_1)^\top \bm\tau.
  \]
  Because $\overline{\mathcal{S}}$ is compact and
  \[
    \bm 1^\top \overline{\bm s}_1
    \ge
    \frac{1}{\|\bm c\|_\infty}
    \qquad
    \forall \overline{\bm s}_1\in\overline{\mathcal{S}},
  \]
  the image
  \[
    K
    \defeq
    \bigl\{\bm u(\overline{\bm s}_1):\overline{\bm s}_1\in\overline{\mathcal{S}}\bigr\}
    \subset \R^D\setminus\{\bm 0\}
  \]
  is compact.

  We claim that
  \[
    \sup_{\overline{\bm s}_1\in\overline{\mathcal{S}}}
    \P\!\bigl(|\xi(\overline{\bm s}_1,\bm\tau)|\le \delta\bigr)
    =
    \sup_{\bm u\in K}\P\!\bigl(|\bm u^\top\bm\tau|\le \delta\bigr)
    \longrightarrow 0
    \qquad\text{as }\delta\downarrow 0.
  \]
  Suppose not.
  Then there exist $\epsilon_0>0$, a sequence $\delta_n\downarrow 0$, and $\bm u_n\in K$ such that
  \[
    \P\!\bigl(|\bm u_n^\top\bm\tau|\le \delta_n\bigr)\ge \epsilon_0
    \qquad
    \forall n.
  \]
  By compactness of $K$, after passing to a subsequence we may assume $\bm u_n\to \bm u\in K$.

  Fix $x\in\R^D$ such that $\bm u^\top x\neq 0$.
  Since $\bm u_n^\top x\to \bm u^\top x$ and $\delta_n\downarrow 0$, we have
  \[
    \indic{|\bm u_n^\top x|\le \delta_n}\longrightarrow 0.
  \]
  The exceptional set $\{x:\bm u^\top x=0\}$ is a hyperplane and therefore has Lebesgue measure zero.
  Since the law of $\bm\tau$ is absolutely continuous, it follows that
  \[
    \indic{|\bm u_n^\top\bm\tau|\le \delta_n}\longrightarrow 0
    \qquad\text{a.s.}
  \]
  Moreover, these indicators are bounded by $1$, so dominated convergence gives
  \[
    \P\!\bigl(|\bm u_n^\top\bm\tau|\le \delta_n\bigr)
    =
    \E\!\left[\indic{|\bm u_n^\top\bm\tau|\le \delta_n}\right]
    \longrightarrow 0,
  \]
  contradicting the lower bound by $\epsilon_0$.
  This proves the uniform boundary anti-concentration condition.

  \emph{Step 3: Sign separation.}
  Fix $\overline{\bm s}_1\in\overline{\mathcal{S}}$ not proportional to $\bm s_3$.
  Let $\bm u\defeq \overline{\bm s}_1$.
  Since $\bm u$ is not proportional to $\bm s_3$, there exists $x_0\in\R^D$ such that
  \[
    \bm u^\top x_0 = 0,
    \qquad
    \bm s_3^\top x_0\neq 0.
  \]
  Replacing $x_0$ by $-x_0$ if necessary, assume $\bm s_3^\top x_0<0$.
  Then for any $\eta>0$,
  \[
    \bm u^\top(x_0+\eta \bm u)=\eta\|\bm u\|_2^2>0,
  \]
  and, for all sufficiently small $\eta>0$,
  \[
    \bm s_3^\top(x_0+\eta \bm u)
    =
    \bm s_3^\top x_0+\eta\,\bm s_3^\top \bm u
    <0.
  \]
  Hence the set
  \[
    A
    \defeq
    \{x\in\R^D:(\overline{\bm s}_1^\top x)(\bm s_3^\top x)<0\}
  \]
  is nonempty.
  It is also open, and it is a cone: if $x\in A$ and $t>0$, then $tx\in A$.
  Therefore $A\cap B_r$ is a nonempty open subset of $B_r$.

  Since $f>0$ for Lebesgue-a.e. $x\in B_r$, we have $f>0$ for Lebesgue-a.e. $x\in A\cap B_r$.
  Because $A\cap B_r$ is nonempty and open, it has strictly positive Lebesgue measure, and thus
  \[
    \P\!\bigl((\overline{\bm s}_1^\top\bm\tau)(\bm s_3^\top\bm\tau)<0\bigr)
    \ge
    \P(\bm\tau\in A\cap B_r)
    =
    \int_{A\cap B_r} f(x)\,dx
    >
    0.
  \]
  This proves sign separation.
\end{proof}